\documentclass[11pt]{article}

\usepackage{titlesec}
\usepackage{tocloft}
\usepackage{multibib}

\titlespacing*{\section}          {0pt}{4pt}{0pt}
\titlespacing*{\subsection}       {0pt}{2pt}{0pt}
\titlespacing*{\subsubsection}    {0pt}{2pt}{0pt}

\usepackage{xr} %% for cross-reference 
\usepackage{mathrsfs}  
\usepackage{dsfont}

\usepackage[pdftex]{graphicx}
\usepackage{amsfonts}
\usepackage{setspace}
\usepackage{amssymb}
\usepackage{amsthm}
\usepackage{graphicx}
\usepackage{amsmath}
\usepackage{mathtools}
\usepackage{lscape}
\usepackage{subcaption}
\usepackage{suffix}
\usepackage{color}
\usepackage{bm}
\usepackage{bbm}
\usepackage{verbatim}
\usepackage{enumerate}

\usepackage{wrapfig}

\usepackage{geometry}
\definecolor{darkblue}{rgb}{0.0,0.0,0.55}
\RequirePackage[colorlinks,citecolor=darkblue,urlcolor=darkblue,linkcolor=darkblue]{hyperref} 

\usepackage{xcolor}
\usepackage{caption}
\usepackage[labelfont={color=darkblue}]{caption}

\usepackage[normalem]{ulem}

\newtheorem{theorem}{\textsc{{Theorem}}}%[section]
\newtheorem{corollary}{\textsc{{Corollary}}}%[theorem]
\newtheorem{lemma}{\textsc{{Lemma}}}

\newtheorem{assumption}{\textsc{{Assumption}}}

\newtheorem{definition}{\textsc{{Definition}}}%[section]
\theoremstyle{remark}
\newtheorem{remark}{{\textsc{Remark}}}

\usepackage{etoolbox}
\makeatletter
\patchcmd{\@begintheorem}{\itshape}{}{}{}
\patchcmd{\@opargbegintheorem}{\itshape}{}{}{}
\makeatother

\usepackage{enumitem}

\usepackage{algorithm, float}% http://ctan.org/pkg/algorithms
\usepackage{algpseudocode}% http://ctan.org/pkg/algorithmicx
\newlength{\continueindent}
\usepackage{booktabs}
\usepackage{array}
\usepackage{multirow}
\newcommand{\bfa}{\mathbf{a}}

\newcommand{\bfc}{\mathbf{c}}
\newcommand{\bfh}{\mathbf{h}}

\newcommand{\bfx}{\mathbf{x}}
\newcommand{\bfy}{\mathbf{y}}
\newcommand{\bfz}{\mathbf{z}}

\newcommand{\bfC}{\mathbf{C}}

\newcommand{\bfH}{\mathbf{H}}

\newcommand{\bfK}{\mathbf{K}}
\newcommand{\bfL}{\mathbf{L}}

\newcommand{\bfX}{\mathbf{X}}
\newcommand{\bfY}{\mathbf{Y}}

\newcommand{\Ell}{\boldsymbol \ell}

\newcommand{\bftheta}{\boldsymbol \theta}

\newcommand{\bflambda}{\boldsymbol \lambda}

\newcommand{\bfmu}{\boldsymbol \mu}

\newcommand{\bfSigma}{\boldsymbol \Sigma}

\newcommand{\bfzero}{\mathbf{0}}

\newcommand{\GP}{\mathcal{GP}}

\newcommand{\diag}{\text{diag}}

\newcommand{\calA}{{\cal A}}
\newcommand{\calB}{{\cal B}}

\newcommand{\calD}{{\cal D}}

\newcommand{\calI}{{\cal I}}

\newcommand{\calK}{{\cal K}}
\newcommand{\calL}{{\cal L}}
\newcommand{\calM}{{\cal M}}

\newcommand{\calO}{{\cal O}}
\newcommand{\calP}{{\cal P}}

\newcommand{\calR}{{\cal R}}

\newcommand{\calT}{{\cal T}}

\newcommand{\scrD}{\mathscr{D}}

\newcommand{{\Pa}}{\text{Pa}}
\newcommand{{\Phild}}{\text{child}}
\newcommand{\lchild}{{\text{l}}}
\newcommand{\rchild}{{\text{r}}}

\usepackage[sort,authoryear]{natbib}
\usepackage{mathrsfs}  

\setcitestyle{citesep={;},aysep={}}

\newtheoremstyle{plain}
  {3pt}              % Space above
  {3pt}              % Space below
  {\itshape}         % Body font
  {}                 % Indentation
  {\bfseries}        % Heading font
  {.}                % Heading punctuation
  {5pt plus 1pt minus 1pt} % Space after heading
  {}                 % Heading specification

\newtheoremstyle{remark}
  {3pt}              % Space above
  {3pt}              % Space below
  {\normalfont}      % Body font
  {}
  {\itshape}         % Heading font
  {.}
  {5pt plus 1pt minus 1pt}
  {}

\renewcommand{\baselinestretch}{1.5}

\newcommand{\down}{\downarrow\!}

\allowdisplaybreaks

\def\spacingset#1{\renewcommand{\baselinestretch}%
{#1}\small\normalsize} \spacingset{1}

\expandafter\def\expandafter\normalsize\expandafter{%
    \normalsize%
    \setlength\abovedisplayskip{2pt}%
    \setlength\belowdisplayskip{4pt}%
    \setlength\abovedisplayshortskip{1pt}%
    \setlength\belowdisplayshortskip{2pt}%
}

\newcommand{\blind}{1}

\date{}
\begin{document}
\if1\blind
{
  \title{\bf A Residual Tree Gaussian Process Modeling Framework for High-Dimensional Data}
  \author{
    Pulong Ma%\footnote{\textit{Address of Correspondence}: Pulong Ma, Assistant Professor, Department of Statistics, Iowa State University;  2438 Osborn Dr, Ames, Iowa 50011.Email: plma@iastate.edu.}
    \\
    Department of Statistics, Iowa State University\\
    %plma@iastate.edu\\
    and\\
    Li Ma \\
    Department of Statistics \& Data Science Institute, University of Chicago\\
    }
      \date{}
  \maketitle

} \fi

\if0\blind
{

    {\LARGE\bf A Residual Tree Gaussian Process Modeling Framework for High-Dimensional Data}

  \medskip
} \fi

\begin{abstract}

With the advance of measurement technologies and increasing computing power, large spatial data with heterogeneous structures are often collected over high-dimensional domains. Existing Gaussian process (GP) models and computational strategies are often inadequate for analyzing such datasets in multi-dimensional domains. To address these challenges, we develop a Bayesian residual tree GP methodology called ResTGP for large spatial data with potentially heterogeneous structures in multi-dimensional domains. The key idea is to decompose a Gaussian process at a cascade of resolutions along a dyadic tree through iteratively computing predictive and residual processes so that the residual process on each tree node, both interior and leaf, becomes sufficient for the finer-level dependency within that node. This allows characterization of the underlying covariance structure in a flexible, multi-scale manner while achieving divide-and-conquer on the data domain, which leads to computational efficiency. To allow efficient tree inference, we introduce a computational strategy for Bayesian inference based on recursive message passing, which scales linearly with the sample size given the tree. This paper also proves posterior consistency of the model for estimating continuous functions in a nonparametric regression framework. Extensive numerical examples and the storm surge application confirm the advantages of the proposed method.
\end{abstract}
\bigskip

%\newpage
%\spacingset{1.45} % DON'T change the spacing!
%\setstretch{1.65} %27 pages per line
\setstretch{1.6} %27 pages per line

%\doublespacing 

\onehalfspacing

\section{Introduction}\label{sec:intro}

Many physical processes such as atmospheric or oceanographic processes inherently possess \emph{multiscale} properties and can exhibit distinctive behavior over a range of spatial resolutions. With the advance in measurement techniques and increasing computing power, large amounts of measurements and large-scale numerical simulations from complex computer models become available to address various uncertainty quantification (UQ) applications in various disciplines including, but not limited to, remote sensing science, ocean science, environmental science, and medical science. Uncertainty Quantification has been widely used to understand and characterize such physical processes using  ground- or space-based observations, and/or simulations of complex computer models (or simulators). 

Among the most popular tools for inferring a physical process of interest from correlated data are the Gaussian process (GP), which has been widely used in spatial statistics \citep{Cressie1993, Banerjee2014, Stein1999},  computer experiments and uncertainty quantification (UQ) \citep{Santner2018}, and machine learning \citep{Rasmussen2006}.

There is a dense literature on GP modeling for addressing the well-known big ``$n$'' problem \citep[e.g.,][]{Banerjee2008, Cressie2008, Sang2012,Nychka2015,Guhaniyogi2018, Ma2020FGP, Datta2016,Katzfuss2017MRA, Katzfuss2021,Stein2004,Gramacy2015, Park2018, Peruzzi2022, Tzeng2018} and spatial nonstationarity \citep{Kim2005, Konomi2019, Luo2024, Paciorek2006,Gramacy2008}. However, there is still a fundamental gap in GP modeling for addressing these issues for higher-dimensional domain. In particular, many spatial GP methods \citep[e.g.,][]{Cressie2008, Sang2012,Nychka2015,Guhaniyogi2018, Peruzzi2022} are not designed for higher-dimensional domain beyond geostatistical applications. A promising class of  GP approximation methods that could be extended to higher-dimensional domain is Vecchia approximations \citep{Vecchia1988, Katzfuss2017MRA, Katzfuss2021, Datta2016, Gramacy2015, Stein2004}, however Vecchia approximations do not provide nonstationarity modeling but rather rely on a pre-specified nonstationary covariance functions that are often  challenging to construct in multi-dimensional domains with dimension beyond two.  Hybrid ideas such as deep GP \citep{Sauer2023} can model nonstationarity for large datasets, but posterior sampling over latent space for higher-dimensional problem is challenging. On the other hand, spatial nonstationarity modeling often requires on highly parameterized covariance function structures \citep[e.g.,][]{Paciorek2006}, but they are not suitable for higher-dimensional domain even with fast GP approximation methods. Other spatial nonstationarity methods \citep{Luo2024, Gramacy2008} can allow spatial nonstationarity, but they could not cope with large datasets in a multi-dimensional domain.

This paper is mainly focused on addressing the gap for GP modeling with large spatial data and heterogeneous structures in a multi-dimensional domain with higher dimension than that in geostatistical applications. In what follows, this setting with large dataset and higher-dimensional domain will be referred to a high-dimensional setting in the sense of large ``$n$'' and large ``$d$'', where $d$ denotes the dimension of the domain. To address this challenge, the paper develops a \textit{residual tree} GP modeling framework that exploits a multi-scale decomposition of a GP into a cascade of residual processes to characterize the underlying covariance structures at different resolutions to allow adaptive shrinkage through random recursive partitioning of the domain so that heterogeneous structures can be inferred at different resolutions with scalable computation for large datasets in multi-dimensional domains. 

The proposed methodology is closely related to the current literature on tree-based modeling and multi-resolution modeling of GPs. For treed GP models, a tree-generating process can be used to allow probabilistic partitioning of the input space. Given the tree structure, many different kinds of GP modeling techniques can be used to model dependence structures within the subregion. The typical example is given by \cite{Gramacy2008}, who develop a Bayesian treed Gaussian process (BTGP) methodology in which stationary Gaussian process models are fitted separately in each leaf (or terminal node) of the tree, by extending the idea of the tree partitioning model \citep{Chipman1998, Chipman2002}. This approach allows nonstationarity modeling via leaf-specific regression parameters and covariance parameters through a tree-generating process. This model in essence enforces a global blockwise independence assumption over the underlying covariance, where the blocks are the leaves of an underlying partition tree that are inferred under a Bayesian paradigm. Hence dependencies are modeled only at the leaf level, which is a major limitation we aim to resolve in this work by modeling dependencies at all levels of a dyadic partition tree. BTGP does enjoy computational ease as it divides the data into packets supported on the leaf nodes, which are sufficient for the subprocesses within the nodes due to the block independence constraint. We aim to maintain this computational advantage in our framework. 

Building upon the BTGP, \cite{Konomi2014} consider a class of nonstationary covariance functions developed by \cite{Paciorek2006} to form a globally valid covariance function across all the terminal nodes, and further alleviate computational challenges that can arise when there are a large number of spatial observations in each terminal node via the full-scale approximation \citep{Sang2012}. This class of treed Gaussian process models relies on a reversible-jump Markov chain Monte Carlo (RJ-MCMC) algorithm to enable stochastic tree search operations (grow, prune, change, swap, and rotation), which can be computationally intensive and can have slow mixing. %without a good proposal of the tree. 

In a different vein, multiresolution GP models generally relax the restrictive global blockwise independence assumption made under treed GP models to capture complex dependency at different ranges, but in comparison to treed GP models, existing approaches generally suffer two limitations. The first is the loss of the ability to adaptively learn the underlying multiscale representation either in terms of a tree or otherwise. The second is that they no longer divide and conquer the data by splitting them into small subsets of sufficient statistics. 

%These methods often represent a spatial process of interest via a linear combination of multiresolution basis functions and their associated random coefficients, with low-resolutions capturing spatial variability at larger distances and high-resolutions capturing spatial variability at smaller distances. 

Past work on multiresolution modeling falls into two categories. The first category allows multiresolution modeling without using partition tree of the domain \citep{Nychka2002, Cressie2008, Ma2020FGP,Nychka2015}. By assuming a pre-specified number of maximum resolutions, these methods estimate different structures of random coefficients from data to learn nonstationary features in spatial processes. 

The second category again appeals to a dyadic partition tree of the domain to allow GP models to capture different variabilities at different resolutions. For instance, \cite{Fox2012} propose a multiresolution GP over one-dimensional input domain that aggregates over a collection of stationary smooth GPs, each of which is defined on a node of a binary partition tree so that long- and short-range dependence and nonstationarity are captured. However, the resulting model does not allow the resulting likelihood to factorize according to the tree structure, and so the proposed Bayesian inference needs to proceed by sampling {\em all} model parameters over the entire tree jointly using the entire dataset, which is computationally very challenging even for a moderate number of data points. In contrast, scalable multiresolution GP methods such as \citet{Katzfuss2017MRA} and \citet{Guhaniyogi2020} can allow scalable computation but they rely on a fixed partition tree that is given {\em a priori}. %Similarly, the sparse generalized Vecchia framework \citep[SGV;][]{Katzfuss2021} does not allow uncertainty quantification of sparsity patterns.  

Compared to the Bayesian treed GP models, these Bayesian and non-Bayesian multiresolution GP models either suffer from computational burden \citep[e.g.,][]{Fox2012} or ignore uncertainties in the underlying partition tree \citep{Katzfuss2017MRA, Guhaniyogi2018}. Both the multiresolution GP \citep{Fox2012} and the MRA \citep{Katzfuss2017MRA} in essence relax the global blockwise independence assumption imposed by the treed GP \citep{Gramacy2008} to a more flexible, resolution-specific blockwise independence assumption defined along a dyadic tree. When this tree structure is ``correctly'' specified, they can be much more effective than the treed GP in capturing long- and short-range dependencies. For instance, MRA does not allow the learning of the underlying tree partition. In complex data problems, such oracular specification of the tree may be unattainable, and thus treed GP may still perform better in many scenarios.

Recognizing the unique benefits and drawbacks of these tree-based and multiresolution GP methods, we aim to introduce a new model that integrates key features of these approaches and thereby maintains their benefits while avoiding their shortcomings. Specifically, we propose a new class of tree-based multiresolution GP model which incorporates both the stochastic learning of the tree {\em and} avoids the restrictive global blockwise independence assumption. Interestingly, we show that it is possible to achieve efficient adaptive tree learning in more flexible multiresolution GP modeling under an effective decomposition of the resulting likelihood over tree nodes under the weaker resolution-specific blockwise independence assumptions made under multiresolution GP and MRA.

The core development of our framework relies on iteratively constructing \emph{residual processes} along the tree nodes to achieve a likelihood factorization over the tree nodes. In particular, we partition the observations on a tree structure according to a random partition tree. Once the tree is specified, in a top-down (i.e., root-to-leaf) fashion, we construct a \textit{predictive process} and corresponding residual process on each node. The residual process is the ``sufficient statistic'' for all finer dependency structures within that node. This procedure continues recursively with the ``data'' on each node corresponding to the residual process passed down from its parent, and a predictive process can thus be computed accordingly. 

By adopting a prior distribution on the underlying partition tree topology, Bayesian inference can be accomplished through sampling algorithms similar to that for the treed GP. The posterior will place small probability on trees that further split a node should the residuals be well characterized by a \emph{baseline} covariance model; otherwise, the posterior will favor those trees which further split the node such that the residuals therein are modeled with one additional level of resolution up to a maximum number of tree depth being reached. The resulting GP model is termed as a \emph{residual tree GP} (ResTGP) model, which allows Bayesian learning on the underlying partition tree and can capture large-scale dependency over large distances at lower resolutions and small-scale dependency over short distances at deeper resolutions as in multiresolution GP methods. Under a common per-resolution blockwise independence assumption similar to those made under both the MRA and the multiresolution GP (not the much stronger global blockwise independence imposed by treed GPs), the likelihood under our model given the tree can be factorized over the tree nodes due to the sufficiency of the residual process on each node; thereby, efficient Bayesian computation at cost linear in the sample size is achieved. This paper also establishes posterior consistency result to provide strong support of the proposed model for estimating continuous functions. 

The rest of this article is organized as follows. 
Section~\ref{sec: ResTGP} introduces the residual tree GP modeling framework. Section~\ref{sec: Posterior inference} presents the detailed message-passing algorithms for posterior inference. Section~\ref{sec: theory} presents posterior consistency of the ResTGP model for estimating continuous functions.  Section~\ref{sec: simulation} illustrates the proposed methodology with simulation examples. Section~\ref{sec: app} demonstrates the proposed methodology through a real-world application to storm surge.  Section~\ref{sec: discussion} concludes the paper with discussions and possible research directions. Supplementary Materials contain all the supporting details of the main paper.

\section{A Residual Tree GP Modeling Framework} \label{sec: ResTGP}

Let $f(\cdot)=:\{f(\bfx) \in \mathbb{R}: \bfx \in\scrD \subset \mathbb{R}^d \}$ be a real-valued Gaussian process (GP) on the input space $\scrD$ with mean function $\mu(\cdot): \scrD \to \mathbb{R}$ and covariance function $C: \scrD\times \scrD\to \mathbb{R}$, where $d$ is the dimension of the input space $\scrD$.  Suppose that we have a set of inputs $\calA^0:=\{\bfx_i\}_{i=1}^n$ and corresponding outputs $\{y(\bfx_i)\}_{i=1}^n$. Standard spatial modeling for data process $y(\cdot)$ would consider a noisy version of $f$ with independent nugget effect or measurement error incorporated: $y(\bfx) = f(\bfx) + \epsilon(\bfx), \, \bfx \in \scrD$, where $\text{cov}(\epsilon(\bfx_1), \epsilon(\bfx_2)) = v \mathbbm{1}(\bfx_1=\bfx_2)$ with $v>0$. Define $K(\bfx, \bfx'):=C(\bfx, \bfx') + v \mathds{1}(\bfx=\bfx')$.

\subsection{Overview of the model} \label{subsec: model description}

%We consider the additive structure: $y(\bfx) = f(\bfx) + \epsilon(\bfx)$, where $\epsilon(\cdot)$ is independent of $f(\cdot)$ with $\text{cov}(\epsilon(\bfx_1), \epsilon(\bfx_2)) = v \mathbbm{1}(\bfx_1=\bfx_2)$ with $v>0$. This error term $\epsilon(\cdot)$ is often accounting for nugget effect or measurement error. $f(\cdot)$ is a spatial process that will be specified later. Write $K(\bfx, \bfx'):=C(\bfx, \bfx') + v \mathds{1}(\bfx=\bfx')$. %Then $y(\cdot)\sim \mathcal{GP}(\mu, K)$ with $K(\bfx, \bfx'):=C(\bfx, \bfx') + v \mathds{1}(\bfx=\bfx')$. 

We specify the proposed model for $y(\cdot)$ hierarchically based upon a base GP prior $\GP(\mu, K)$ and a stochastic tree prior over the tree space $\mathbb{T}$. The stochastic tree $\calT\in \mathbb{T}$ is constructed using a random recursive partitioning strategy as in \cite{Wong2010}, which is a variant of the marginal prior on the partition tree in Bayesian CART prior \citep{Chipman1998, Chipman2002}. Let $\pi(\calT)$ denote the tree prior for $\calT$.  
For any dyadic tree that arises under this model, its nodes can be grouped into two collections: a set of internal nodes $\calI$ (that do have children) and a set of leaf nodes $\calL$ (that have no children). Let $\downarrow\!\calA$ denote the collection of ancestral nodes from root to the parent of $\calA$. For notational convenience, $\calA$ is used to denote a set of observations (or data points) or the subregion that is used for recursive partitioning whenever there is no scope of confusion.

A key device for the proposed model construction is a recursively constructed ``\textit{residual process}" at any node $\calA$ of the partition tree, which is formally defined below. 
\begin{definition}[Residual process]
Let $\xi_{\calA}(\cdot)$ denote a Gaussian process at node $\calA$ with covariance function $C_\calA(\bfx, \bfx'):=\text{cov}\{\xi_\calA(\bfx), \xi_\calA(\bfx')\}$.
Fix a pre-specified knot set $\calR_\calA$ from the region $\calA$. Define $f_{\calA}(\cdot) := E[\xi_{\calA}(\cdot) | \xi_{\calA}(\calR_{\calA})]$. At the root, define $\xi_\calA=y$; at internal node $\calA$, let $B$ be a child of $\calA$, define $\xi_{\calB}(\cdot):=\{\xi_{\calA}(\bfx) - f_{\calA}(\bfx): \bfx \in \calB\}$. Let $\calA_i,\calA_j$ be two different child nodes of $\calA$, assume $\text{cov}(\xi_{\calA_i}(\bfx), \xi_{\calA_j}(\bfx'))=0$. Then We call $\xi_{\calB}(\cdot)$ a residual process at node $\calB$. 
\end{definition}

\begin{remark}
 This residual process construction starts at root where we assume $\xi_\calA(\cdot) = y(\cdot)$ and then repeats recursively at each node down the tree until a maximum tree depth is reached. The covariance function of $\xi_{\calB}(\cdot)$ is given by 
 $C_{\calB}(\bfx, \bfx')=\{C_{\calA}(\bfx, \bfx') - C_{\calA}^*(\bfx, \bfx')\} \mathds{1}(\bfx, \bfx'\in \calB)$, where $C_\calA^*(\bfx, \bfx'):=C_\calA(\bfx, \calR_\calA) C_\calA^{-1}(\calR_\calA, \calR_\calA) C_\calA(\calR_\calA, \bfx')$ is the covariance function for $f_\calA(\cdot)$ \citep{Banerjee2008}. 
 The resulting residual covariance is assumed to be blockwise independent given knots which is different from the \textit{full-scale} residual process $\xi_\calA(\cdot)-f_\calA(\cdot)$.
\end{remark}

From a hierarchical perspective, the proposed model for $y(\cdot)$ can be specified hierarchically. Given the tree $\calT \sim \pi(\calT)$, we can write $y(\cdot)$ as follows:  
\begin{equation} \label{eqn: generative model} 
\begin{gathered}
        y(\cdot) :=   \sum_{\calA \in \calI} f_{\calA}(\cdot) + \sum_{\calA \in \calL} \delta_{\calA}(\cdot), \\
        f_{\calA}(\cdot):=E[\xi_{\calA}(\cdot) | \xi_{\calA}(\calR_\calA)], \quad \calA \in \calI  \\
        \delta_{\calA}(\cdot)  \sim \calM_0(\calA), \quad \calA \in \calL
\end{gathered}
\end{equation}
where $\calR_\calA$ is a set of knots at internal node $\calA$. The term $f_{\calA}(\cdot)$ at internal node $\calA$ is the best linear unbiased predictor via kriging \citep[][Section 1.5]{Stein1999} based on the recursively defined residual process $\xi_\calA(\cdot)$. This term is designed to capture spatial variabilities at large distances across different scales on the partition tree. When the knot set is small, it leads to \emph{low-rank processes} \citep{Banerjee2008, Cressie2008}. Define $h_{\downarrow\!\calA}(\bfx) := \sum_{ \calB \in {\downarrow\!\calA}} f_\calB(\bfx)$. Then standard calculation verifies that $h_{\downarrow\!\calA}(\cdot)$ is a multi-resolution approximation \citep{Katzfuss2017MRA} to the process $y(\cdot)$. At the root, standard kriging formula \citep[][Section 1.5]{Stein1999} gives that $f_\calA(\bfx) = \mu(\bfx) + K(\bfx, \calR_\calA) K^{-1}(\calR_\calA,\calR_\calA)[y(\calR_\calA) - \mu(\calR_\calA)]$.

To complete the construction over the tree, we will specify a \textit{baseline model} $\calM_0$ for the process $\delta_\calA(\cdot)$ at the leaf node.  $\delta_\calA$ at leaf node is assumed to be independent among siblings and is designed to capture spatial variabilities at small distances to accommodate situations either when spatial variations at its finer scales are of no interest or when the data contain weak information at this scale so that no spatial structure can be inferred. As we have a set of residual processes constructed down the tree, the corresponding residuals are also available at each node down the tree when the covariance parameters and knots are specified.  A key distinction between the proposed model and MRA is that the proposed model uses these residuals to infer the leaf node model $\calM_0$ to allow adaptive inference on the tree structure; while MRA maintains all residual covariance structures (i.e., those from $\xi_{\calA}$) at the leaf node and simultaneously relies on a sufficiently deep tree structure such that the number of observations in leaf node is small enough to allow efficient computation, thus MRA does not allow adaptive partition structure by design. The detailed specification of leaf node model $\calM_0$ is  deferred to Section~\ref{subsec: baseline model}.

As residual processes are repeatedly constructed  from a base Gaussian process on a tree structure, we call the resulting model a \textit{residual tree GP} (or ResTGP). The resulting prior for $y(\cdot)$ is a valid stochastic process and we will use the notation $\text{ResTGP}(\bftheta, \calT)$ to denote the residual tree GP prior constructed from a tree $\calT$ and a base GP with covariance parameters $\bftheta$. This ResTGP prior can be used for modeling the response directly or as a latent process; see Section~\ref{sec: theory} for this adoption to study its posterior property. 

%Tackling the computational challenges for large geostatistical data has received lots of attentions in the literature with several successful developments focusing on approximating a model with a pre-specified covariance function \citep[e.g.,][]{Datta2016, Katzfuss2017MRA, Katzfuss2021, Peruzzi2020}. They yield a good approximation of the sampling model under the given covariance function. A key problem that is not adequately addressed in these works is how these methods can be used to infer unknown covariance in a multi-dimensional domain  beyond geostatistical settings. 

Our model construction is an iterative approach to model building and computation by recursively choosing the set of knots  and then regressing the remaining data on these knots to form the predictive and residual processes. This iterative approach is inspired by a theoretical result on an orthogonal decomposition of GP stated in Theorem~\ref{thm: orthogonal decomposition}. 
The model and the inference algorithm are both iterative---in each step (of the model construction and computation), the resulting residual process is then a sufficient statistic for the next iteration. By choosing knots from existing inputs, the computational and modeling effort is only on the residuals over remaining observation locations that have not been selected as knots.

\begin{theorem} \label{thm: orthogonal decomposition}
Let $\xi_0(\cdot) := y(\cdot)$ be any GP on a domain $\mathscr{D}\subset\mathbb{R}^d$. Fix any integer $k \in [1,n]$. Let $\calR^i$ be a nonempty subset of $\calA^i$ such that $\calA^{i+1}:=\calA^{i} \setminus \calR^{i}$ is recursively defined for $i=0,1, \ldots, k-1$. Given the knots $\calR^i$, let $f_i(\cdot)$ be the predictive process of ${\xi}_i(\cdot)$ and  ${\xi}_{i+1}(\cdot):={\xi}_i(\cdot) - f_i(\cdot)$ be the full-scale residual process. It follows that the process $y(\cdot)$ admits the following decomposition
$
y(\bfx) = f_0(\bfx) + f_1(\bfx) + \cdots + f_{k-1}(\bfx) + \xi_{k}(\bfx), \bfx\in \scrD,
$
where $f_0(\cdot), \ldots, f_{k-1}(\cdot), \xi_{k}(\cdot)$ are mutually independent. 
In addition, the sampling density has the following factorization
$
p(y(\calA^0)) = p(f_0(\calR^0)) \times p(f_1(\calR^1)) \times \cdots \times p(f_{k-1}(\calR^{k-1})) \times p(\xi_{k}(\calA^{k})).
$
\end{theorem}

\begin{remark} This orthogonal decomposition and the Vecchia approximations \citep{Datta2016, Katzfuss2017MRA} both use the device of predictive processes. One important distinction is that the orthogonal decomposition leads to exact computation and inference for GP modeling, while the Vecchia's approximations factor the joint density into the product of a collection of conditional densities, each of which is approximated by a small number of random variables in the conditioning set. While the above decomposition does not bring any computational benefit by itself, it will lead to efficient computation when coupled with tree-based modeling strategies (in particular the within-resolution blockwise independence assumptions) described in this paper. %to be described in Section~\ref{subsec: model description}. 
\end{remark}

\subsection{The baseline model $\calM_0(\calA)$ on leaf nodes}  \label{subsec: baseline model}

Modeling spatial nonstationarity is essential for complex data with heterogeneous structures that vary across different regions. The ResTGP methodology can allow spatial nonstationarity over multi-dimensional input space through both stochastic tree prior and baseline models. The stochastic tree prior can allow the ResTGP model to build for different instances of partition tree from the Bayesian CART prior. Not only does this mitigate the discontinuity in the predictive surface around the boundary of partition regions \citep{Katzfuss2017MRA}, but it allows spatial nonstationarity across different trees as in \citet{Gramacy2008} and \citet{Luo2024}. We shall refer to this type of nonstationarity as \textit{across-tree} spatial nonstationarity.

In addition to across-tree spatial nonstationarity, the ResTGP model also allows \textit{within-tree} nonstationarity through different baseline models over the tree nodes. This is evident because each baseline model only captures behaviors within a small region in a tree node whose size shrinks as the tree grows. With different baseline models, region-specific heterogeneous behaviors can be captured. A key requirement for leaf-node model specification is to allow efficient computation of likelihood given the residuals at each node to allow adaptive smoothing as tree grows. We discuss the following choices for the baseline model that balance the need of flexibility with simplicity.

\textit{\textbf{Choice 1}}: A simple choice is to consider the covariance function induced by the residual covariance. In particular, we can use a modified predictive process \citep{Finley2009} representation of the residual covariance function at leaf node as the baseline model so that the leaf model allows efficient computation based on residuals at that leaf. Let $C_{\calA}$ denote the covariance function at leaf node $\calA$. Given a set of knots $\calR_\calA$, the leaf model for $\delta_\calA(\cdot)$ can be constructed: ${\delta}_\calA(\bfx) :=  C_{\calA}^\top(\calR_\calA, \bfx) C_{\calA}(\calR_\calA, \calR_\calA)^{-1} \xi_{\calA}(\calR_\calA) + \epsilon_*(\bfx)$, where $\epsilon_*(\cdot)$ is a residual error process assumed with $\text{cov}(\epsilon_*(\bfx), \epsilon_*(\bfx')) = \{C_{\calA}(\bfx, \bfx') - C_{\calA}^\top(\calR_\calA, \bfx) C_{\calA}^{-1}(\calR_\calA, \calR_\calA)C_{\calA}(\calR_\calA, \bfx')\}$ $\mathds{1}(\bfx=\bfx')$. This specification preserves the orthogonal structure for residual processes between internal nodes and leaf nodes by construction and also allows efficient computation at leaf node. The resulting likelihood function given the tree is essentially a multi-scale predictive process approximation to the base GP likelihood. This specification will adaptively infer the tree partition via the proposed tree prior. This is particularly useful when the goal is to approximate a base GP. This model is referred to as \texttt{ResTGP-PP} hereafter. 

\textit{\textbf{Choice 2}}: Another simple choice for the baseline covariance model is a white-noise model: $\delta_{\calA}(\bfx) \overset{iid}{\sim} N(\mu_\calA, \kappa^2_\calA), \, \bfx \in \calA$  with conjugate priors $\mu_\calA |\kappa^2_\calA \sim N(\mu_0, \kappa^2_\calA/\kappa_0), \kappa^2_\calA \sim IG(a_0, b_0)$, where $z \sim IG(a_0, b_0)$ has density $p(z) = b_0^{a_0} z^{-a_0-1} \exp(-b_0/z) / \Gamma(a_0), z>0, a_0>0, b_0>0$.  This specification will be referred to as \texttt{ResTGP-WN}, hereafter. This model assumes that the residuals at each node do not contain any spatial variation and are modeled with heteroskedastic errors with tree-node dependent means and variances. This can capture different degrees of heteroskedasticity over different regions to allow regime changing behaviors. This specification also facilitates predictive inference as the underlying covariance structure is estimated from the data. Model comparison under this baseline can be performed in a fully Bayesian way; see Section~\ref{sec: likelihood recursion} for details.  

\textit{\textbf{Choice 3}}: A more flexible choice is to assume a parametric covariance function over the tree nodes whose covariance parameters might be shared across different tree levels but allows certain degrees of nonstationary behaviors as in \cite{Fox2012}. However, this situation would then require efficient computation of the marginal likelihood for the baseline model when the number of observations is large. The computational challenge due to the large number of observations can be addressed by adopting existing fast GP approximation methods such as Vecchia approximations \citep{Katzfuss2021}, however, the calculation of the marginal likelihood could still pose potential challenges, and hence this choice is left for future investigation.

\subsection{The dyadic tree prior $\pi(\calT)$} \label{subsec: RRP}

The Bayesian approach to learning a latent tree structure is to endow the tree with a prior distribution, often in the form of a constructive process that generates the tree through a top down stochastic recursive partitioning. We start by describing one such generative tree model which we will use throughout this paper. 
Without loss of generality, we consider recursive dyadic partitioning. The corresponding tree is binary with each internal or non-terminal node having two child nodes (e.g., left-child node and right-child node); see Figure~\ref{fig: binary tree} for an example of a binary tree with fixed cut and adaptive cut. 

\begin{figure}[htpb]
\begin{center}
\makebox[\textwidth][c]{\includegraphics[width=.85\textwidth, height=0.2\textheight]{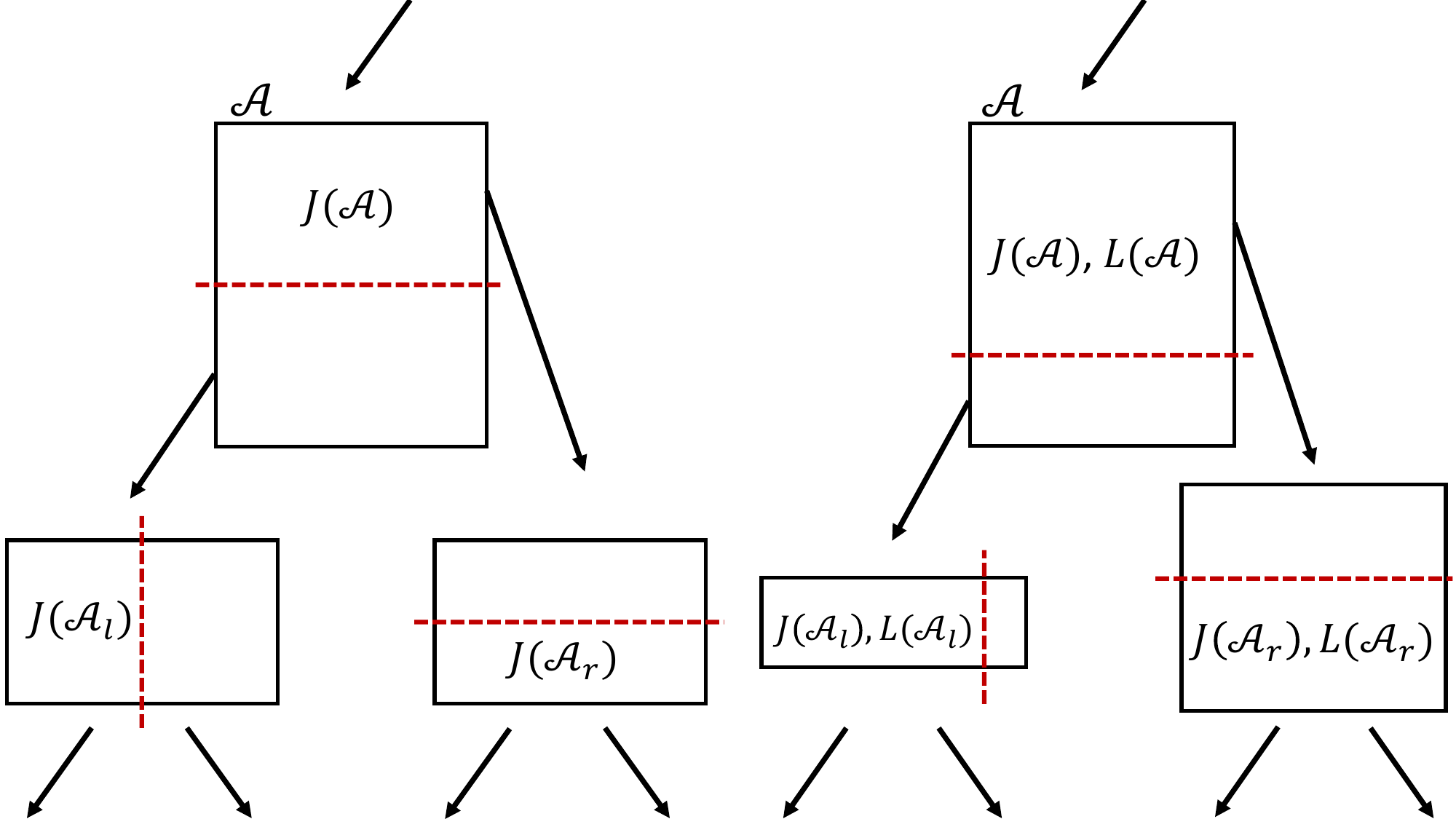}}
\caption{Random recursive partitioning with fixed middle cut (left panel) and adaptive cut (right panel).}
\label{fig: binary tree}
\end{center}
\vspace{-0.6cm}
\end{figure}

For any node $\calA$ that arises during the random partitioning, a random variable $S(\calA)$ with values in $\{0, 1\}$ determines whether the node $\calA$ is further split or not. Let $S(\calA)=1$ correspond to the case that the node $\calA$ is not split with $P(S(\calA)=1) = \rho$, which is the prior stopping probability. That is, $S(\calA) \sim \text{Bernoulli}(\rho)$. We refer to $S(\calA)$ as the \emph{stopping random variable}. $S(\calA)=0$ corresponds to the case that the node $\calA$ is split into its left-child node $\calA_{\lchild}$ and right-child node $\calA_{\rchild}$. Following Bayesian CART prior, we  parameterize $\rho(\calA):=1-\rho_0 (1+ \text{depth}(\calA))^{-\beta_0}$ for $\rho_0\in (0, 1], \beta_0\geq 0$

Let $N_*$ denote the total number of ways to partition the region $\calA$, and let $\bflambda(\calA):=(\lambda_1, \ldots, \lambda_{N_*})^\top$ be the selection probability vector, where $\lambda_j$ denotes the probability of choosing the $j$th way to partition the region $\calA$ such that $\sum_{j=1}^{N_*} \lambda_j(\calA) = 1$ and $\lambda_j \in (0, 1)$.  Let $J(\calA)$ be the selection random variable taken values in $\{1, 2, \ldots, N_*\}$ with $P(J(\calA)=j) = \lambda_j$ for $j=1, \ldots, N_*$.  We say that the region $\calA$ is partitioned according to the partition rule $\calP$ if $J(\calA)=j$ and the region $\calA$ is partitioned into two child nodes according to the $j$th possible way: 
$\calA =  \calA_{\lchild}^j \cup \calA_{\rchild}^j$,    
where the superscript $j$ denotes the $j$th possible way of partitioning $\calA$. A typical choice of $\lambda_j$ is to assume equal prior probabilities among $N_*$ ways.   

A typical example of the partition rule is to consider axis-aligned binary partition (i.e., $N_*=d$), in which $J(\calA)$ becomes a random variable that picks up the dimension to cut. Once $J(\calA)$ is specified, we could assume fixed cutting points such as in the middle; see Figure~\ref{fig: binary tree} for an example. A more flexible way to cut is to endow the cutting point $L$ with a probabilistic distribution $F_L$ on $(0,1)$ so that adaptive partition is achieved. Let $\calA = \prod_{i=1}^d[a_i, b_i]$ and $c\in (0,1)$. Splitting $\calA$ according to axis $J=j$ and cut point $L=c$ gives its child nodes 
\begin{gather*}
\calA_\lchild:=[a_1, b_1] \times \cdots \times [a_j, a_j+c(b_j-a_j)) \times [a_{j+1}, b_{j+1}] \times \cdots \times [a_d, b_d] \\
\calA_\rchild:=[a_1, b_1] \times \cdots \times [a_j+c(b_j-a_j), b_j] \times [a_{j+1}, b_{j+1}] \times \cdots \times [a_d, b_d].
\end{gather*}
$F_L$ can be chosen to be simply as a continuous probability distribution on $(0,1)$, or a discrete distribution as in Bayesian CART prior. When $F_L$ is chosen to be a discrete mixture, e.g., $F_L  = {1}/{M}\sum_{i=1}^M \delta_{c_i}$, where $\bfc:=(c_1,\ldots, c_M)^\top$ is a vector of pre-determined cut points in $(0,1)$ or the empirical quantiles of data points (mapped to $(0,1)$ for consistency) and $M$ is the total number of cut points. $M$ should be chosen as computational budget allows. 

The triple $T:=(S,J,L)$ is the tree decision for partitioning a tree node $\calA$ recursively and defines the probability structure of the partition tree $\calT$ in the same way as the Bayesian CART prior with their hyperparameters denoted by $(\rho_0, \beta_0, \bflambda, \bfc)$.

\subsection{Knot selection} \label{subsec: knots}
Knot selection is a well-known problem in GP modeling \citep[e.g.,][]{Banerjee2008, Finley2009, Ma2019downscaling, Katzfuss2017MRA}. The knot selection in the ResTGP model is similar to the MRA \citep{Katzfuss2017MRA}. The number of knots in each node should be relatively small so that the resulting linear systems can be solved efficiently to construct predictive processes and residual processes at multiple scales. A good starting point is to choose the knots in a similar way to Vecchia approximations \citep{Katzfuss2021, Guinness2018} using methods such as maximin distance design. A decision-theoretic approach to knot selection could be further developed by minimizing certain risk functions such as Kullback-Leibler divergence or mean-squared errors. It should be noted that such an approach may incur computational burdens in searching for optimal knot selection. This investigation is left for future research. 

\section{Posterior Inference} \label{sec: Posterior inference} 
\subsection{Likelihood computation through recursion} \label{sec: likelihood recursion}

It is well-recognized that tree-based GP models \citep{Gramacy2008, Konomi2014, Luo2024} often rely on (reversible-jump) MCMC algorithms to search over the tree space, which can lead to very expensive computation.  A computational advantage of ResTGP is that inference based on the model can avoid MCMC-based stochastic tree search. This can lead to computational efficiency, especially in multi-dimensional domains with dimensions beyond two. 

The ResTGP model allows a closed-form, recursive representation of the marginal likelihood function by analytically integrating out the partition tree prior $\pi(\calT)$. This result is stated in the following theorem, with its proof given in Section S.4.3 of the Supplementary Material. 

\begin{theorem}\label{thm: marginal likelihood}  Suppose that the knots are sequentially selected from existing observations.
Let $q_0(\calA)$ denote the (marginal) likelihood when $\calA$ is a leaf node under either PP or WN leaf model.   Let $\Phi(\calA)$ be the marginal residual likelihood function in node $\calA$ when the partition tree $\calT$ is integrated out, where hyperparameters are suppressed. It can be established that if the node $\calA$ is a leaf node, $\Phi(\calA) = q_0(\calA)$, which is the likelihood of the corresponding residual process on $\calA$; otherwise, the marginal likelihood can be expressed recursively as
\begin{equation}\label{eqn: marginal likelihood}
\begin{split}
\Phi(\calA) = \rho(\calA) q_0(\calA) + (1-\rho(\calA)) \biggr\{ \sum_{j=1}^{d} \lambda_j q_{\calA}(\calR_{\calA}) 
\times \int \Phi(\calA^{j,c}_{\lchild}) \cdot \Phi(\calA^{j,c}_{\rchild}) F_\calA(\mathrm{d} c | j)  \biggr\}.
\end{split}
\end{equation}
\end{theorem}

\begin{remark}
It follows immediately from Theorem~\ref{thm: marginal likelihood} that a bottom-up recursion can be used to compute the marginal likelihood %of observations $y(\calA^0)$ 
under the ResTGP($\bftheta, \calT$) model. It suggests an alternative to MCMC-based computation for the posterior. For a fixed tree with depth $D$,
the computational complexity for one likelihood evaluation with $n$ observations 
is bounded by $\calO(nr^2D^2 + nrdD)$, where $r$ is the maximum number of knots among all the tree nodes. Detailed analysis of computational cost is given in Section S.1.1. The linear scaling in the sample size is notable, but it does not indicate that exact computation of the posterior is always feasible, though often the case when there are less than a handful of dimensions in the domain. In higher dimensional cases, we present later a sequential Monte Carlo algorithm for posterior inference that maintains scalable computation given the tree. 
\end{remark}

\begin{remark} Theorem~\ref{thm: marginal likelihood} also implies that when the data are collected sequentially, it is only necessary to update the likelihood over the nodes where data are added without recomputing the likelihood over all the tree nodes. This has important consequences for online streaming data and sequential design of computer experiments. 
\end{remark}

\begin{remark} For given hyperparameters, let $\rho(\calA \mid \calA^0):=\mathbb{P}(S(\calA)|\bfy(\calA^0))$ be the \textit{posterior stopping probability} of $S$ at node $\calA$ and $\tilde{\lambda}_{j}(\calA \mid \calA^0):=\mathbb{P}(J=j \mid \calA^0)$ be the \textit{posterior selection probability} of $J$ for the $j$th way of partitioning for node $\calA$. Denote the posterior distribution of cutting variable $L$ by $\tilde{F}_{\calA}(\mathrm{d}c|j) = \mathbb{P}(L\in  \mathrm{d}c|J=j,S=0, \calA^0)$. Then it follows from Bayes' theorem and Theorem~\ref{thm: marginal likelihood} that 
\begin{equation} \label{eqn: tree posterior}
\begin{split}
    %\tilde{\rho}(\calA) &= \frac{\rho q_0(\calA)}{\Phi(\calA)}, \\
    \tilde{\lambda}_{j}(\calA) &\propto \lambda_j(\calA) q_{\calA}(\calR_{\calA})\times \int \Phi(\calA_{\lchild}^{j,c}) \Phi(\calA_{\rchild}^{j,c}) F_{\calA}(\mathrm{d} c | j), \quad j=1, \ldots, d, \\
    \tilde{\rho}(\calA) &= \frac{\rho q_0(\calA)}{\Phi(\calA)}, \quad \tilde{F}_{\calA}(\mathrm{d} c \mid j) = \frac{\Phi(\calA_{\lchild}^{j,c} ) \Phi(\calA_{\rchild}^{j,c} )}{\int \Phi(\calA_{\lchild}^{j,u} ) \Phi(\calA_{\rchild}^{j,u} ) F_{\calA}(\mathrm{d} u | j)} \cdot F_{\calA}(\mathrm{d} c | j)
\end{split}
\end{equation}
It should be mentioned that sampling from the posterior distribution of the partition tree using above formula is quite straightforward and computationally cheap as regular Monte Carlo simulation with Bernoulli distribution and categorical distribution can be used. 
\end{remark}

\subsection{Sequential Monte Carlo tree-sampling algorithms}

Suppose that the hyperparameters (i.e., covariance parameters) $\bftheta$ is assigned with a prior $\pi(\bftheta)$. The joint posterior of $\bftheta$ and tree $\calT$ under the ResTGP model is $p(\bftheta, \calT|\bfy) \propto p(\bfy|\calT, \bftheta) \times p(\calT) \times p(\bftheta)$. A key challenge here is to sample the tree structure. Exploiting the success of SMC algorithms for tree-based models \citep{Lakshminarayanan2013, Lu2013, Awaya2024}, we present a one-step look-ahead sequential Monte Carlo (SMC)  algorithm \citep{Lin2013}; see \citet{Liu2001} for an overview. The key idea of SMC method is to design a sequence of \textit{target} distributions to approximate the true distribution by sampling from a sequence of \textit{proposal} distributions, where each sample is called a particle. Each particle in ResTGP model contains a tree structure.  

Let $\calT_t$ denote the tree structure at step $t$ that is formed based on a sequence of tree decisions $\{T_k: k=0,\ldots, t\}$ up to $t$, where $T_t=(S_t, J_t, L_t)$ denotes the tree decision at step $t$. Then the target distribution in SMC is the posterior distribution of $\calT_t$ for which we wish to infer and  by Bayes theorem it is given by $p(T_t|\bfy,\bftheta) \propto \pi(\calT_t) p(\bfy|\calT_t, \bftheta)$. 
Given the current tree $\calT_t$, the next tree is formed by expanding the current tree with the next tree decision $T_{t+1}$. Following  \citet{Lakshminarayanan2013}, the next tree decision can be made by searching over each leaf node $\calA$ of $\calT$ in a breadth-first strategy until tree leaves are explored. This process repeats until the maximum tree level is reached. 

The tree growth process involves sampling from a proposal distribution $q(T_{t}|\calT_{t-1})$ and updating the tree decision $T_{t}$.  The posterior distribution of $\calT_{t}=\calT_{t-1} \cup T_t$  grown from $\calT_{t-1}$ is obtained as follows 
\begin{align} \nonumber 
p(\calT_t | \bfy, \bftheta)& \propto  p(\calT_t) \times p(\bfy|\calT_t, \bftheta) 
 = p(\calT_{t-1}) p(T_t| \calT_{t-1}) p(\bfy|\calT_{t-1}, \bftheta) \times \frac{p(\bfy|\calT_{t}, \bftheta) }{p(\bfy|\calT_{t-1}, \bftheta) } \nonumber \\ \label{eqn: target distribution}
& \propto p(\calT_{t-1}|\bfy, \bftheta) \cdot \text{BF}(\calT_{t}, \calT_{t-1}|\bftheta) \cdot p(T_t|\calT_{t-1}),
\end{align}
where $\text{BF}(\calT_{t}, \calT_{t-1}|\bftheta):=p(\bfy|\calT_{t}, \bftheta) / p(\bfy|\calT_{t-1}, \bftheta)$ is the Bayes factor (when the common covariance parameters are given). $p(T_t|\calT_{t-1})$ is the prior transition distribution of growing the  tree  $\calT_{t-1}$ to produce the tree $\calT_t$ via tree decision rule $T_t:=(S_t, J_t, L_t)$ whose distribution is specified by the Bayesian CART prior. Equation~\eqref{eqn: target distribution} provides a sequential formula to evaluate the current target distribution given the previous step of the target distribution. In order to sample from the target $p(\calT_t | \bfy, \bftheta)$, we use a one-step ahead strategy to construct the proposal distribution $q(T_{t}|\calT_{t-1})$ and then draw from this one-step ahead tree posterior to grow the tree in a node-wise breadth-first manner. The one-step look-ahead proposal distribution is given in Section S.1.

Suppose that there are $N$ particles with initial weights $\{w_{(0)}^i: i=1,\ldots, N\}$ are initialized in the SMC algorithm. At step $t$, the posterior of tree decision $T_t$ is used to generated according to the one-step look-ahead proposal distribution. This is performed independently for each particle and thus leads to a set of $N$ realized trees $\{\calT_t^i:i=1,\ldots, N\}$. The weight for the $i$th particle after the $t$th step (to grow the tree) is updated from  $w_{(t-1)}^i$ to current weight 
\begin{align*}
w_{(t)}^i = w_{(t-1)}^i \cdot \frac{p(\calT_t^i|\calT_{t-1}^i)}{q_{t|t-1}(\calT_t^i|\calT_{t-1}^i, \bfy) } \cdot \text{BF}(\calT_t^i, \calT_{t-1}^i|\bftheta),
\end{align*} 
where $q_{t|t-1}(T_t^i|\calT_{t-1}^i, \bfy):=q(T_t^i|\calT_{t-1}^i)$ is the one-step look-ahead proposal distribution that produces the realized tree $\calT_t^i$ from the tree $\calT_{t-1}^i$.  As $t$ evolves, the distribution of importance weights $w_{(t)}^i$ often becomes highly skewed and results in the so-called weight degeneracy issue \citep{Liu2001, Lin2013}. This can be alleviated by resampling techniques. Most often the effective sample size (ESS) is used to trigger the resampling step \citep{Lin2013}. 

In the SMC algorithm, we can obtain an unbiased estimator $\hat{Z}(\bftheta)$ for the marginal likelihood $Z(\bftheta):=\int p(\bfy| \calT, \bftheta) \mathbb{P}(\mathrm{d}\calT)$, which is also referred as model evidence. This quantity will be used for likelihood comparison in Section 5.1 and SMC diagnostics by monitoring its value across SMC steps. We introduce a novel \textit{path-sharing} implementation of the SMC algorithm with detailed complexity analysis in Algorithm~1 in Section S.1.  The central design strategy is to identify each node's information by its full ancestor path, evaluate the computation once, and let each particle use the result to grow the tree according to the look-ahead strategy in SMC. The resulting computational cost is bounded by  $\calO(nr^2D^2 + nrdD)$ for one likelihood evaluation given the tree and memory cost is bounded by $\calO(NnrD)$ for $N$ particles. This strategy avoids computational bottleneck and memory explosion in a naive implementation where all the computation within each particle is run independently. 

A convenient way to estimate $\bftheta$ is to maximize $\hat{Z}(\bftheta) \pi(\bftheta)$ to obtain the maximum a posterior (MAP) estimate. To account for the uncertainty in $\bftheta$,  
a fully Bayesian inference algorithm called particle marginal Metropolis-Hastings (PMMH) sampler \citep{Andrieu2010} is also given in Algorithm~2 in Section~S.1.

\subsection{Predictive inference}
A common inferential task in spatial data analysis is to make prediction at any input $\bfx_0$ and quantify the uncertainty through the (posterior) predictive distribution of $y(\bfx_0)$ given the data $y(\calA^0)$. A nice property of the ResTGP model is that the joint predictive distribution over multiple prediction locations is valid, and thus provides legitimate predictive uncertainty quantification, for which local methods often fail to offer. 

To derive the predictive distribution $p(y(\bfx_0)\mid y(\calA^0), \bftheta, v)$, we need to first find the predictive distribution $p(y(\bfx_0)\mid y(\calA^0), \bftheta,v, \calT)$ given the covariance parameters $\bftheta$ and the partition tree. Let $\phi_\calT(\bfx_0)$ and $V_\calT(\bfx_0)$ denote the predictive mean and predictive variance of $y(\bfx_0)$ given $\bfy, \calT, \bftheta, v$. Let $\calA_*(\bfx)$ denote a leaf node containing $\bfx$. Denote the ancestor nodes of $\calA_*$ by $\downarrow\!\calA_*(\bfx):=\{\calA \in \calI: \bfx \in \calA\}$. At the root, define $h_{\calA}(\bfx):= \mu(\bfx) + C_{\calA}(\calR_\calA, \bfx)^\top C_{\calA}(\calR_\calA, \calR_\calA)^{-1}\{\xi_{\calA}(\calR_\calA) - \mu(\calR_\calA)\}$; at non-root node, define $h_{\calA}(\bfx):= C_{\calA}(\calR_\calA, \bfx)^\top C_{\calA}(\calR_\calA, \calR_\calA)^{-1}\xi_{\calA}(\calR_\calA)$. Denote by $h_{\calT}(\bfx):= \sum_{\calA \in \downarrow \!\calA_*(\bfx)} h_{\calA}(\bfx)$. Then $y(\bfx_0)$ can be expressed as $y(\bfx_0) =  h_{\calT}(\bfx_0) + \xi_{\calA_*}(\bfx_0)$. 

Let $q(z \mid \calA):=p_{\xi(\bfx_0)}(z \mid \bfy, \bftheta, \calT), z \in \mathbb{R}$ be the \textit{residual predictive distribution} of $\xi(\bfx_0)=z$ given $\bfy$, $\bftheta$ and $\calT$.  If $\calA$ is a leaf node, $\xi_{\calA}(\bfx_0)$ is modeled according to the leaf model $q_0$ with $q(z\mid \calA) = q_0(z \mid \calA)$. For PP, $\xi_{\calA}(\bfx_0) \sim N(\phi_0(\bfx_0; \calA),\, V_0(\bfx_0; \calA))$; for WN,  $\xi_{\calA}(\bfx_0) \sim N(\mu_{\calA}, \kappa^2_\calA), \mu_{\calA} \sim N(\mu_0, \kappa^2_\calA/\kappa_0), \kappa^2_\calA \sim IG(a_0, b_0)$. If $\calA$ is an internal node with immediate child $\calB(\bfx_0)$ that contains $\bfx_0$,  $q(z | \calA) = q(z - h_{\calA}(\bfx_0) \mid \calB(\bfx_0))$ since $ h_{\calA}(\bfx_0)$ is fixed whenever $\calT, \bftheta$ are fixed. Recursively repeating this construction for internal nodes along the ancestor nodes that contain $\bfx_0$ would yield the closed-form expressions for the predictive distribution  $p(y(\bfx_0)\mid y(\calA^0), \bftheta, v, \calT)$.  

\begin{theorem} \label{thm: predictive distribution}
Let $\calT, \bftheta, v$ be fixed. Let $\calA$ denote the terminal node containing $\bfx_0$.  Then 
\begin{itemize}[itemsep=-2pt, topsep=0pt, partopsep=0pt]
\item[(i)] Under PP leaf model, $y(\bfx_0) | \bfy, \calT, \bftheta, v \sim N(\phi_{\calT}(\bfx_0), V_{\calT}(\bfx_0))$ where 
$
\phi_{\calT}(\bfx_0) :=  h_{\calT}(\bfx_0) + h_{\calA}(\bfx_0)$, 
$V_{\calT}(\bfx_0) := C_{\calA}(\bfx_0, \bfx_0) - C_{\calA}(\calR_\calA, \bfx_0)^\top {C}_{\calA}(\calR_\calA, \calR_\calA)^{-1} C_{\calA}(\calR_\calA, \bfx_0) 
$.
\item[(ii)] Under WN leaf model, $y(\bfx_0) | \bfy, \calT, \bftheta, v \sim  t_{2\tilde{a}}(\phi_{\calT}(\bfx_0), V_{\calT}(\bfx_0) (\tilde{a}-1)/\tilde{a})$ provided that $\tilde{a}>1$, where
$
\phi_{\calT}(\bfx_0):=  h_{\calT}(\bfx_0) + \tilde{\kappa}^{-1} \{ \kappa_0 \mu_0 + n_\calA \bar{\xi}(\calA) \},\, 
V_{\calT}(\bfx_0):= \tilde{b} (1 + \tilde{\kappa}^{-1}) / (\tilde{a} -1)
$
with $\tilde{a} := a_0 + n_\calA/2, \tilde{\kappa}:=\kappa_0 + n_\calA, \bar{\xi}(\calA) := {1}/{n_\calA}\sum_{i=1}^{n_\calA} \xi_i(\calA)$, $\tilde{b}:=b_0 + 1/2 \sum_{i=1}^{n_\calA} [ \xi_i(\calA)  - \bar{\xi}(\calA)]^2 + \kappa_0 n_\calA / (2\tilde{\kappa}) [\bar{\xi}(\calA) - \mu_0]^2$.
\end{itemize}
\end{theorem}

The next corollary shows that after marginalizing out the partition tree, the residual predictive mean can still be derived in closed form through a bottom-up recursive formula. 
\begin{corollary} \label{cor: predictive mean}
Let $\phi^*(\bfx_0; \calA):=E_{p(\xi(\bfx_0)\mid \bfy, \bftheta, v)}[\xi(\bfx_0)]$ denote the residual predictive mean of the predictive distribution $p(\xi(\bfx_0)\mid \bfy, \bftheta, v)$ on any potential node $\calA$ that could arise in a partition tree.
If the node $\calA$ is a leaf node, $\phi^*(\bfx_0; \calA)=\phi_0(\bfx_0; \calA)$; otherwise, 
\begin{equation*}
    \begin{split}
        \phi^*(\bfx_0; \calA) = \tilde{\rho}(\calA) \phi_0(\bfx_0; \calA)  &+ (1-\tilde{\rho}(\calA)) \bigr\{h_{\calA}(\bfx_0)
     + \sum_{j=1}^{d}\tilde{\lambda}_j(\calA) \int \phi^*(\bfx_0; \calB^{j,c}) \tilde{F}_\calA(\mathrm{d} c | j) \bigr\}.
    \end{split}
\end{equation*}
where $\calB^{j,c}$ denotes the child node of $\calA$ that contains $\bfx_0$ under cut rule $J=j, L=c$. 
\end{corollary}

\section{Theoretical results}\label{sec: theory}
We establish almost-sure posterior consistency for the white-noise leaf model of the residual-tree GP. All proofs are in Section~S.5 of the Supplementary Material.%\ref{supp: main results}.

\subsection{Notation and preliminaries}\label{sec: notations}
Let $g$ denote the regression function under the ResTGP model on $\scrD=[0,1]^d$: 
$$
g(\bfx) = \sum_{\calA \in \calI} f_{\calA}(\bfx) + \sum_{\calA\in \calL} \delta_{\calA} \mathds{1}(\bfx \in \calA),
$$
where $f_A$ is assumed to be constructed from the process $f$ and $\delta_{\calA} \sim N(\mu_{\calA}, \kappa_{\calA}^2)$ with $\mu_\calA|\kappa^2_\calA\sim N(\mu_0, \kappa^2_\calA/\kappa_0)$. Further assume that $\mu(\bfx)=0$ without loss of generality. Suppose that $f$ has the input-scaled Mat\'ern covariance function with fixed smoothness $\nu>0$: $\calM_{\nu}(\bfh):=\sigma^2 2^{1-\nu} h_0^{\nu}\mathcal{K}_{\nu}(h_0) /\Gamma(\nu)$, where $\calK_{\nu}$ is the modified Bessel function of the second kind of order $\nu$ and $h_0:=\sqrt{\sum_{i=1}^d h_i^2/\ell_i^2}$. Write $\bftheta=(\theta_0,\ldots,\theta_d)=(\log\sigma^2,$ $\log\ell_1,\ldots,\log\ell_d)$ for variance parameter and range parameters in log scale.

Let $\Pi_g$ be the resulting function prior under the assumptions below and $\Pi_v$ the prior for measurement-error variance; write $\Pi_{g,v}=\Pi_g\otimes\Pi_v$.
Write $y_i:=y(\bfx_i)$ and $\mathcal O_n=(O_1,\ldots,O_n)$, where $O_i:=y_i$ in fixed design and $O_i:=(\bfx_i,y_i)$ in random design. Single subscripts refer to the individual experiment; fixed locations/inputs may change with $n$. For a possible observation vector $o$, with $\varphi_v$ denoting $N(0,v)$, the sampling density is
$
p_{g,v}^{(n)}(o)=\prod_{i=1}^n\varphi_v(y_i-g(\bfx_i)).
$
It is relative to Lebesgue measure in fixed design and $(\mathbb P_X\otimes \mathrm{d}y)^{\otimes n}$ in random design, where $\mathbb P_X$ is the probability measure on input. For Borel sets $A$ in the space of $(g, v)$, the posterior is
\begin{equation}\label{eqn: function-posterior}
\Pi_{g,v}(A\mid\mathcal O_n)
=\frac{\int_A p_{g,v}^{(n)}(\mathcal O_n)\,\Pi_{g,v}(\mathrm{d}g,\mathrm{d}v)}
{\int p_{g,v}^{(n)}(\mathcal O_n)\,\Pi_{g,v}(\mathrm{d}g,\mathrm{d}v)}.
\end{equation}
Its marginals are $\Pi_g(\,\cdot\mid\mathcal O_n)$ and $\Pi_v(\,\cdot\mid\mathcal O_n)$ \citep[Section 1.3]{Ghosal2017}. 

Put $\mathbb P_{X,n}=n^{-1}\sum_i\delta_{x_i}$ in fixed design and $\mathbb P_{X,n}=\mathbb P_X$ in random design. For an input probability measure $\mathbb P$, use $d_\mathbb P(g,g_0)=\inf\{u>0:\mathbb P(|g-g_0|>u)<u\}$ for convergence in measure, and the  Hellinger distance
$
H_\mathbb P^2((g,v),(g_0,v_0))=\int
\{\sqrt{\varphi_v(y-g(\bfx))}-\sqrt{\varphi_{v_0}(y-g_0(\bfx))}\}^{2}\,\mathrm{d}\mathbb P(\bfx) \times \mathrm{d} y.
$

\subsection{General consistency conditions}\label{sec: general-conditions}
Consider a fixed proper prior $\Pi_{g,v}$, jointly measurable positive sampling densities and a proper posterior, with all experiments defined on a common probability space under $\mathbb P_0$. For each $\epsilon>0$, the following conditions are sufficient for almost-sure posterior consistency.
\begin{description}[leftmargin=0.8cm,nosep]
\item[C1.] For every $a>0$,
\begin{equation*}%\label{cond: C1}
\int\frac{p_{g,v}^{(n)}(\mathcal O_n)}{p_{g_0,v_0}^{(n)}(\mathcal O_n)}\,\Pi_{g,v}(\mathrm{d}g,\mathrm{d}v)\geq e^{-an}
\quad \mathbb P_0\text{-almost surely, eventually}.
\end{equation*}
\item[C2.] There are  measurable sieves $\mathcal S_n$ and constants $0<C_\epsilon<\infty$, $c_\epsilon>0$ such that, for all sufficiently large $n$, 
$\Pi_{g,v}(\mathcal S_n^c)\leq C_\epsilon e^{-c_\epsilon n}.
$
\item[C3.] On the same sieves there are deterministic Borel tests $\psi_n$ taking values in $\{0,1\}$ on the observation space such that, for all sufficiently large $n$,
\begin{equation*}%\label{cond: C3}
\mathbb E_0\psi_n(\mathcal O_n)\leq C_\epsilon e^{-c_\epsilon n},\qquad
\sup_{\substack{(g,v)\in\mathcal S_n\\H_{\mathbb P_{X,n}}((g,v),(g_0,v_0))>\epsilon}}
\mathbb E_{g,v}^{(n)}(1-\psi_n)\leq C_\epsilon e^{-c_\epsilon n}.
\end{equation*}
Here $\mathbb E_0, \mathbb E_{g,v}^{(n)}$ denote expectation under the true law and the sampling density $p_{g,v}^{(n)}$, respectively. 
\end{description}
Condition C1 is the posterior denominator bound used in the proof of \citet[][Theorem 6.39]{Ghosal2017}. For a sequence of nested observations, this bound appears in \citet[][Appendix A, Equation (A.5)]{Choi2007}. Here C1 is obtained from positive prior mass and likelihood moments using \citet[Lemma 8.21]{Ghosal2017}. Independence between experiments is not required. Condition C2 is the exponential prior-mass condition applied to the sieve complement; see \citet[][Theorem 6.39 (b)]{Ghosal2017}. Condition C3 is the testing condition which is a stronger version of \citet[][Theorem 6.39(a)]{Ghosal2017}. Conditions C2 and C3 imply the testing condition in Theorem 1 of \citet{Choi2007}.  Under C1--C3, $\Pi_{g,v}(H_{\mathbb P_{X,n}}((g,v),(g_0,v_0))>\epsilon\mid\mathcal O_n)\to0$ $\mathbb P_0$-almost surely for every $\epsilon>0$. This criterion follows the consistency framework of \citet[Theorems 6.17 and 6.39]{Ghosal2017}; its general-experiment proof is in Section~S.5.5. What follows is to state explicit assumptions needed to verify these conditions.  

\subsection{Assumptions}\label{ass: model assumptions}
The following hypotheses specify the truth, design and fixed prior.
\begin{assumption}\label{ass: truth}
Fix $g_0\in C(\scrD)$ and $v_0>0$. Under the true law $\mathbb P_0$, use either:
\begin{enumerate}[label=(\roman*),leftmargin=*,nosep]
\item For each $n$, let $\bfx_1, \ldots, \bfx_n \in \scrD$ be design locations with their dependence on $n$ suppressed. Let $y_i$ denote their random response at $\bfx_i$ on a common probability space, with $
y_i=g_0(\bfx_i)+\sqrt{v_0}\varepsilon_i,
\varepsilon_1,\ldots,\varepsilon_n\overset{iid}{\sim}N(0,1).
$
\item An iid sequence $(\bfx_i,\varepsilon_i)$ with law $\mathbb P_X\otimes N(0,1)$, where $\mathbb P_X$ is a fixed Borel probability measure on $\scrD$, and $y_i=g_0(\bfx_i)+\sqrt{v_0}\varepsilon_i$.
\end{enumerate}
\end{assumption}

\begin{assumption}\label{ass: tree-prior}
The Bayesian CART prior satisfies: (i) $0<\rho_0<1$ and $\beta_0>0$; (ii) $\lambda_j>0$, and $\sum_{j=1}^d\lambda_j=1$; (iii) $\{c_1, \ldots, c_M\} \subset (0,1)$ with one of them being $1/2$.
\end{assumption}

\begin{assumption}\label{ass: knots}
Every internal node has a deterministic ordered knot set $\mathcal R_{\mathcal A}$ in its cell interior, with $1\leq |\mathcal R_{\mathcal A}|\leq r$ for a fixed finite integer $r$. Knots are distinct within a node and do not coincide with ancestor knots. The rule  does not depend on $n$. 
\end{assumption}

\begin{assumption}\label{ass: covariance}
The prior on covariance parameters is supported on
$[\log\sigma_-^2,\log\sigma_+^2]\times\mathbb R^d$, where $0<\sigma_-^2\leq\sigma_+^2<\infty$, and satisfies, for all sufficiently large $n$ and constants $C_\ell, c_\ell$,
\begin{equation}\label{eqn: range-tail}
\Pi_{\bftheta}\!\left(\max_{1\leq j\leq d}|\theta_j|>n\right)\leq C_\ell e^{-c_\ell n},
\qquad 0\leq C_\ell<\infty,\quad c_\ell>0.
\end{equation}
The error variance $v$ has a proper prior $\Pi_v$ supported on $[v_-,v_+]$, where $0<v_-<v_0<v_+<\infty$, with $\Pi_v(|v-v_0|<\delta)>0$ for every $\delta>0$. Assume that $\mu_0\in \mathbb{R}$, $\kappa_0>0$ are fixed constants, and there exists a constant $\kappa_+^2>0$ such that $\kappa^2_\calA<\kappa_+^2$ for all $\calT$ and $\calA \in \calL$. 
\end{assumption}

\begin{remark}
Only continuity of $g_0$ is assumed but prior functions may jump. The knots require no uniform separation or space-filling condition. Covariance parameters may be fixed, and random ranges are allowed under \eqref{eqn: range-tail}. 
\end{remark}

\subsection{Posterior consistency}\label{sec: posterior consistency}
Under these assumptions, Propositions~P4--P5 verify C1, Proposition~P6 verifies C2, and Proposition~P7 verifies C3 on the same sieve. These propositions show new model-specific results. The formal consistency theorem is stated below.
\begin{theorem}[Posterior consistency]\label{thm: consistency}
Under Assumptions~\ref{ass: truth}--\ref{ass: covariance}, with fixed truth $g_0\in C(\scrD)$, for the fixed prior $\Pi_{g,v}$ described above and every $\epsilon>0$,
\[
\Pi_{g,v}\!\left(H_{\mathbb P_{X,n}}((g,v),(g_0,v_0))>\epsilon\mid\mathcal O_n\right)
\longrightarrow0\qquad \mathbb P_0\text{-almost surely}.
\]
\end{theorem}

\begin{remark}
This theorem implies that under regularity conditions, the posterior distributions of the regression function and the variance parameter under the proposed model converge to the truth under the corresponding topologies as the sample size goes to infinity. This ensures that the proposed model behaves well and can provide enough support to predict the class of functions in Assumption~\ref{ass: truth}. 
\end{remark}

\begin{corollary}\label{cor: convergence in measure}
Under the same assumptions in Theorem~\ref{thm: consistency}, for every $\epsilon>0$,
\[
\Pi_{g,v}\!\left(\{d_{\mathbb P_{X,n}}(g,g_0)\geq\epsilon\}
\cup\{|v/v_0-1|\geq\epsilon\}\,\middle|\,\mathcal O_n\right)
\longrightarrow0\qquad \mathbb P_0\text{-almost surely}.
\]
In particular, on the same probability-one event,
\[
\Pi_g(d_{\mathbb P_{X,n}}(g,g_0)\geq\epsilon\mid\mathcal O_n)\to0,
\qquad
\Pi_v(|v/v_0-1|\geq\epsilon\mid\mathcal O_n)\to0.
\]
\end{corollary}

\begin{remark}
This result concerns convergence in measure under the empirical input distribution in fixed design and under $\mathbb P_X$ in random design. This implies that for every fixed error tolerance, the posterior concentrates on functions whose error exceeds that tolerance on a vanishing fraction of design locations in fixed design, or on a set of vanishing $\mathbb{P}_{X}$-probability in random design. 
\end{remark}

Corollary~\ref{cor: convergence in measure} provides function consistency under the model induced prior. The next result in Corollary~\ref{cor: bounded fun} provides a stronger loss guarantee under a modified prior. 
For fixed $M>0$, write $\Pi_{g,v}^M=\Pi_{g,v}(\,\cdot\mid\|g\|_\infty\leq M)$ for the prior conditioned on the uniform envelope. Its posterior is thus given by 
\[
\Pi_{g,v}^{M}(A\mid\mathcal O_n)
=
\frac{
\Pi_{g,v}\!\left(
A\cap\{\|g\|_\infty\leq M\}
\mid\mathcal O_n
\right)
}{
\Pi_g(\|g\|_\infty\leq M\mid\mathcal O_n)
}
\]
for every measurable set $A$.

\begin{corollary}\label{cor: bounded fun}
If Assumptions~\ref{ass: truth}--\ref{ass: covariance} hold and $\|g_0\|_\infty<M$, then, for every $\epsilon>0$,
\[
\Pi_{g,v}^M\!\left(\{\|g-g_0\|_{L_1(\mathbb P_{X,n})}>\epsilon\}
\cup\{|v/v_0-1|\geq\epsilon\}\,\middle|\,\mathcal O_n\right)
\longrightarrow0\qquad \mathbb P_0\text{-almost surely}.
\]
If the random-design distribution $\mathbb{P}_{X}$ has a Lebesgue density $p_X$ that is bounded below away from zero, the same conclusion holds with $L_1(\mathrm{d}\bfx)$ in place of $L_1(\mathbb P_X)$.
\end{corollary}

\begin{remark}
This result implies that under fixed design, the  absolute regression error averaged over observed inputs vanishes in posterior probability almost surely; under random design, this  absolute regression error averaged over the population distribution of inputs vanishes in posterior probability almost surely. 
\end{remark}

\section{Numerical Illustrations} \label{sec: simulation}

\subsection{Likelihood approximation}
Likelihood approximation to known targets has been used as a standard assessment tool in spatial statistics \citep[e.g.,][]{Stein2004, Sang2012, Katzfuss2017MRA, Katzfuss2021, Guinness2018, Datta2016}. In most spatial applications, this has been well studied for $d\le 3$, but not for higher dimensions. Among these methods, Vecchia approximations \citep{Katzfuss2021, Datta2016} are among the most successful methods to achieve this task. While Bayesian methods are not inherently designed for likelihood approximation, it is still worth comparing the proposed ResTGP-PP with existing methods in a non-Bayesian setting. The goal of this subsection is to demonstrate the quality of likelihood approximation for the proposed method ResTGP under higher input dimensions ($d$ ranges from 2 to 20) and compare it with the scaled sparse general Vecchia  \citep[\texttt{sSGV};][]{Katzfuss2022} that extends \citet{Katzfuss2021} to higher input dimensions.

In particular, the underlying truth is simulated from a stationary Gaussian process with an input rescaled Mat\'ern covariance over the domain, where the parameter settings are $\sigma^2=1, \nu=1.5, \ell_j=4^{(j-1)/(d-1)}, j=1,\ldots, d$ with a nugget $v=0.1$ included. The observation locations are generated via maximin Latin hypercube design for different $n$. This example examines the quality of approximation due to several factors including the number of observations ($n$) and input dimension ($d$). The computing time for one likelihood evaluation is recorded for all methods.   ResTGP method has several special implementations: (1) one is to use the full residual covariance in leaf node with a pre-specified tree structure drawn from Bayesian CART prior. This implementation achieves a full-scale approximation on a tree structure that extends MRA to higher dimensions, and will be referred to as \texttt{ResTGP-Full}; (2) the second one is \texttt{ResTGP-PP}.  This model also approximates the original base covariance but differs from \texttt{ResTGP-Full} in that it allows its tree structure to be adaptively inferred based on the data. (3) The third version of ResTGP is \texttt{ResTGP-WN}.  This model is more suitable to capture region changing spatial nonstationarity that will be demonstrated in subsequent sections. The implementation details are given in Section S.2.1. 

\begin{figure}[!htpb]

\begin{subfigure}[t]{1\textwidth}
\makebox[\textwidth][c]{\includegraphics[width=1\textwidth, height=0.2\textheight]{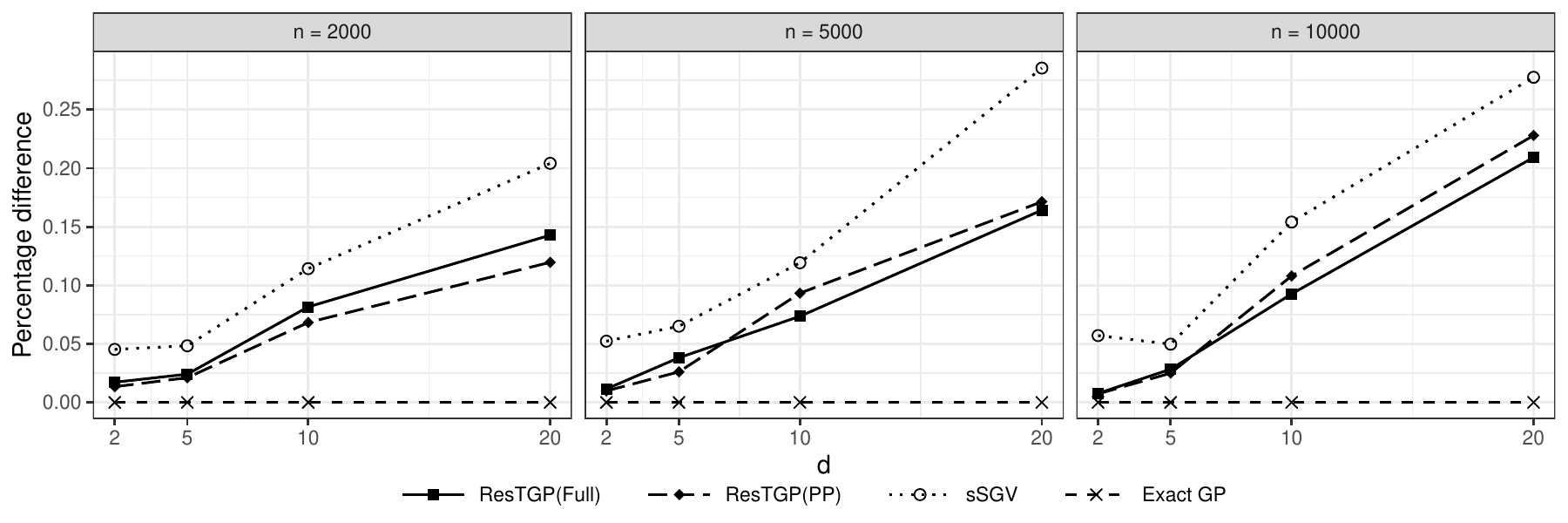}}
\end{subfigure}
\begin{subfigure}[t]{1\textwidth}
\makebox[\textwidth][c]{\includegraphics[width=1\textwidth, height=0.2\textheight]{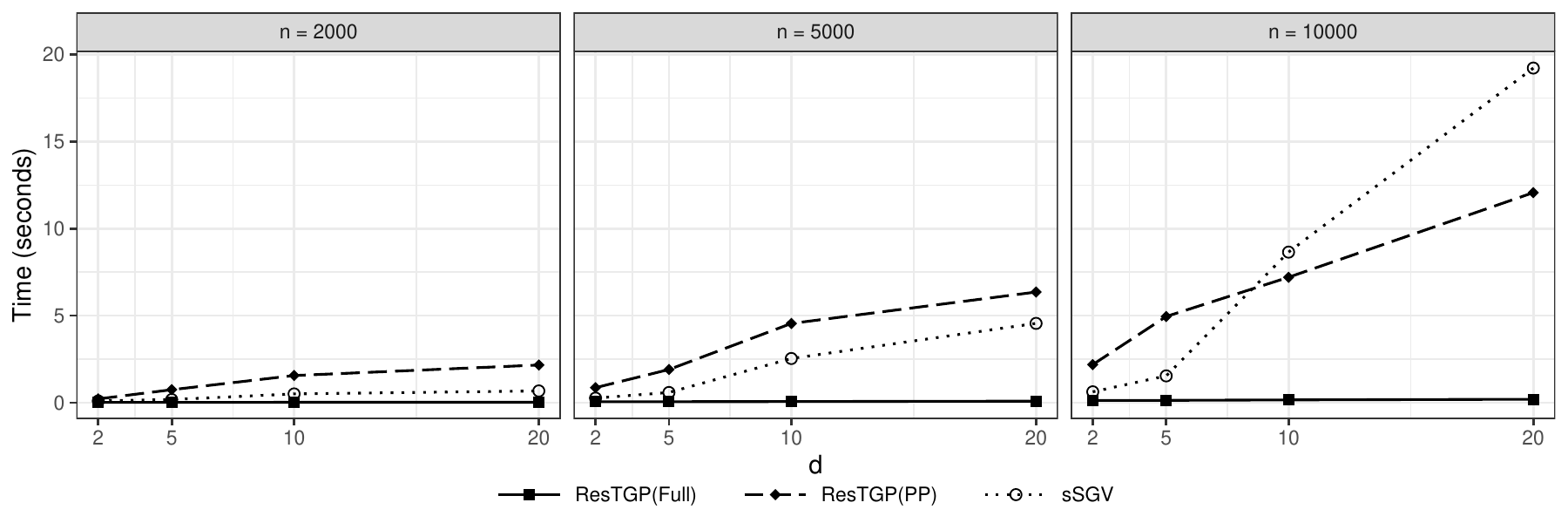}}
\end{subfigure}
\caption{Likelihood approximation versus input dimension $d$ under different $n$.  All methods have fixed $r=30$ knots or nearest neighbors chosen via maximin design. Top panels show the percentage error of log-likelihood compared to exact GP. Bottom panels show the timing.}
\label{fig: likelihood approximation}
%\vspace{-0.6cm}
\end{figure}

The results in Figure~\ref{fig: likelihood approximation} show that all the three models can approximate the true likelihood very well when data is rich in low-dimensional space ($d\leq 5$), measured by $|\ell_{\text{approx}}-\ell_{\text{truth}}|/|\ell_{\text{truth}}|$, where $\ell_{\text{approx}}$ denotes approximate log-likelihood under different methods and $\ell_{\text{truth}}$ denotes the truth log-likelihood. When $d$ is large (e.g., 10 or 20), the performances of all the methods deteriorate with their loglikelihood percentage difference ranging from 5\% to 25\% depending on $n$. However, both \texttt{ResTGP-PP} and \texttt{ResTGP-Full} outperform sSGV with 30 nearest neighbors (or knots).  As shown in Section S.2.1 of the Supplementary Material, when the number of knots or nearest neighbors is increased to $r=100$, the performances under \texttt{ResTGP-PP} and \texttt{ResTGP-Full} improve substantially with their log-likelihood percentage difference dropping from 20\% ($r=30$) to 12\% ($r=100$).  In terms of timing, \texttt{ResTGP-Full} is fastest among all the models due to a pre-determined single tree specification from the Bayesian CART prior even though its performance can be sensitive to tree specification. \texttt{ResTGP-PP} is much slower than the other two because of tree inference via SMC, but parallel computing with 8 cores can substantially reduce computing time when tree is deep (e.g., $d=20$). It is worth noting that different methods have been implemented in different languages such as \textsf{R} or \texttt{C++} so the timing does not necessarily reflect their theoretical computational complexity.

\subsection{Adaptation to input dependent noise} 
In this section, we focus on illustrating the benefit of the ResTGP over treed GP and SGV in terms of capturing input-dependent noise structures as this is a key feature for ResTGP with the white-noise leaf model. 
For $x\in[0,1]$, let $y(x) = f(x) + \sum_{k=1}^{4}\sigma_k\epsilon(x) \mathds{1}_{I_k}(x)$, where 
$f(x)=-3e^{3\min(x,1/4)}\sin\!\bigl(10\pi\min(x,1/4)\bigr)
\mathds{1}_{\{x<1/2\}}
+e^{3x}\sin(40\pi x/3)\mathds{1}_{\{x\ge3/4\}}$. $(\sigma_1,\sigma_2,\sigma_3,\sigma_4)=(0.8,0.2,0.5,1.5)$, $I_1=[0,1/4]$, and $I_k=((k-1)/4,k/4]$ for $k=2,3,4$.
$\epsilon(x)$ is an $N(0,1)$ random variable with $\text{cov}(\epsilon(x), \epsilon(x'))= \mathds{1}(x=x')$. This function has a globally nonstationary but locally stationary feature with different smoothness structures and noises in four different regions.  The data points are simulated over $n=320$ training locations and the predictive performance over 100 held-out new locations that are equally spaced. All models are fitted with an isotropic Mat\'ern covariance function. To measure the performance, continuous ranked probability score (CRPS) and log-score (i.e., averaged negative pointwise log predictive density) are used to measure the accuracy of the entire predictive distribution with root mean squared error (RMSE) and  empirical coverage probability (CVG) of 95\% prediction intervals also reported. 

Figure~\ref{fig: 1d example 1} indicates that all the methods seem to give satisfactory predictive mean with numerical measures summarized in Table~S.1 of the Supplement.  Both \texttt{tgp} and \texttt{ResTGP-WN} can adapt to the input noise as seen from Panel (b) of Figure~\ref{fig: 1d example 1}, however, \texttt{ResTGP-WN} gives the best performance in terms of CRPS and log-score. SGV is over-confident in the region $(0.75, 1]$ and cannot adapt to the heterogeneous noise settings across different regions. This applies for \texttt{ResTGP-PP} because both SGV and ResTGP-PP are designed to approximate a stationary Mat\'ern covariance function. As \texttt{ResTGP-PP} accounts for the uncertainty in tree partitions, \texttt{ResTGP-PP} gives better CRPS and logscore than SGV.

We also investigated the performance for the same simulation setting with $n=3,200$ observations in Section S.2.2. The conclusion also confirms the finding that \texttt{ResTGP-WN} can achieve better performance in terms of CRPS and log-score than other methods including SGV and BART \citep{Chipman2010} when the noise depends on inputs. All timing comparisons are also discussed in Section S.2. 

\begin{figure}[!htpb]

\begin{subfigure}[t]{\textwidth}
\makebox[\textwidth][c]{\includegraphics[width=1\textwidth, height=0.25\textheight]{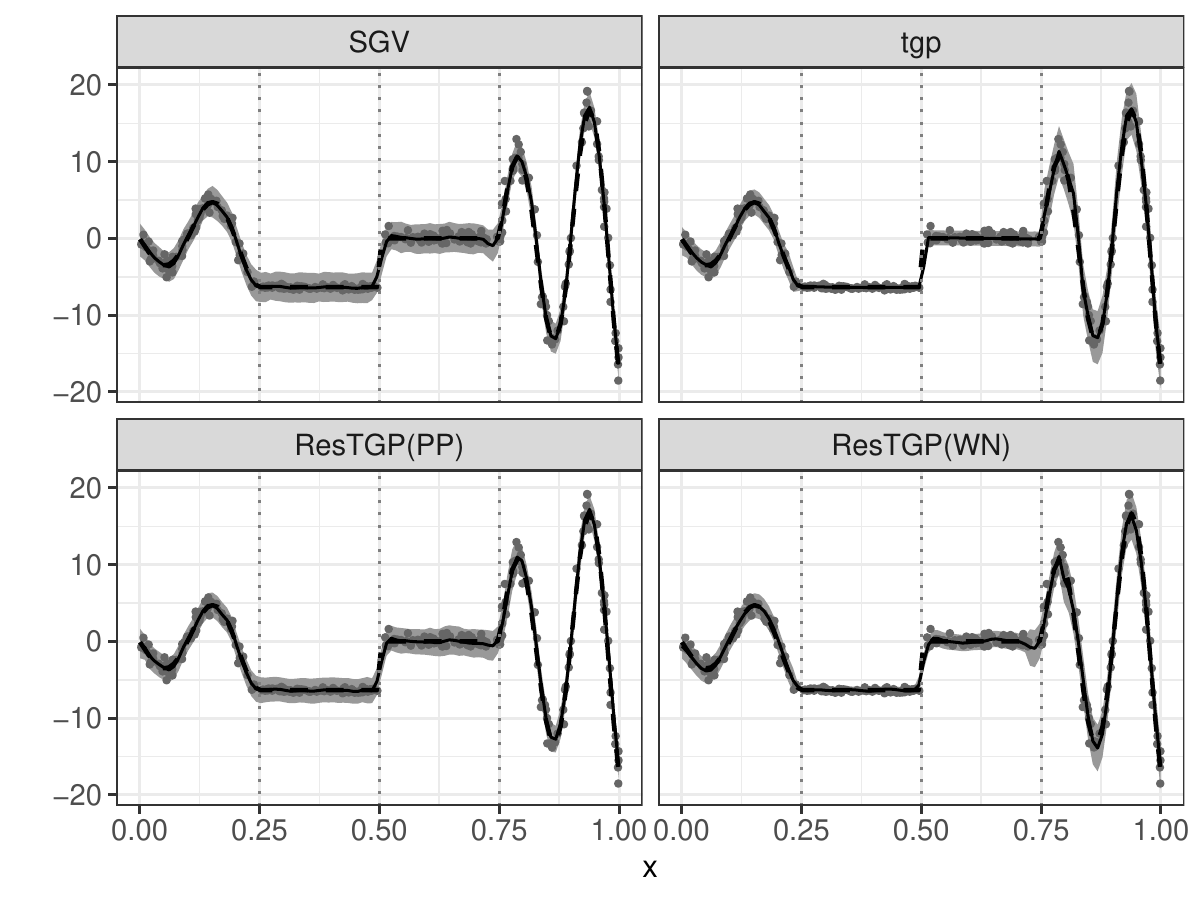}}
\caption{Predictions under different models with 95\% prediction intervals}
\end{subfigure}%

\begin{subfigure}[t]{\textwidth}
\makebox[\textwidth][c]{\includegraphics[width=\textwidth, height=0.18\textheight]{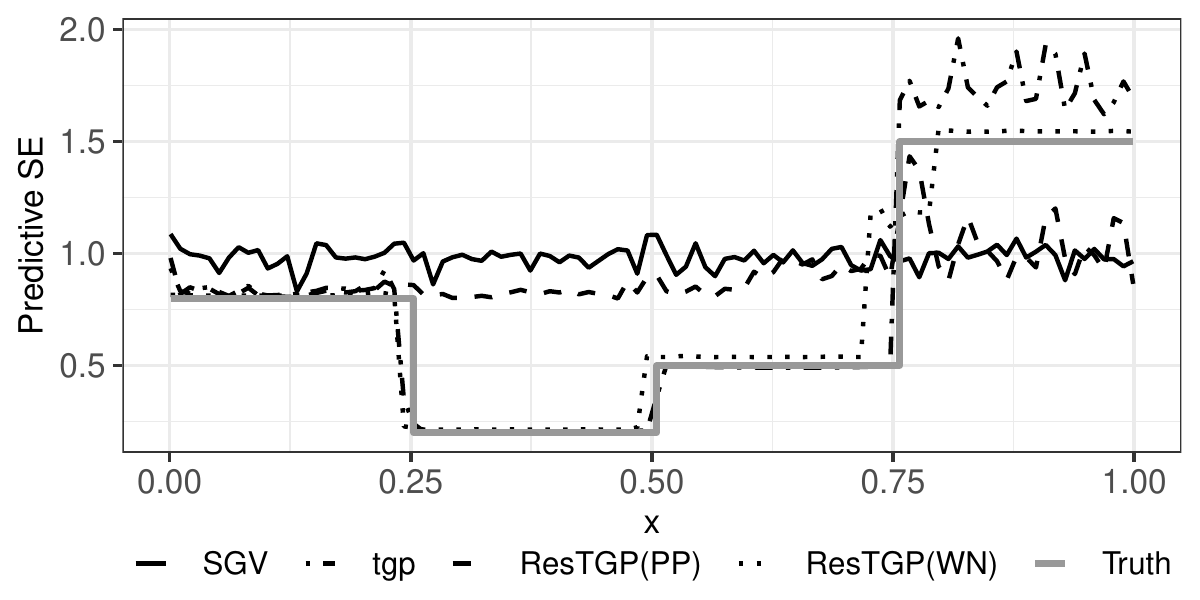}}
\caption{Predictive standard deviations}
\end{subfigure}
\caption{Predictive performance of different models  based on 320 observations (dots). %Panel (a) shows the predictive mean and 95\% predictive intervals. Panel (b) shows predictive standard deviations.
}
\label{fig: 1d example 1}
%\vspace{-0.2cm}
\end{figure}

\subsection{Higher-dimensional examples}

The dataset is simulated from $y(\bfx) = f(\bfx) + \sigma(\bfx) \epsilon$, where $\epsilon\overset{iid}{\sim} N(0,1)$, $\sigma(\bfx)=0.05 + 0.95 g(x_2; 0.5, 0.05)$, $f(\bfx) = 2\{[1-g(x_1; 0.5, 0.02)] \sin(4\pi r(\bfx)) + g(x_1; 0.5, 0.02) \sin(6\pi r(\bfx) +1)\} + 2 g(x_1; 0.5, 0.02) + 0.5\sum_{j=1}^d \cos(\pi x_j)/\sqrt{d}$ with $r(\bfx) := | (x_3, x_4, x_5) - (0.3, 0.6, 0.4)|$ and $g(t; a, b):=\{1 + \exp(-(t-a)/b)\}^{-1}$. The function $f$ has a multiscale structure with smooth behaviors. The noise function $\sigma$ smoothly varies as a function of $x_2$. The $g$ function allows smooth transition at center $a$ with shape $b$. The training dataset is simulated over a maximin Latin hypercube sampling (LHS) of $n=10,000$ input locations in $[0,1]^d$. The testing dataset is simulated over a maximin LHS of $n^*=0.2n$. All the knots or nearest neighbors are chosen based on maximin design with $r=60$ for sSGV and ResTGP models. ResTGP-Full is estimated with maximum likelihood estimation via gradient based optimization while sSGV is estimated via Fisher scoring method and both ResTGP-PP and ResTGP-WN are estimated via the derivative free method using \texttt{Nelder-Mead} in \textsf{R} function \texttt{optim}.

\begin{table}[!htbp]
\centering
\normalsize
\caption{Predictive performance under different models with $n=10{,}000$
and $d=10,20$. %The reported metrics are the root mean squared error (RMSE), continuous ranked probability score (CRPS), empirical coverage of the nominal 95\% predictive intervals (CVG), and log-score.
The ResTGP models are estimated using the empirical Bayes method unless mentioned.}
{\resizebox{1.0\textwidth}{!}{%
\setlength{\tabcolsep}{1.0em}
\begin{tabular}{l l l l l l l}
\toprule \noalign{\vskip 1.5pt}
$d$ & Models & RMSE & CRPS & log-score & CVG(95\%) & Timing (mins) \\
\midrule \noalign{\vskip 2.5pt}
10 & BART        & 0.985 & 0.546 & 1.400 & 0.943 & 65 \\
   & sSGV        & 1.122 & 0.622 & 1.526 & 0.951 & 1.2 \\
   & ResTGP-Full & 0.963 & 0.527 & 1.374 & 0.944 & 6.7  \\
   & ResTGP-PP   & 0.806 & 0.439 & 1.200 & 0.937 & 10 \\
   & ResTGP-WN   & 0.913 & 0.474 & 1.243 & 0.927 & 8  \\
   & ResTGP-WN (PMMH) & 0.857 & 0.442 &  1.080  &  0.965  & 246\\
\midrule \noalign{\vskip 2.5pt}
20 & BART        & 1.035 & 0.575 & 1.452 & 0.949 & 112 \\
   & sSGV        & 1.328 & 0.741 & 1.691 & 0.948 & 3.3 \\
   & ResTGP-Full & 0.931 & 0.511 & 1.341 & 0.936 & 8.0  \\
   & ResTGP-PP   & 0.874 & 0.481 & 1.279 & 0.936 & 18 \\
   & ResTGP-WN   & 0.952 & 0.502 & 1.286 & 0.937 & 23  \\
   & ResTGP-WN (PMMH)   & 0.895 & 0.468 & 1.157 & 0.962 & 770  \\
\bottomrule
\end{tabular}%
}}
\label{table: 10D-20D-illustration}
\end{table} \vspace{-0.2cm}

As seen from Table~\ref{table: 10D-20D-illustration}, ResTGP models show better predictive performance than BART \citep{Chipman2010} and sSGV for $d=10, 20$. Even though a stationary input-scaled Mat\'ern covariance is used, \texttt{ResTGP-Full} has better RMSE than \texttt{sSGV} and \texttt{BART}. Among ResTGP models,  \texttt{ResTGP-PP} has smallest RMSE while \texttt{ResTGP-WN} with PMMH  method has smallest log-score for both $d=10$ and $d=20$. For $d=10$, \texttt{ResTGP-PP} gives best performance in terms of RMSE, CRPS among all the methods; for $d=20$, \texttt{ResTGP-PP} performs best in terms of RMSE while \texttt{ResTGP-WN} with the PMMH method outperforms other methods in terms of CRPS and log-score. This indicates that adaptively inferred tree with a PP or WN as leaf models can better capture the covariance structure in the data. The better performance in \texttt{ResTGP-PP} over \texttt{ResTGP-Full} is likely because of accounting for the model uncertainty in the tree structure. In this example, \texttt{ResTGP-WN} achieves better predictive performance with the PMMH method than that with the MAP method while the PMMH method requires significant amount of more computing time. 

The diagnostics in Section S.2.3 indicate that all methods show good agreement for predictive mean; see Panel (a) in Figures S.3 and S.4. Among all the models, \texttt{ResTGP-WN} shows input-dependent predictive standard deviations; see Panel (b) in Figures S.3 and S.4. \texttt{BART} has nearly flat (or constant) predictive standard deviations which might suggest a limitation in homogeneous noise variance assumption. \texttt{sSGV}, \texttt{ResTGP-Full}, and \texttt{ResTGP-PP} fail to show adaptation to input-dependent noise structure in their predictive standard deviations. The timing results are discussed in Section S.2.3.

\section{Application to Storm Surge}  \label{sec: app}

Risk assessment of storm surges is often performed by running large-scale simulations of ocean circulation models over many different storm characteristics so that devastating economic damages and loss of lives can be prevented. In the United States, the ADCIRC ocean circulation model \citep{Luettich2004, Dietrich2011, Dietrich2012} is the primary computer model (or simulator) to predict storm surges in coastal areas that can provide simulations over a wide range of spatial resolutions with different physical processes.  ADCIRC has also been certified by the Federal Emergency Management Agency (FEMA) for coastal flood hazard study. For more details about this computer model, see \citet{Ma2022PPCokrig}. 

In this study, we consider the ADCIRC simulator  with  inputs  $(\Delta P$, $R_p$, $V_f$, $\vartheta$, $B$, $\mathbf{\Ell})$ whose  physical meanings are given in Table~S.3 of Supplementary Material and surge outputs at a single spatial location. The storm surge simulations from  the ADCIRC simulator were produced by a maximin LHS of 3,500 design points  given by  the five storm parameters: $\Delta P$, $R_p$, $V_f$, $\vartheta$, $B$, where each storm's landfall location is repeated along the coastline with spacing $R_p/2$. Only storms that make landfall in or near the study regions were included  between latitudes 25.80 degrees and 27.39 degrees along the coastline of Southwestern Florida. For each storm, the meteorological input files start 3 days and 23 hours before the storm makes landfall and continue for 24 hours post-landfall with a total duration of 4 days and 23 hours. The storm input design covers a fairly wide space that is constructed from all historical storm information. Based on these 3,500 LHS design points, the entire study produces $56,277$ model runs with more than 10 million mesh nodes covering the full study region. To illustrate the proposed methodology, where scalar output is assumed in the model, we shall focus on storm surge output at one spatial location with sea level depth between 0 and 10 meters which could be of great interest for risk assessment due to storm surge. Generalization of the results to massive spatial locations over the entire domain is investigated elsewhere. 

Exploratory analysis in Figure S.6 of Section S.3 shows that the spatial dependence shows different ranges for each input parameter and the variance has modest variation in the semivariogram plots. This might justify a plausible covariance model with different range parameters such as the input-rescaled Mat\'ern correlation. The storm surge response surface generally shows some stable trend that might vary according to the landfall location, which also suggests that an input varying model might be a more suitable model. The storm surge also shows some very small values near zeros across different input settings, which  indicates a type of regime change in the response function. Because this dataset comes from physical simulation, standard framework of Gaussian process modeling \citep[e.g.][]{Sacks1989, Santner2018} is not too unreasonable; see \citet{Ma2022PPCokrig}.      

To perform a comparative study, the 3,500 unique storm settings are randomly split into 2,800 (80\%) for training and 700 (20\%) for testing, resulting in $n = 44,801$ simulations for training various models and $n^*=11,476$ for testing held out performance. All the inputs are rescaled into $[0,1]^d$ with $d=7$. We compared predictive performance for the following models: sSGV, BART, ResTGP-Full, ResTGP-PP, ResTGP-WN. Other methods such as tgp and deep GP \citep{Sauer2023} are computationally too expensive to be included in the comparison. In all the GP models, we use the input-scaled Mat\'ern correlation function with a nugget parameter. The results in Table~\ref{table: Storm surge} show the predictive advantage of ResTGP over sSGV, albeit with more computing time since ResTGP uses SMC algorithms for tree inference. ResTGP-PP gives smallest RMSE but ResTGP-WN gives smallest CRPS and log-score which are expected because  the white-noise leaf model can improve predictive distribution.

\begin{table}[htbp] 
\centering
\normalsize
   \caption{Predictive performance under different models over $n^*=11,476$ held out data. 
%The reported metrics are the root mean squared error (RMSE), continuous ranked probability score (CRPS), empirical coverage of the nominal 95\% predictive intervals (CVG), and log-score.  Timing is reported for both estimation and prediction. 
MAP=Maximum a posteriori. PMMH=particle marginal Metropolis-Hastings.}
  {\resizebox{1.0\textwidth}{!}{%
  \setlength{\tabcolsep}{1.5em}
   \begin{tabular}{l l  l lll} 
   \toprule \noalign{\vskip 1.5pt} 
%Models   & RMSE &  computing time (in hours) \\ \noalign{\vskip 1.5pt} 
  Models & RMSE & CRPS & log-score & CVG(95\%)  & Timing (mins) \\
\midrule \noalign{\vskip 2.5pt} 
BART & 0.126   & 0.065    &  -0.653 & 0.881 & 780   \\ 
% SoftBART &    &     &  &  \\ (numerical issues, fail always)
  sSGV ($r=30$)  & 0.097 & 0.046 & -1.047 & 0.926 &6     \\
    sSGV ($r=60$)  & 0.070 & 0.033 & -1.339 & 0.938 &18     \\
   % sSGV ($r=120$)  & 0.056 & 0.026 & -1.550 & 0.939 &56     \\
 %deepGP &    &     &  &  \\
 ResTGP-Full &  0.076  & 0.032  & -1.406 &0.938 & 24    \\
 ResTGP-PP (MAP) & 0.055   & 0.024    & -1.676  & 0.946 & 105   \\
  %ResTGP-PP (PMMH) &0.059    &0.025     &-1.691  & 0.954 & 548   \\
 %ResTGP-WN (MAP) & 0.078   &  0.033   & -1.702 & 0.944 & 108   \\
 ResTGP-WN (PMMH) & 0.072   & 0.023    &-1.966  & 0.987 & 334  \\
\bottomrule
   \end{tabular}%
   }}
   \label{table: Storm surge}
\end{table}  

Prediction results in Figure S.7 from Section S.3 suggest that  sSGV, ResTGP-PP, and ResTGP-WN perform very well in capturing the predictive mean. Predictive standard deviations shown in Figures S.8-S.10 from Section S.3 indicate that ResTGP-PP has smallest predictive standard deviations, while ResTGP-WN shows more heterogeneity in predictive uncertainty. This partly explains the lower log-score for ResTGP-WN but similar CRPS for ResTGP-WN and ResTGP-PP. The heterogeneity in predictive uncertainty is especially important for risk assessment of storm surges where the spatially heterogeneous tail probability of extreme events could be exploited for scientific decision making.

\section{Concluding Remarks} \label{sec: discussion}
In this paper we develop a residual tree learning framework for Gaussian processes (GP) for modeling data in multi-dimensional domain with large datasets. The proposed method exploits the multiscale decomposition of GP over tree structures and learns local information adaptive along the tree structures to allow adaptive shrinkage.  This framework leads to two key instances of models according to the leaf model specification (i.e., PP and WN). The resulting likelihood functions can be recursively computed with a linear computational cost in terms of the number of tree nodes given the tree. To search over the tree space, we developed a sequential Monte Carlo algorithm to effectively allow global exploration for Bayesian inference.  Viewing this model as a nonparametric prior, we prove that the proposed model enjoys posterior consistency in the almost sure sense for estimating continuous functions under several different loss functions. This guarantees that the proposed nonparametric prior is well-behaved and provides enough support for predicting the function class.     

The proposed method provides a new framework for integrating GP and tree-based methods in a computationally scalable way for higher dimensional space. A key challenge  in partition-based GP methods is to allow long-range dependence between partition regions, which is addressed through the use of multiscale modeling across different levels of recursive random partitioning of the domain, thus it allows both spatial adaptation and computational scalability. There are several modeling choices such as knot selection which could be further exploited. This is investigated elsewhere. 

Finally, there are several potential future research directions. Many real-world applications produce multivariate and space-time datasets, thus, it would be interesting to extend the proposed idea to account for multivariate and space-time dependency among multiple variables and even to allow different random recursive partitioning priors for each variable. While the proposed method is focused on axis-aligned partitions, it would be more beneficial to adopt an adaptive (or non-axis aligned) recursive partitioning strategy, which may require new efficient inferential approaches. The current framework is focused on single tree structure and performs satisfactorily. The proposed model could also be extended to allow an ensemble method, where additive tree structures are used. This could further improve predictive accuracy. %On the theoretical side, it would be interesting to derive posterior contraction rates for estimating certain function class with the proposed model. 

\section*{Supplementary Material}
The Supplementary Material  includes all the additional details. 
\textsf{R} code to reproduce the results are at \url{https://github.com/pulongma/ResidualTree-paper}. 
%a zipped file that includes \textsf{R} code for reproducing the numerical illustrations. 

%\newpage 
\begin{singlespacing}
\bibliographystyle{apalike}
\bibliography{main}
\end{singlespacing}

\newpage
\appendix

%%%% Main text entry area:
\setcounter{equation}{0}
\setcounter{page}{1}
\setcounter{table}{0}
\setcounter{section}{0}
\setcounter{figure}{0}
%\numberwithin{table}{section}
\renewcommand{\theequation}{S.\arabic{equation}}
\renewcommand{\thesection}{S.\arabic{section}}
\renewcommand{\thesubsection}{S.\arabic{section}.\arabic{subsection}}
\renewcommand{\thepage}{S.\arabic{page}}
\renewcommand{\thetable}{S.\arabic{table}}
\renewcommand{\thefigure}{S.\arabic{figure}}
%\baselineskip=15pt

%\clearpage\pagebreak\newpage

\renewcommand{\theequation}{S.\arabic{equation}}
\renewcommand{\thesubsection}{S.\arabic{section}.\arabic{subsection}}
\renewcommand{\thelemma}{S.\arabic{lemma}}
\renewcommand{\theremark}{S.\arabic{remark}}
\renewcommand{\thecorollary}{S.\arabic{corollary}}

\newtheorem{prop}{\textbf{Proposition}}
\renewcommand{\theprop}{P\arabic{prop}}

\spacingset{1.45} % DON'T change the spacing!

%\newpage

\setlength{\cftsecnumwidth}{2.5em}     % Space for S.1
\setlength{\cftsubsecnumwidth}{3em}  % Space for S.1.1

\newtheorem*{proposition*}{\bf Proposition}

\clearpage\pagebreak\newpage

%\if1\blind
%{
%  \title{\bf Supplementary Materials: A Residual Tree Gaussian Process Modeling Framework for High-Dimensional Data}
%  \author{
%    Pulong Ma\footnote{\textit{Address of Correspondence}: Pulong Ma, Assistant Professor, Department of Statistics, Iowa State University;  2438 Osborn Dr, Ames, Iowa 50011.
%Email: plma@iastate.edu.}\\
%    Department of Statistics, Iowa State University\\
%    %plma@iastate.edu\\
%    and\\
%    Li Ma \\
%    Department of Statistics \& Data Science Institute, University of Chicago\\
%    }
%      \date{}
%  \maketitle
%
%} \fi
%
%\if0\blind
%{
%    \title{\LARGE\bf Supplementary Materials: A Residual Tree Gaussian Process Modeling Framework for High-Dimensional Data}
%\maketitle
%  \medskip
%} \fi

\begin{center}
{\Large \textbf{Supplementary Material: A Residual Tree Gaussian Process Modeling Framework for High-Dimensional Data}} 
\end{center}

The Supplementary Material consists of five sections. 
Section~\ref{supp: posterior inference} presents the details of posterior inference including parameter estimation via the sequential Monte Carlo algorithm and prediction. 
Section~\ref{supp: numerical example} presents additional numerical results in the simulation results. 
Section~\ref{supp: real application} presents some additional numerical results in the real application. 
Section~\ref{supp: ms decomposition} contains some  additional technical results related to multi-scale Gaussian process decomposition with their proofs. 
Section~\ref{supp: posterior consistency} contains all the results to validate three conditions and the detailed proofs for Theorem 4 and related results. 

\newpage

\section{Posterior Inference} \label{supp: posterior inference}

The posterior distribution of the tree $\calT$ and global covariance parameters $\bftheta$ given the data $p(\calT, \bftheta |\bfy) \propto p(\bfy | \calT, \bftheta) p(\bftheta) p(\calT)$. The key challenge for posterior inference is to sample from the tree posterior. Section~\ref{supp: SMC} describes the detailed implementation of a sequential Monte Carlo (SMC) algorithm and how it can be used to perform empirical Bayes estimation. Section~\ref{supp: PMMH} describes a particle marginal Metropolis-Hastings  \citep[PMMH;][]{Andrieu2010} sampler to perform fully Bayesian inference for both $\calT$ and $\bftheta$.

\subsection{Sequential Monte Carlo algorithm} \label{supp: SMC} 

The sequential Monte Carlo algorithm introduced in Section 3 of the main paper provides the theoretical basis for estimating the marginal tree posterior $p(\calT)$ when $\bftheta$ is given. A naive implementation of the SMC algorithm is to run all the computation in each particle, however this can quickly lead to computational bottleneck and memory explosion. This section presents a novel implementation strategy to avoid this limitation. The central design strategy is to identify each node's  information by its full ancestor path, evaluate the computation once, and let each particle use the result to grow the tree according to the look-ahead strategy in SMC. This removes duplicate numerical work while retaining the particle population and its SMC weights.

This SMC algorithm consists of three main stages: \textbf{(A)}: Shared computation at any level; \textbf{(B)}: Particle pass including sampling and propagation, \textbf{(C)} Resampling SMC over residual trees at fixed $\bftheta$. In SMC, the posterior distributions of decision variables $S(v), J(v), L(v)$  are used as a one-step look ahead proposal $q_{t|t-1}(\calT_t | \calT_{t-1})$ for expanding the tree. In particular, let $(j,c)$ be a proposed split for node $\calA$ with two child nodes $\calA_\lchild^{j,c}, \calA_\rchild^{j,c}$. Define 
\begin{align*}
\Phi^{(1)}(\calA)
&= \rho(\calA) q_0(\calA)
 + \bigl(1-\rho(\calA)\bigr) q_\calA(\calR_\calA)
   \sum_{j=1}^{d} \lambda_j(\calA) \\
&\qquad \times
   \int q_0(\calA_\lchild^{j,c})
        q_0(\calA_\rchild^{j,c})
        F_\calA(\mathrm{d}c \mid j).
\end{align*}

Using terminal likelihoods for the prospective children in the
recursion of Theorem~2 gives the one-step-ahead proposal $Q_{\down{\calA}}$ specified as 
\begin{equation}\label{eqn: proposal}
\begin{aligned}
\mathbb{P}\bigl(S(\calA)=1\bigr)
&= \frac{\rho(\calA) q_0(\calA)}
         {\Phi^{(1)}(\calA)}, \\[6pt]
\mathbb{P}\bigl(
    J(\calA)=j,\,
    L(\calA)\in\mathrm{d}c
    \mid S(\calA)=0
\bigr)
&= 
\frac{
    \bigl(1-\rho(\calA)\bigr)
    q_\calA(\calR_\calA)\lambda_j(\calA)
    q_0(\calA_\lchild^{j,c})
    q_0(\calA_\rchild^{j,c})
}{
    \Phi^{(1)}(\calA)-\rho(\calA)q_0(\calA)
} \\
& \quad \times F_\calA(\mathrm{d}c \mid j).
\end{aligned}
\end{equation}

Throughout, $t$ indexes the SMC step, so that the pending nodes processed at step $t$ are at depth $t-1$; $\downarrow\! v$ is the ancestor path of a node $v$ (from the root to $\Pa[v]$); $M$ is the number of candidate cut points per axis, with $M=1$ for the middle cut; $\calP_t$ is the set of distinct paths $\downarrow\! v$ with $v$ pending at step $t$; $q_0(v)$ is the terminal likelihood of the residuals at $v$ under the leaf model $\calM_0$; and $\Phi^{(1)}(v)$ is the one-step marginal, that is, the recursion of Theorem~2 with $\Phi$ at the two children replaced by $q_0$.
 
\paragraph{Shared state along a path.}
Write $L_uL_u^{\top}=C_u(\calR_u,\calR_u)$ and $b_u(\bfx):=L_u^{-1}C_u(\calR_u,\bfx)\in\mathbb{R}^{r}$ for the \emph{residual features} of location $\bfx$ at ancestor $u$. Then $C_v$ of Section~2.1 can be computed along $\downarrow\! v$: 
\begin{equation}\label{eqn: feature stack}
\begin{split}
C_v(\bfx,\bfx')=K(\bfx,\bfx')-\sum_{u\in\downarrow v}b_u(\bfx)^{\top}b_u(\bfx'),
\\
C_v(\calR_v,\bfx)=K(\calR_v,\bfx)-\sum_{u\in\downarrow v}B_u(\calR_v)^{\top}b_u(\bfx),
\end{split}
\end{equation}
with $B_u(\calR_v)=[b_u(\bfx)]_{\bfx\in\calR_v}$ an $r\times r$ block. Every $\bfx\in v$ belonged to every ancestor of $v$, so when the stack $\{b_u(\bfx)\}_{u\in\downarrow v}$ is cached along the path, only the newest layer $b_{\Pa[v]}(\bfx)$ is computed on arrival at $v$, and its right-hand side $C_{\Pa[v]}(\calR_{\Pa[v]},\bfx)$ was already formed at the parent. Forming $C_v(\calR_v,\bfx)$ for all $\bfx\in v$ at step $t$ therefore costs $\calO(n_v(rd+r^2t))$ flops, against $\calO(n_vr^2t^2)$ if all layers are recomputed. The state retained for a path is the residual $\xi_v$, the stack ($n_vrt$ numbers, or $n_vr$ if only the newest layer is kept and the rest recomputed), $L_v$ and the blocks $B_u(\calR_v)$, that is, $\calO(n_vrt+r^2t)$ memory.
 
\paragraph{Stage A: shared computation.}
At step $t$ every pending node is at depth $t-1$. The algorithm processes the whole pending frontier of each particle, but computes each quantity once per \emph{distinct path} $\downarrow\! v\in\calP_t$: two particles whose trees agree along $\downarrow\! v$ hold identical residuals and covariance factors at $v$ and share the work there, whatever they decided elsewhere. For a pending node $v$ with $n_v$ residual observations, the cut-independent work is maximin knot selection, $\calO(n_vrd)$; the cross-covariances~\eqref{eqn: feature stack}, $\calO(n_v(rd+r^2t))$; one $r\times r$ Cholesky factor; and the knot likelihood $q_v(\calR_v)$ together with the regressed residuals 
$\xi_v(\bfx)-C_v(\calR_v,\bfx)^{\top}C_v(\calR_v,\calR_v)^{-1}\xi_v(\calR_v)$, $\calO(n_vr)$. Because $\calR_v$ does not depend on the cut, this part, $\calO(n_vr(d+rt))$, is shared by the leaf term and by all candidate splits. 

Under the \texttt{WN} leaf model the terminal likelihood $q_0(v)$ follows from the normal--inverse-gamma marginal in Theorem~2 and costs $\calO(n_v)$. The candidate children of each split need only the sums and sums of squares of the regressed residuals; after one sort per axis, prefix sums deliver them for all $M$ cut points of that axis, so enumerating $(S,J,L)$ costs $\calO(n_vd\log n_v + d M)$ and the per-path total is
\begin{align*}
A^{\texttt{WN}}_v=\calO\bigl(n_vr(d+rt)+n_vd\log n_v + dM\bigr).
\end{align*}
Under the \texttt{PP} leaf model the terminal likelihood is the modified predictive process which requires $\calO(n_vr^2+r^3)$ flops. Each candidate split, however, requires its two children to be equipped with their own knots. This would require cross-covariance construction one layer deeper and PP likelihood evaluation for each child, thus requiring $\calO(n_{\mathrm{child}}\,r(d+r(t+1)))$. The total cost is 
\begin{align*}
A^{\texttt{PP}}_v=\calO\bigl(n_v\,r\,dM\,(d+rt)\bigr),
\end{align*}
which is $\calO(n_vrd(d+rt))$ for the middle cut ($M=1$) and $M$ times larger
for an adaptive cut with $M$ candidate points per axis.
 
Finally, summing over the distinct paths of step $t$ introduces the
\emph{path-sharing factor}
\begin{equation}\label{eqn: sharing factor}
\gamma_t:=\frac{1}{n}\sum_{\downarrow v\in\calP_t}n_v\in (0,N],
\end{equation}
since the pending nodes of one particle are disjoint and hence $\sum_vn_v\le n$ within a particle. The value $\gamma_t\leq 1$ corresponds to a single tree and $\gamma_t=N$ to no sharing at all. The cut rule governs $\gamma_t$ as well as $M$. With the middle (or median) cut, a path is determined by its sequence of split axes and left/right choices, so $|\calP_t|\le(2d)^{t-1}$ independently of $N$, and particles share all computation on the part of the tree where they agree. With an adaptive cut the candidate positions depend on the data in the node (or on a particle-specific draw), children of particles that disagree anywhere above $v$ are distinct, so $|\calP_t|\le\min\{N2^{t-1},(2dM)^{t-1}, N \lfloor n/r \rfloor\}$.  Stage~A is embarrassingly parallel over $\calP_t$.
 
\paragraph{Stage B: propagation and weighting.}
For every particle and every pending node $v$, the algorithm draws $(S,J,L)$ from the cached one-step-ahead proposal $Q_{\downarrow v}$ of \eqref{eqn: proposal}. If $S=1$, $v$ is a terminal node; otherwise $v$ is split into its two children. % multiplies the particle weight by $G_{\downarrow v}=\Phi^{(1)}(v)/q_0(v)$. 
Because $Q_{\downarrow v}$ is the exact one-step posterior of the decision, the incremental weight is the same for every outcome of the draw, which is why it can be computed in Stage~A and applied without reference to the realized decision. The computational work is one categorical draw over $1+dM$ outcomes and one multiplication per pending node. Since a particle has at most $\min\{2^{t-1},n/r\}$ pending nodes,  Stage~B costs $\calO(N\,dM\min\{2^{t-1},n/r\})$ per step and $\calO(N\,dM\,n/r)$ in total, which involves no covariance evaluations.
 
\paragraph{Stage C: tempered resampling.}
When the effective sample size $ESS:=(\sum_i(W^i)^2)^{-1}$ falls below $\epsilon_0N$ (default $\epsilon_0=0.5$) and pending nodes remain, particles are resampled with probabilities $p^i\propto(W^i)^{\alpha}$, $\alpha\in(0,1]$ \citep{Liu2001}, where $\alpha$ controls the diversity of samples to be sampled.  Each offspring of ancestor $a_i$ receives the correction $\omega^i=W^{a_i}/p^{a_i}$. The resampled measure $N^{-1}\sum_i\omega^i\delta_{\calT^{a_i}}$ has conditional expectation $\sum_jW^j\delta_{\calT^j}$, so the target is preserved for every $\alpha$. $\alpha=1$ is the bootstrap resample, and smaller $\alpha$ retains particles whose splits pay off only at deeper levels. This step costs $\calO(N)$ with decision records shared by reference, independently of $\alpha$.

\paragraph{Computational complexity.}
Let $D$ be the realized number of SMC rounds, $\bar\gamma=\max_{t\le D}\gamma_t$, and $H=\sum_{t=1}^D\sum_{i=1}^N|\calP_t^i|$ the total number of particle-node decisions. Since each split removes $r$ knot observations, $D\le\lfloor n/r\rfloor$ and $H\le ND\lfloor n/r\rfloor$; logarithmic depth requires an additional sample-count balance condition.

Suppose that each ancestral state has at most $M$ shared candidate cuts per coordinate and that ancestral factors are cached. Let $s(B)$ denote the cost of drawing from a cached distribution over $1+B$ outcomes: $s(B)=1$ with alias sampling, $\log(1+B)$ with binary search, and $1+B$ with a linear scan. With linear-time resampling and shared particle records, the total costs satisfy
\begin{align}
T^{\mathrm{PP}}
&=\calO\!\left(
\bar\gamma nrdM(dD+rD^2)
+s(dM)H+ND 
\right), \label{eq: complexity PP}\\
T^{\mathrm{WN}}
&=\calO\!\left(
\bar\gamma n
\left[
rdD+r^2D^2+dD\log(n+1)+\frac{dMD}{r}
\right]
+s(dM)H+ND
\right), \label{eq: complexity WN}
\end{align}
where the WN bound uses sorting and prefix statistics for ordered or rank-indexed candidate cuts. 

For fixed $N,r,d,M,D$, the displayed PP and midpoint-WN numerical costs are linear in $n$, whereas fresh sorting introduces an $n\log n$ term for adaptive WN. Under sample-count balanced growth, $D=\calO(\log(n/r))$, these numerical costs are at most $\calO(n\log^2 n)$ when the other factors are fixed.  For a given fixed tree, cached likelihood computation costs $\calO(nrdD+nr^2D^2)$.

At step $t$, the numerical caches require $\calO(\gamma_tnrt+|\calP_t|r^2t)$ memory.  Total memory also includes retained ancestral states, proposal tables, candidate buffers, particle records, and data. WN reduces candidate-split evaluation to residual sufficient statistics, whereas PP generally requires additional low-rank covariance calculations for candidate children. The resulting computational advantage depends on the visited ancestral states, realized depth, candidate construction, and cache policy. If obsolete caches are released, the peak memory is the maximum over $t\in \{1, \ldots, D\}$. 

Under the proposed path-share SMC implementation, we never require storing these caches across all SMC steps. Caches are shared across particles and retained only while needed by surviving pending nodes. Unselected candidate states and unused caches are released after propagation and resampling. Thus, active cache memory is bounded by $\calO(\max_{1\le t\le D}\{\gamma_t nrt+|\calP_t|r^2t\})$, excluding temporary candidate buffers and retained tree records. Since $|\calP_t|r\le\gamma_t n$ and $\gamma_t\le N$, the active numerical caches require $\calO(\gamma_t nrt)$ memory. Releasing obsolete caches yields a peak bound of $\calO(\bar\gamma nrD)\le \calO(NnrD)$, where $\bar\gamma=\max_{t\le D}\gamma_t$.

\paragraph{Evidence estimate.}

Denote by $Z(\bftheta)=\int p(\bfy\mid\calT,\bftheta)\,\pi(\mathrm{d}\calT)$
the marginal likelihood of $\bftheta$. With normalized weights $W^i_{t-1}$
entering step $t$ and $G^i_t=\prod_{v}G_{\downarrow v}$ the product over the
pending nodes of particle $i$, Algorithm~\ref{alg: SMC} produces
\citep{Andrieu2010}
\begin{equation}\label{eqn: evidence}
\widehat Z_t(\bftheta)=\widehat Z_{t-1}(\bftheta)\,
\Bigl\{\sum_{i=1}^N W^i_{t-1}G^i_t\Bigr\}\,\bar\omega_t,
\qquad \widehat Z_0(\bftheta)=q_0(\calA^0),
\end{equation}
where $\bar\omega_t=N^{-1}\sum_i\omega^i$ if resampling occurred at step $t$
and $\bar\omega_t=1$ otherwise. The final value
$\widehat Z_D(\bftheta)$ is unbiased for $Z(\bftheta)$ \citep{DelMoral2006},
which is the property required by the particle marginal Metropolis--Hastings
sampler of Algorithm~2. The three stages are described in
Algorithm~\ref{alg: SMC}.

Our experience found that \texttt{ResTGP-PP} with adaptive cut is significantly more expensive than the middle cut, while \texttt{ResTGP-WN} is always scalable even with adaptive cut. Adaptive cut often gives better predictive performance. 
 
\begin{algorithm}[!htbp]
\caption{Path-shared one-step-ahead SMC at fixed $\bftheta$}
\label{alg: SMC}
\begin{algorithmic}[1]
\Require data, $\bftheta$, knot budget $r$, particles $N$, leaf model,
         resampling threshold $\epsilon_0\in(0,1]$, tempering exponent
         $\alpha\in(0,1]$
\Ensure evidence estimate $\widehat Z$ and weighted distinct trees
        $\{(\calT^{(k)},\Omega_k)\}_{k=1}^K$
\State Start every particle $\calT_0^i$ at the root, pending iff $n\ge r$;
       $W_0^i\gets 1/N$; $\widehat Z_0\gets q_0(\calA^0)$; $t\gets 0$
\While{some particle has a pending node}
  \State $t\gets t+1$
  \Statex \textit{A. Shared computation (parallel over distinct paths)}
  \ForAll{distinct paths $\downarrow\! v\in\calP_t$, $v$ pending in some
          $\calT_{t-1}^i$}
     \State From the cached state of $\downarrow\! v$, compute $q_0(v)$,
            the one-step marginal $\Phi^{(1)}(v)$, the proposal
            $Q_{\downarrow v}$ in \eqref{eqn: proposal}, and
            $G_{\downarrow v}\gets\Phi^{(1)}(v)/q_0(v)$; cache the child
            states
  \EndFor
  \Statex \textit{B. Propagation and weighting}
  \ForAll{particles $i$}
     \State $\calT_t^i\gets\calT_{t-1}^i$;\quad $W_t^i\gets W_{t-1}^i$
     \ForAll{pending nodes $v\in\calT_{t-1}^i$}
        \State Draw $(S,J,L)\sim Q_{\downarrow v}$; in $\calT_t^i$ close
               $v$ if $S=1$, otherwise attach its two children, each
               pending iff it holds at least $r$ observations
        \State $W_t^i\gets W_t^i\,G_{\downarrow v}$
     \EndFor
  \EndFor
  \State $c_t\gets\sum_i W_t^i$;\quad $\widehat Z_t\gets\widehat Z_{t-1}\,c_t$;
         \quad $W_t^i\gets W_t^i/c_t$
  \Statex \textit{C. Tempered resampling} \citep{Liu2001}
  \If{pending nodes remain and $\bigl(\sum_i (W_t^i)^2\bigr)^{-1}<\epsilon_0 N$}
     \State $p^i\gets (W_t^i)^{\alpha}\big/\sum_j (W_t^j)^{\alpha}$
     \State $a_1,\ldots,a_N\overset{\mathrm{iid}}{\sim}
            \mathrm{Cat}(p^1,\ldots,p^N)$
     \State $\omega^i\gets W_t^{a_i}/p^{a_i}$;\quad
            $\bar\omega\gets N^{-1}\sum_i\omega^i$
     \State $(\calT_t^1,\ldots,\calT_t^N)\gets
            (\calT_t^{a_1},\ldots,\calT_t^{a_N})$
     \State $\widehat Z_t\gets\widehat Z_t\,\bar\omega$;\quad
            $W_t^i\gets\omega^i/(N\bar\omega)$
  \EndIf
\EndWhile
\State \Return $\widehat Z\gets\widehat Z_t$ and the distinct final trees
       $\{(\calT^{(k)},\Omega_k)\}_{k=1}^K$ with summed weights
       $\Omega_k\gets\sum_{i:\,\calT_t^i=\calT^{(k)}}W_t^i$
\end{algorithmic}
\end{algorithm}

\newpage 
\subsection{Particle marginal Metropolis-Hastings sampler} \label{supp: PMMH} 

A starting point for estimating $\bftheta$ is to maximize the marginal likelihood $\hat{Z}(\bftheta)$ leading to a maximum likelihood estimate (MLE). If priors on $\bftheta$ are used, the maximum a posterior estimate can be obtained. Both the MLE and MAP have been implemented in our package \texttt{ResTree}. A fully Bayesian estimation based on particle marginal Metropolis-Hastings sampler is introduced in Algorithm~\ref{alg: pmmh}. The clear advantage of MLE or MAP lies in its computational time compared to the PMMH sampler but our experience in these examples has indicated that they may sacrifice statistical accuracy in certain occasions. 

Some default tuning parameter settings in the algorithm that work well for the presented numerical examples and are summarized below: 
\begin{itemize}[noitemsep]
\item The nominal depth $D$ of the tree is set to be $D=\lceil \log_2 (n/r) \rceil$. The realized $D$ might be different in SMC algorithm for PP and WN. However, Full leaf model always has maximum tree depth $D$. 
\item The default number of particle size $N$ is set to be $N=100 D$.
\item In the Bayesian CART prior, the hyperparameters are set to be the default one: $\rho_0=0.95, \beta_0=2$. 
\item The resampling  threshold $\epsilon_0$ is set to be $\epsilon_0=0.5$. 
\item The knot selection is default to maximin design that is performed sequentially down the tree from root to leaf. 
\item The default number of knots is set to be $r=30$. 
\item The white-noise leaf model has prior parameter settings: $\mu_0=0, \kappa_0=1, a_0=2, b_0=0.01$.
\item Priors on Mat\'ern covariance parameters $(\sigma^2, \ell_1, \ldots, \ell_d, v)$ are $\sigma^2 \sim IG(2,0.01),$ $p(\ell_i) \propto 1/(1+\ell_i^2) \mathds{1}(\ell_i>0), p(v) \propto 1/(1+v^2) \mathds{1}(v>0)$.
\end{itemize}

\begin{algorithm}[!htbp] 
\caption{Particle Marginal Metropolis-Hastings (PMMH) Sampler}
\label{alg: pmmh} 
%\begin{raggedright}
%\par\end{raggedright}
 \begin{algorithmic}[1]  
 \Require data $(\bfx_i,\bfy_i)_{i\le n}$, $\bftheta^{(0)}$, depth $D$, particles $N$, cut rule, leaf model,  threshold $\epsilon_0\in(0,1]$ %tempering $
 \Ensure $\{(\bftheta^{(k)}, Z^k) :k=1,\ldots, k_0 \}$
 \State  $Z^{0} \gets$ the SMC in Algorithm~\ref{alg: SMC} with $ \bftheta^{(0)}$
 \For{$k=1,\ldots, k_0$} 
 \State   Draw from a proposal: $\bftheta^* \sim J(\bftheta | \bftheta^{(k-1)})$
 \State  $Z^* \gets$ the  SMC in Algorithm~\ref{alg: SMC} with $\bftheta^{*}$
\State Compute MH ratio: 
$$
\alpha:=\min\left\{1, \frac{Z^* p(\bftheta^{*}) J(\bftheta^{(k-1)} | \bftheta^*) }{Z^{k-1} p(\bftheta^{(k-1)}) J(\bftheta^* | \bftheta^{(k-1)} ) } \right\}
$$
\State Accept $\bftheta^{*}$ and $Z^*$ with probability $\alpha$.
\EndFor

\State \Return $\{(\bftheta^{(k)}, Z^k) :k=1,\ldots, k_0 \}$
\end{algorithmic} 
\end{algorithm}
\newpage

%\subsection{Comparison of empirical Bayes estimates and PHHM estimates} 

\section{Additional Simulation Study} \label{supp: numerical example}

\subsection{Likelihood approximation}

In Section 5.1, \texttt{ResTGP-Full} requires the number of leaf observations to be smaller than $r=30$ to determine the maximum depth and \texttt{ResTGP-PP} uses $100\lceil \log_2 (n/r) \rceil$ as the default number of particles for tree inference via SMC. Default hyperparameters are used for \texttt{ResTGP-PP}. 
 
In Figure~\ref{fig: likelihood approximation2}, all the settings are the same except the number of nearest neighbors (or knots) are increased from $r=30$ to $r=100$. The results also show consistent better performance of \texttt{ResTGP-Full} and \texttt{ResTGP-PP} than  \texttt{sSGV} for likelihood approximation, especially in higher input dimensions. Between \texttt{ResTGP-Full} and \texttt{ResTGP-PP}, there is not much difference in approximation accuracy even though a PP  is used at the leaf compared to a full residual covariance. This is partly due to the SMC inference over tree since the knot selections changes along with the tree structure and in part due to the nearly equivalence of full residual structure and PP residual structure. The timing results again make \texttt{ResTGP-Full} stand out. It is worth mentioning that in all these experiments, both \texttt{ResTGP-Full} and \texttt{ResTGP-PP} could be better suited for addressing higher dimensional problems than \texttt{sSGV} when a target covariance function is known. Both of them can be improved in various ways. For instance, all the methods can be improved through better choice of knots or nearest neighbors. For instance, correlation based design for nearest neighbors can improve Vecchia approximation \citep{Kang2023, Katzfuss2022}. 
 
 \begin{figure}[!htpb]
\begin{center}
\begin{subfigure}[t]{1\textwidth}
\makebox[\textwidth][c]{\includegraphics[width=1\textwidth, height=0.25\textheight]{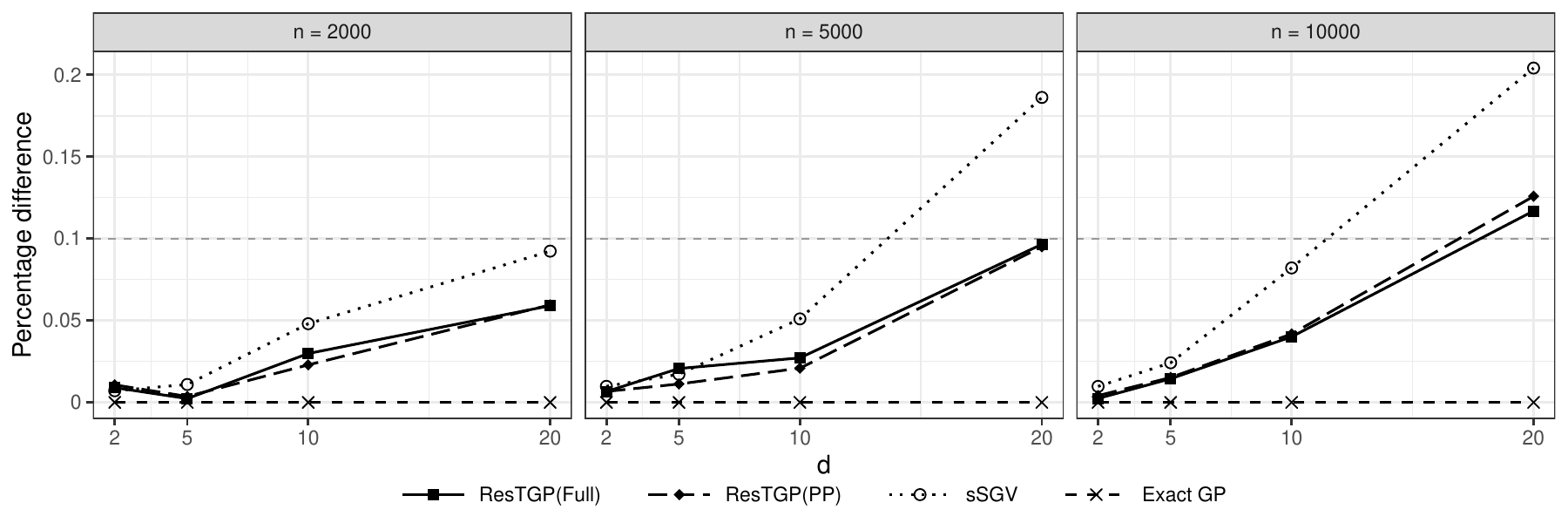}}
%\caption{Likelihood against sample size under different models and input dimension. $x$-axis is on the $\log_{10}$ scale.}
\end{subfigure}
\begin{subfigure}[t]{1\textwidth}
\makebox[\textwidth][c]{\includegraphics[width=1\textwidth, height=0.25\textheight]{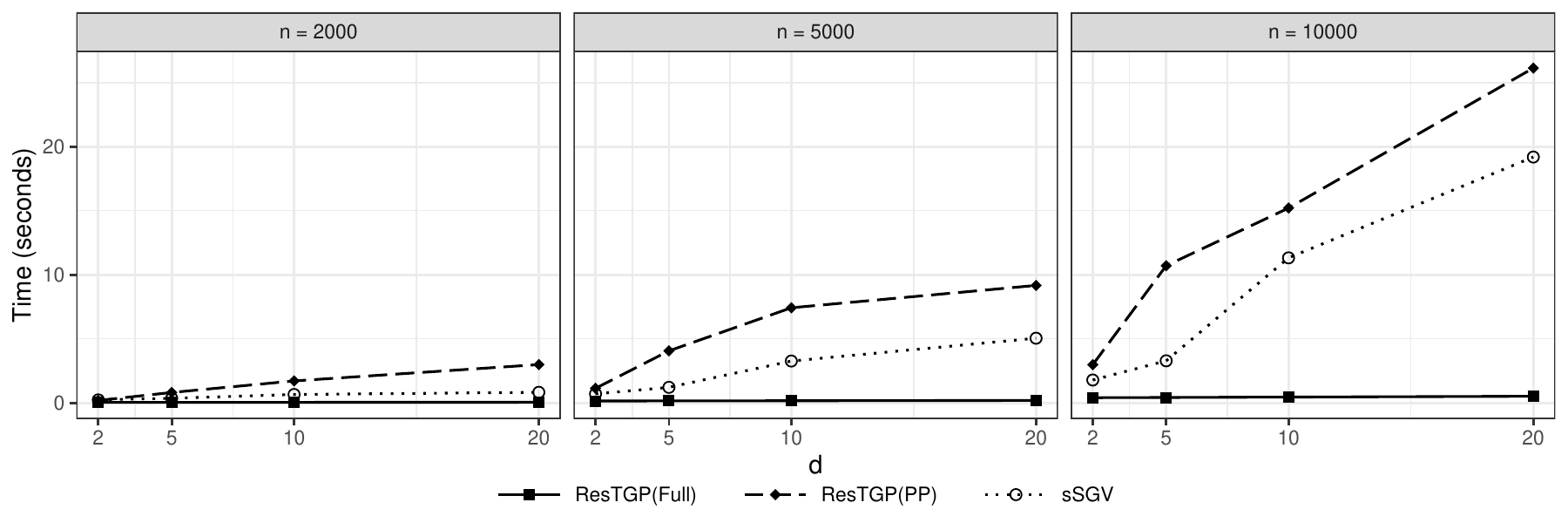}}
\end{subfigure}
\caption{Likelihood approximation versus input dimension $d$ under different $n$.  All methods have fixed $r=100$ knots or nearest neighbors chosen via maximin design. Top panels show the percentage error of log-likelihood compared to exact GP. Bottom panels show the timing.}
\label{fig: likelihood approximation2}
\end{center}

\end{figure}

\newpage

\subsection{Adaptation to input-dependent noise} 

\subsubsection{Results under $n=320$} 
In this example, SGV uses 30 nearest neighbors, \texttt{ResTGP-PP} and \texttt{ResTGP-WN} use 15 knots. \texttt{tgp} is run with 7000 samples including 2000 iterations for burn-in period. The numerical summaries of predictive performance are given in Table~\ref{table: 1D illustration}. This table also confirms the advantage of the proposed ResTGP method with its two models \texttt{ResTGP-PP} and \texttt{ResTGP-WN}.  \texttt{ResTGP-WN} outperform others in terms of CRPS and log-score but \texttt{ResTGP-PP} gives smallest RMSE and \texttt{tgp} gives better CVG that is closer to the nominal level. In terms of timing, \texttt{ResTGP-WN} is competitive to \texttt{SGV} with this small dataset. \texttt{ResTGP-PP} with adaptive cut is slowest. 

\begin{table}[!htbp] 
\centering
\normalsize
   \caption{Predictive performance under different models with $n=320$ training data points over 100 held out data points.
The reported metrics are the root mean squared error (RMSE), continuous ranked probability
score (CRPS), empirical coverage of the nominal 95\% predictive intervals (CVG),
and log-score. ResTGP are estimated using the empirical Bayes method.}
  {\resizebox{1.0\textwidth}{!}{%
  \setlength{\tabcolsep}{1.5em}
   \begin{tabular}{l l  l lll} 
   \toprule \noalign{\vskip 1.5pt} 
  Models & RMSE & CRPS & log-score & CVG(95\%)  & Timing (seconds) \\
\midrule \noalign{\vskip 2.5pt} 
SGV  & 1.03 & 0.527 & 1.44 & 0.934 & 0.25 \\
tgp & 1.08 & 0.489 & 1.38 & 0.942 & 32.3  \\
 ResTGP-PP &  {1.01}  & 0.510   & 1.37 & 0.922 & 69\\
 ResTGP-WN & 1.02   &  0.484   & 1.08 &  0.932  & 0.64\\
\bottomrule
   \end{tabular}%
   }}
   \label{table: 1D illustration}
\end{table}  

\subsubsection{Results under $n=3200$}
This example considers $n=3200$ training data points and compare methods including \texttt{SGV}, \texttt{BART} \citep{Chipman2010}, \texttt{ResTGP-PP}, and \texttt{ResTGP-WN}. \texttt{tgp} was not compared due to its excessive computing time. \texttt{BART} is run with 200 trees and 4000 posterior samples after 1000 burn-in period. \texttt{SGV} used 30 neighbors while both \texttt{ResTGP-PP} and \texttt{ResTGP-WN} used 15 knots. \texttt{ResTGP-PP} used middle cut as it is much faster than that with adaptive cut even though it may sacrifice performance accuracy. \texttt{ResTGP-WN} used uniform cut.  

As expected, all instances of {ResTGP} models outperform SGV and BART in terms of CRPS, log-score, and CVG even when SGV has $2r$ nearest neighbors. The prediction results are visualized in Figure~\ref{fig: 1d example n=3200}. BART shows constant predictive standard deviation and cannot adapt to the input-dependent noise very well. Compared with BART, the predictive standard deviation from GP models can use spatial covariance to allow stronger spatial smoothing. SGV's predictive standard deviation shows local adaptive structure but cannot  capture the input-dependent noises well since its predictive standard deviation is quite scattered around the same value. This is because SGV uses information from their selected neighbors and a stationary covariance function. 

In contrast, both ResTGP-PP and ResTGP-WN can adaptively capture the function with their 95\% prediction interval bands adapted to different regions. Between these two models, ResTGP-WN can better adaptively capture the heterogeneity and show adaptive predictive standard deviations. Table~\ref{table: 1D illustration n=3200} shows that ResTGP-WN outperform other models in terms of CRPS, log-score, CVG. The computing time for ResTGP-WN also shows potential scalability compared with SGV.

\begin{table}[!htbp] 
\centering
\normalsize
   \caption{Predictive performance under different models with $n=3200$ training data points over 100 held out data points.
The reported metrics are the root mean squared error (RMSE), continuous ranked probability
score (CRPS), empirical coverage of the nominal 95\% predictive intervals (CVG),
and log-score. ResTGP models are estimated using the empirical Bayes method.}
  {\resizebox{1.0\textwidth}{!}{%
  \setlength{\tabcolsep}{1.5em}
   \begin{tabular}{l l  l lll} 
   \toprule \noalign{\vskip 1.5pt} 
  Models & RMSE & CRPS & log-score & CVG(95\%)  & Timing (seconds) \\
\midrule \noalign{\vskip 2.5pt} 
              SGV  & 0.949 & 0.492 & 1.371 & 0.938 & 1.4 \\
            BART & 1.082 & 0.569 & 1.503 & 0.940 & 181  \\
  ResTGP-PP &  {0.936}  & 0.480   & 1.336 & 0.941 & 66\\
 ResTGP-WN & 0.989   &  0.458   & 0.953 &  0.947  & 4\\
\bottomrule
   \end{tabular}%
   }}
   \label{table: 1D illustration n=3200}
\end{table}

\begin{figure}[!htpb]
\begin{center}
\begin{subfigure}[t]{1\textwidth}
\makebox[\textwidth][c]{\includegraphics[width=1\textwidth, height=0.4\textheight]{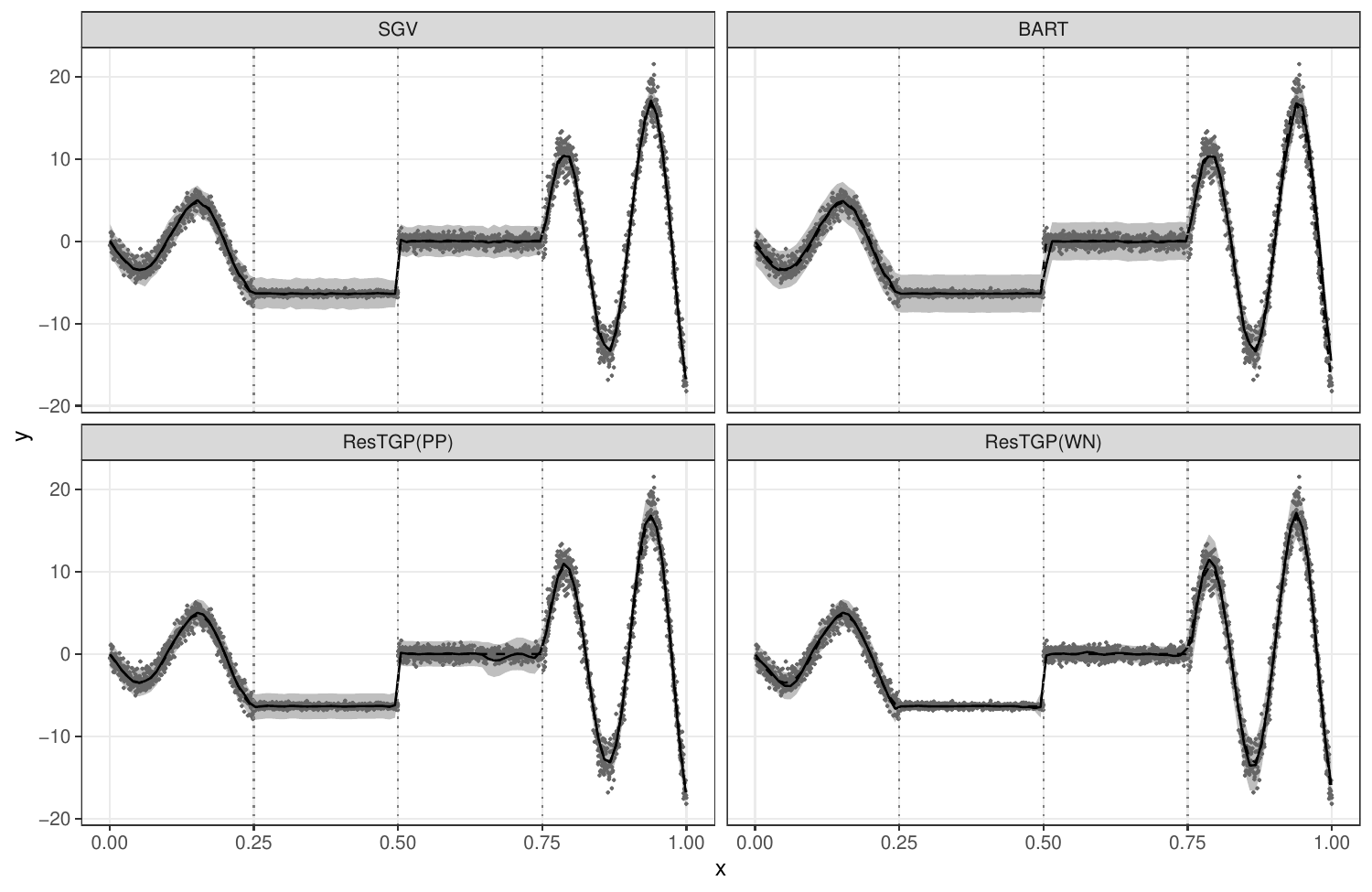}}
\caption{Predictions under different models}
\end{subfigure}%

\begin{subfigure}[t]{1\textwidth}
\makebox[\textwidth][c]{\includegraphics[width=1\textwidth, height=0.25\textheight]{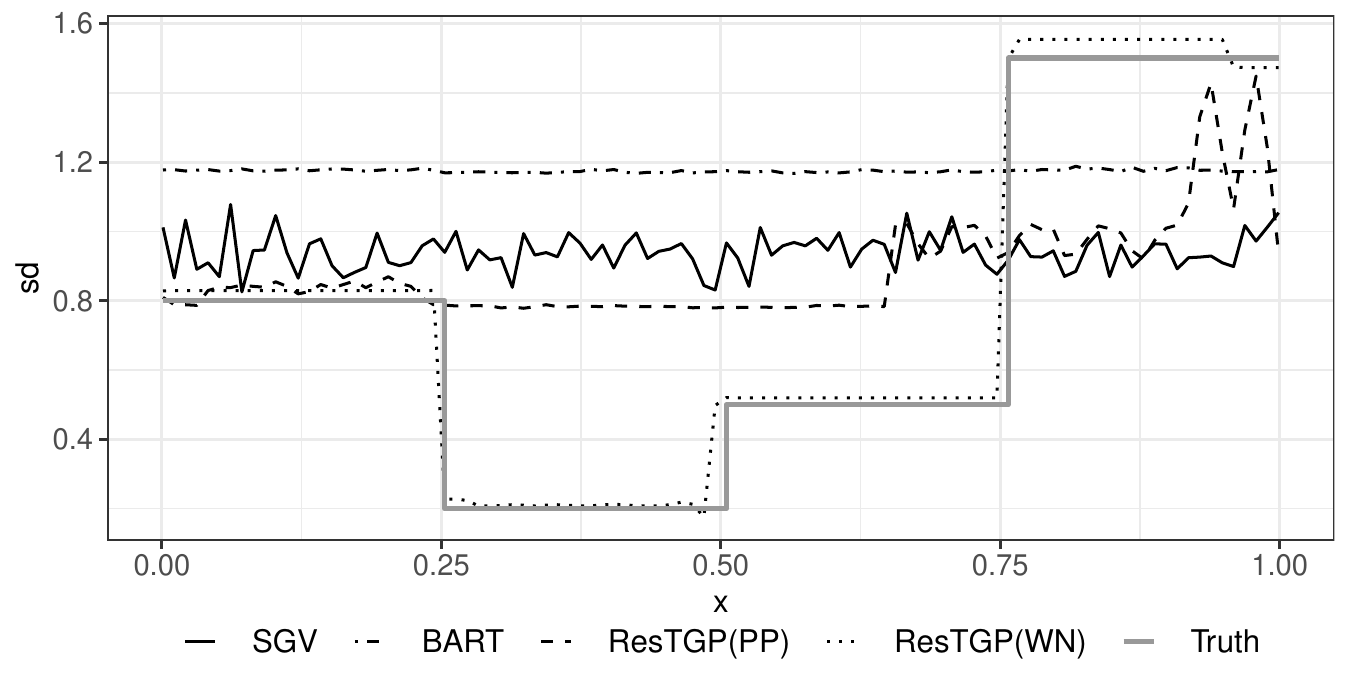}}
\caption{Predictive standard deviations}
\end{subfigure}
\caption{Predictive performance of different models over 100 testing locations equally spaced in $[0, 1]$. All models are trained based on 3200 observations (dots). Panel (a) shows the predictive mean and 95\% predictive intervals under different models. Panel (b) shows how predictive standard deviations change with input. Result for SGV is obtained from the \texttt{GpGp} package.  Result based on BART is obtained from the \texttt{BayesTree} package.}
\label{fig: 1d example n=3200}
\end{center}
\end{figure}

\newpage

\subsection{High-dimensional examples}

\textbf{Timing:} To measure the computing time, all models are fitted on a 10-core machine with 64GB memory. BART is run with 3,000 burn-in and 7,000 post burn-in iterations. ResTGP-PP and ResTGP-WN are run with default settings for hyperparameters in prior distributions. The knots and nearest neighbors are set to be consistent for sSGV and ResTGP models as $r=60$. The number of particles for ResTGP-PP and ResTGP-WN is set to be $200 \lceil \log_2(n/r) \rceil$. The PMMH algorithm for ResTGP-WN is run with 3000 iterations for burn-in and 2000 posterior draws after burn-in. 

Different estimation methods require different number of likelihood evaluations that would significantly impact computing time. Nevertheless, \texttt{sSGV} is fastest among all methods with \texttt{ResTGP-WN} being slowest with PMMH method. Using optimization methods, \texttt{ResTGP-WN}'s computing time is changed from 8 to 23 minutes when $d$ is changed from 10 to 20. 
From $d=10$ to $d=20$, the computing time for \texttt{sSGV} and \texttt{ResTGP-Full} is not affected too much even though their single likelihood evaluations become more expensive. This is partly because the optimization procedures can have a large impact on the number of likelihood evaluations before the algorithm stops. \texttt{ResTGP-WN} and \texttt{ResTGP-PP} becomes more expensive when $d$ increases due to the number of particles and inference for tree structures. 

\noindent \textbf{Predictive diagnostics: }Diagnostics plots shown in Figure~\ref{fig: 10D-illustration} and~\ref{fig: 20D-illustration} show the predictive mean and standard deviations. These results are discussed in Section 5.3. 

\begin{figure}[!htpb]
\begin{center}
\begin{subfigure}[t]{1\textwidth}
\makebox[\textwidth][c]{\includegraphics[width=1\textwidth, height=0.4\textheight]{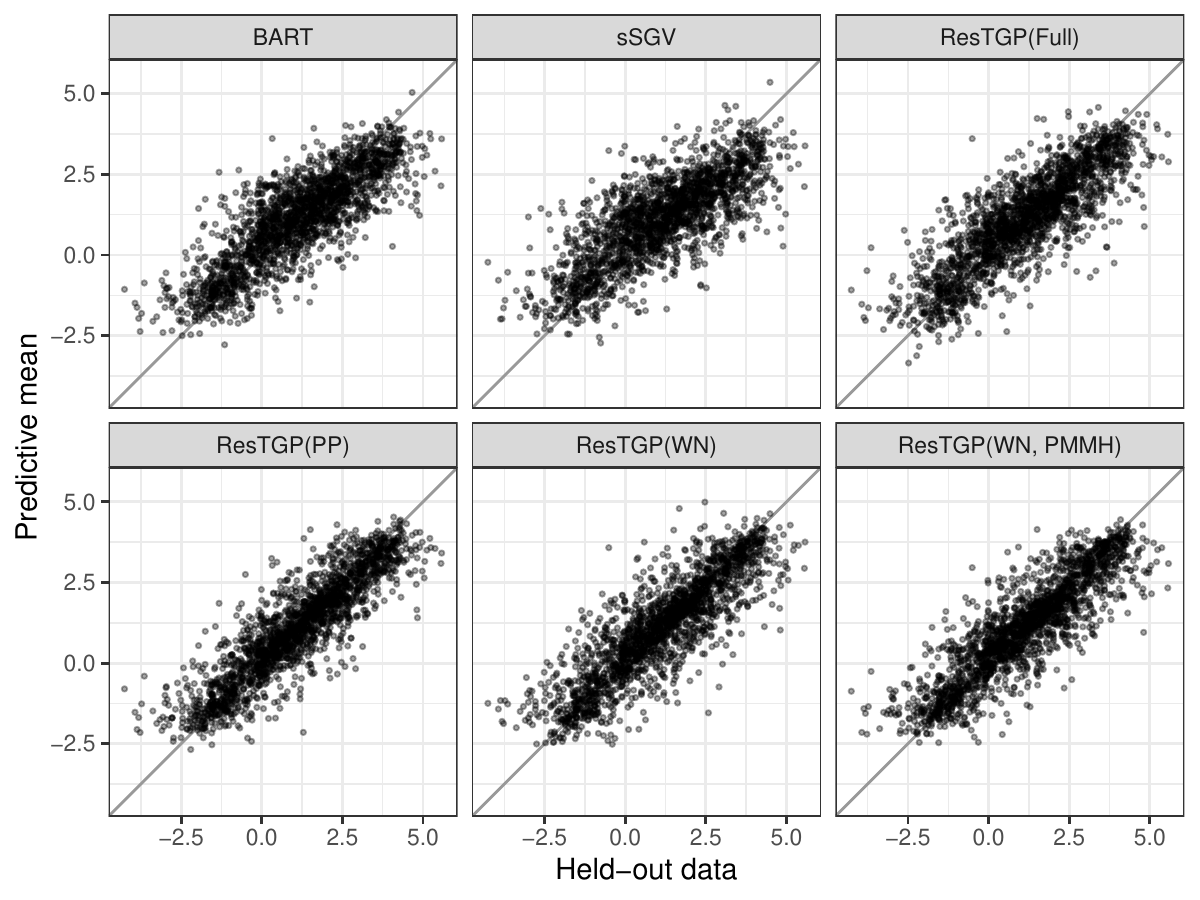}}
\caption{Predictions under different models}
\end{subfigure}%

\begin{subfigure}[t]{1\textwidth}
\makebox[\textwidth][c]{\includegraphics[width=1\textwidth, height=0.4\textheight]{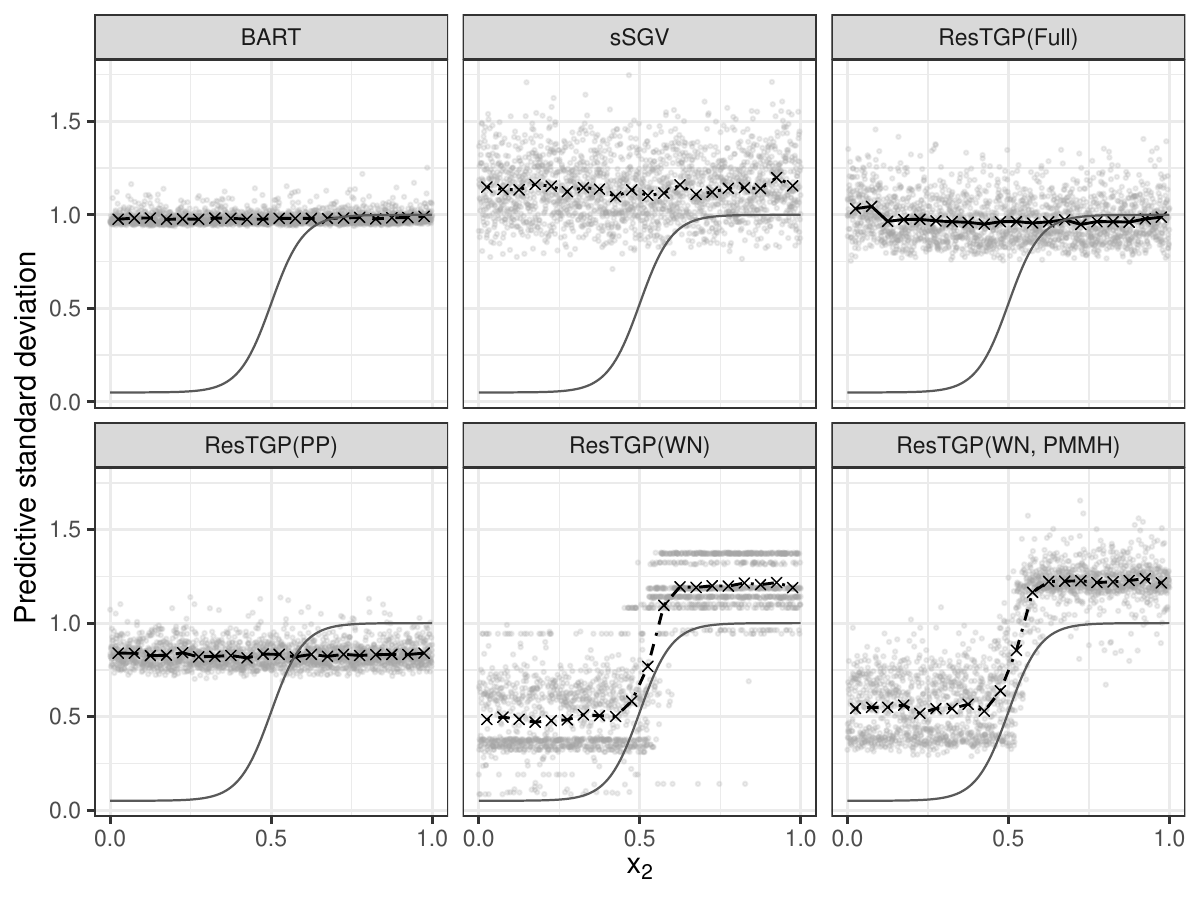}}
\caption{Predictive standard deviations}
\end{subfigure}
\caption{Predictive performance under $d=10$. Panel (a) shows the predictive mean against held-out data under different models with 45 degree line. Panel (b) shows how predictive standard deviations change with input $x_2$. The gray dots reflects different values of other inputs and cross sign shows the average over other input variables except $x_2$. The solid curve shows the true noise function $\sigma(\cdot)$.  Result for sSGV is obtained from the \texttt{GpGp} package. Result based on BART is obtained from the \texttt{BayesTree} package.}
\label{fig: 10D-illustration}
\end{center}
\end{figure}

\begin{figure}[!htpb]
\begin{center}
\begin{subfigure}[t]{1\textwidth}
\makebox[\textwidth][c]{\includegraphics[width=1\textwidth, height=0.4\textheight]{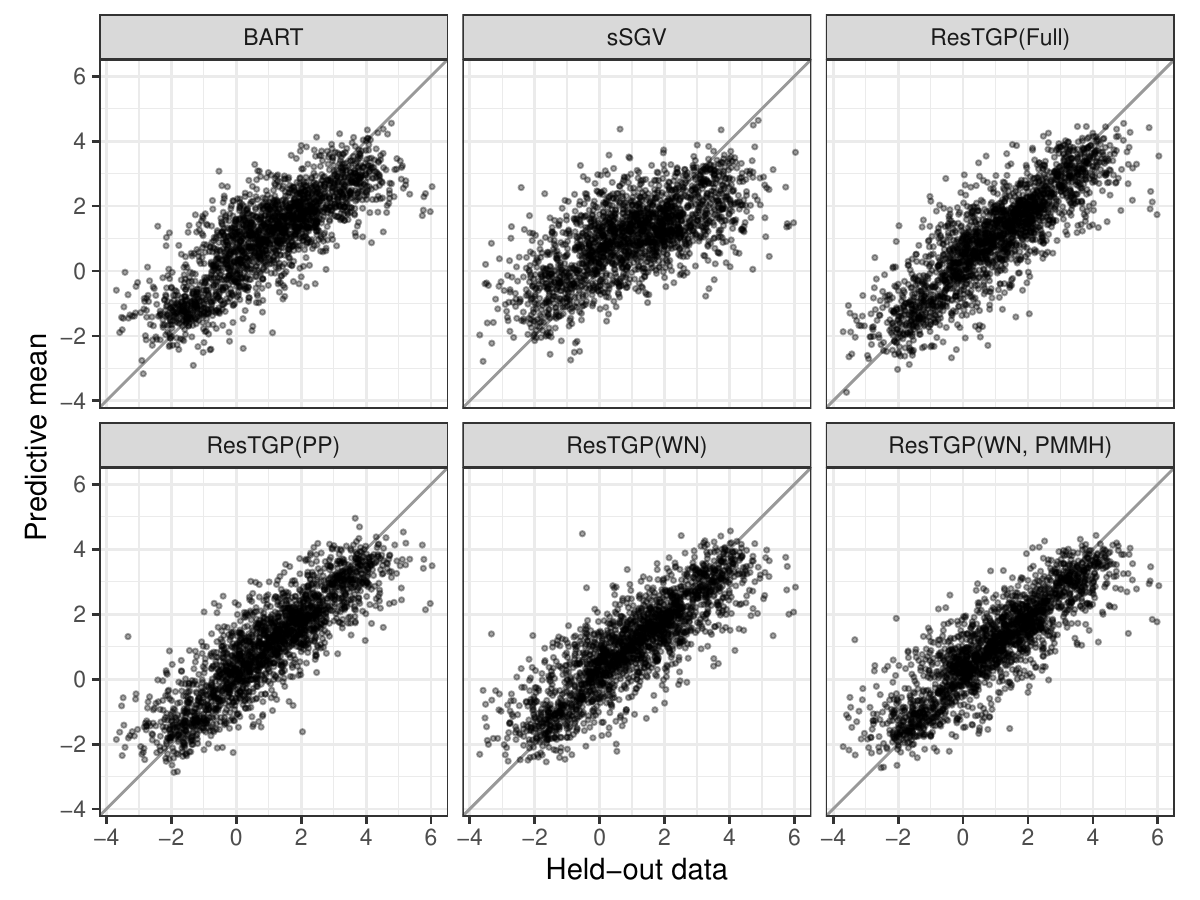}}
\caption{Predictions under different models}
\end{subfigure}%

\begin{subfigure}[t]{1\textwidth}
\makebox[\textwidth][c]{\includegraphics[width=1\textwidth, height=0.4\textheight]{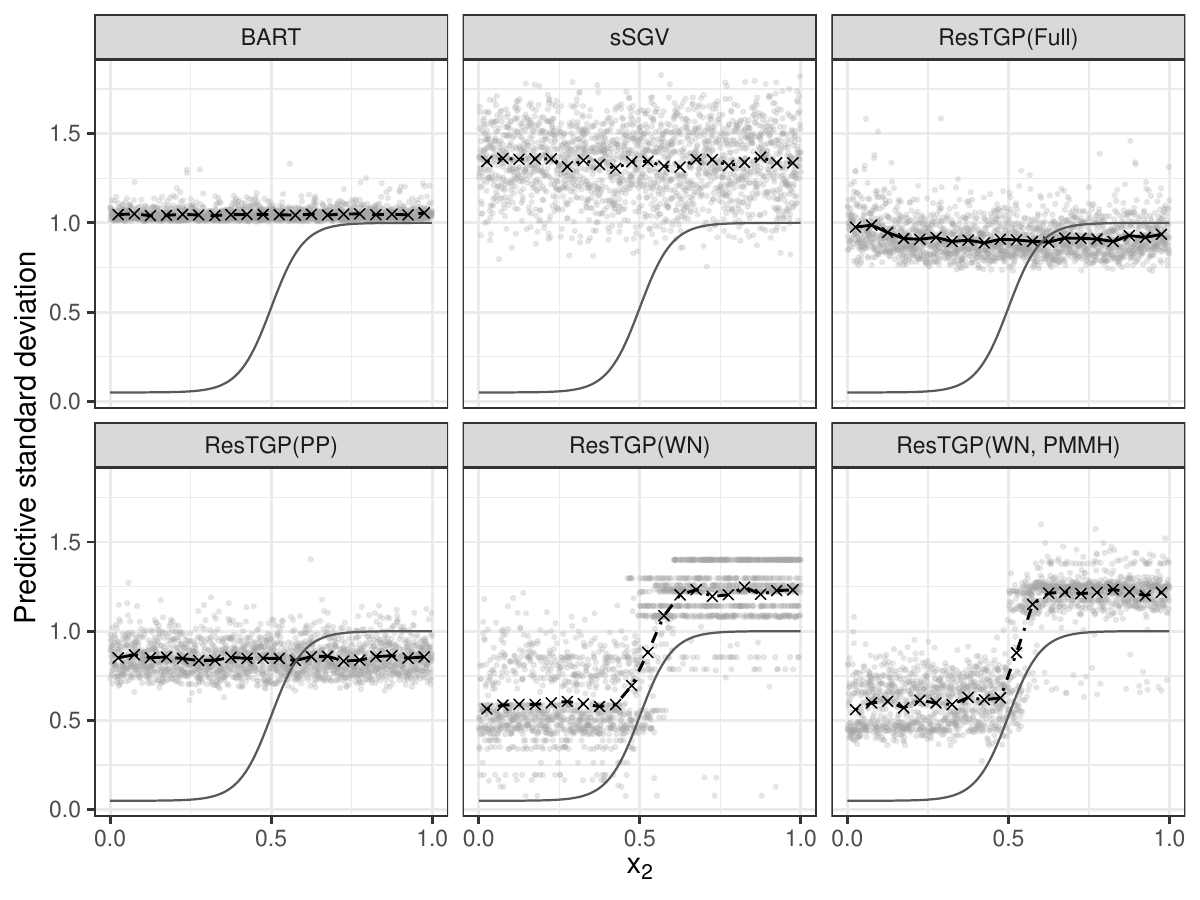}}
\caption{Predictive standard deviations}
\end{subfigure}
\caption{Predictive performance under $d=20$. Panel (a) shows the predictive mean against held-out data under different models with 45 degree line. Panel (b) shows how predictive standard deviations change with input $x_2$. The gray dots reflects different values of other inputs and cross sign shows the average over other input variables except $x_2$. The solid curve shows the true noise function $\sigma(\cdot)$.  Result for sSGV is obtained from the \texttt{GpGp} package. Result based on BART is obtained from the \texttt{BayesTree} package.}
\label{fig: 20D-illustration}
\end{center}
\end{figure}

\noindent \textbf{SMC Diagnostics:} The SMC diagnostic plots are given in Figure~\ref{fig: SMC diagnostics in 10-20D-illustration}. The plots only show ESS (effective sample size) and log evidence $\log Z$ across tree levels for the \texttt{ResTGP-WN} model, although they are updated at each SMC step. From these figures, the ESS exhibits high variation across tree levels. These decreases in ESS trigger resampling, after which the particle weights are reset and the ESS increases. Such behavior is expected when the posterior becomes concentrated on a relatively small number of tree decisions during the expansion process. The estimated marginal likelihood also stabilizes as the tree expansion proceeds.

\begin{figure}[!htpb]
\begin{center}
\begin{subfigure}[t]{1\textwidth}
\makebox[\textwidth][c]{\includegraphics[width=1\textwidth, height=0.25\textheight]{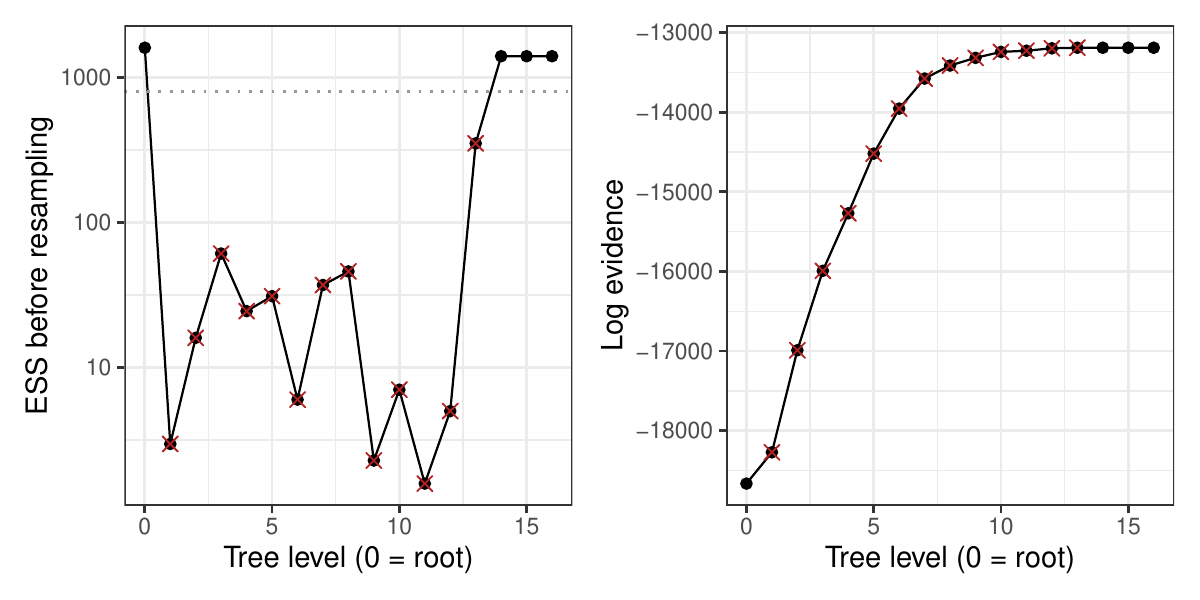}}
\caption{$d=10$.}
\end{subfigure}%

\begin{subfigure}[t]{1\textwidth}
\makebox[\textwidth][c]{\includegraphics[width=1\textwidth, height=0.25\textheight]{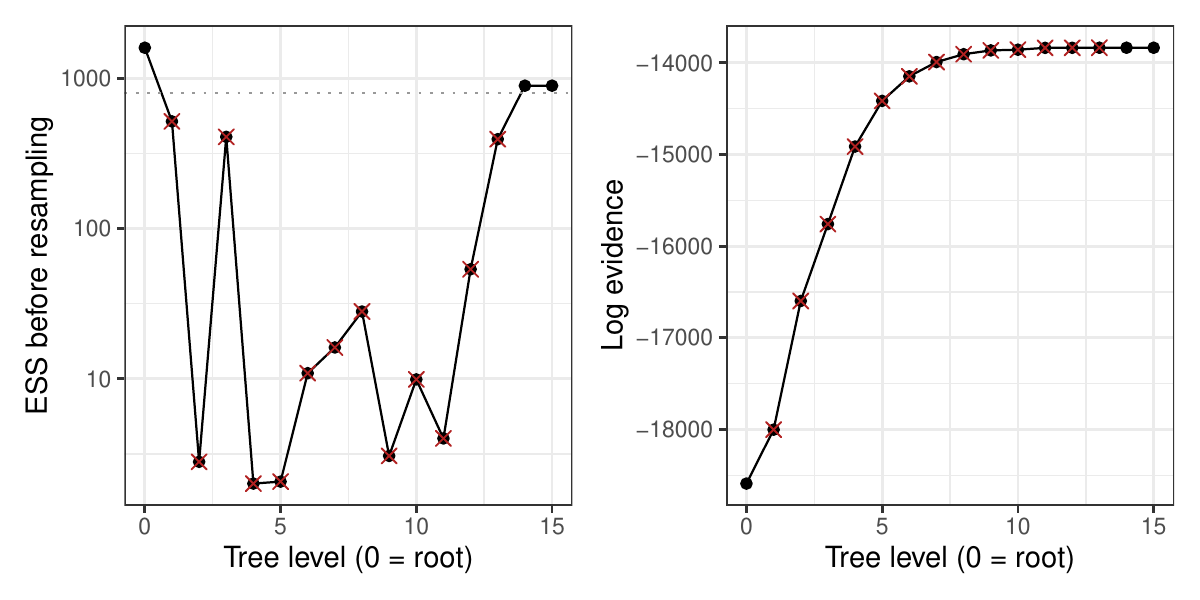}}
\caption{$d=20$}
\end{subfigure}
\caption{SMC diagnostics in \texttt{ResTGP-WN} with ESS and log evidence against tree levels. The cross sign indicates resampling step happened at a tree level.}
\label{fig: SMC diagnostics in 10-20D-illustration}
\end{center}
\end{figure}

\newpage

\section{Results in Real Application} \label{supp: real application} 

\subsection{Methods settings}

The configuration of each model under comparison is given below:
\begin{itemize}[itemsep=-2pt, topsep=0pt, partopsep=0pt]
\item sSGV: $r=30, 60$
\item BART: the number of trees is 200. The number of MCMC iterations is 12,000 with 5,000 for burn-in period. Default settings for hyperparameters are used. 
\item ResTGP-PP: The number of knots is $r=60$. Middle cut is used via the option \verb|cut_method = "middle"|. Default settings for hyperparameters are used. 
\item ResTGP-WN: The number of knots is $r=60$. Adaptive cut is used via the option 
\verb|cut_method = "uniform"|. Default settings for hyperparameters are used. 
\end{itemize}

\subsection{Exploratory analysis}

This section contains the description of inputs in the  storm surge application. 
%\subsection{Storm surge application} \label{supp:sec:storm surge}
\begin{table}[!htbp]
\centering
\normalsize
   \caption{Storm characteristics parameters.}
  {\resizebox{1.0\textwidth}{!}{%
  \setlength{\tabcolsep}{3.5em}
   \begin{tabular}{c l } 
   \toprule \noalign{\vskip 1.5pt} 
Input variables    &  Physical meaning \\ \noalign{\vskip 1.5pt} 
\midrule \noalign{\vskip 2.5pt} 
$\Delta P$  & central pressure deficit of the storm (mb) \\ \noalign{\vskip 2.5pt} 
$R_p$       & scale pressure radius in nautical miles \\ \noalign{\vskip 2.5pt} 
$V_f$       & storm's forward speed (m$/$s) \\ \noalign{\vskip 2.5pt} 
$\vartheta$    & storm's heading in degrees clockwise from north \\ \noalign{\vskip 2.5pt} 
$B$         & Holland's $B$ parameter (unitless) \\ \noalign{\vskip 2.5pt} 
$\mathbf{\Ell}$       & landfall location in latitude and longitude \\  \noalign{\vskip 2.5pt}  \bottomrule
   \end{tabular}%
   }}
   \label{table: input parameters}
\end{table}

\begin{figure}[!htpb]
\begin{center}
\makebox[\textwidth][c]{\includegraphics[width=1\textwidth, height=0.4\textheight]{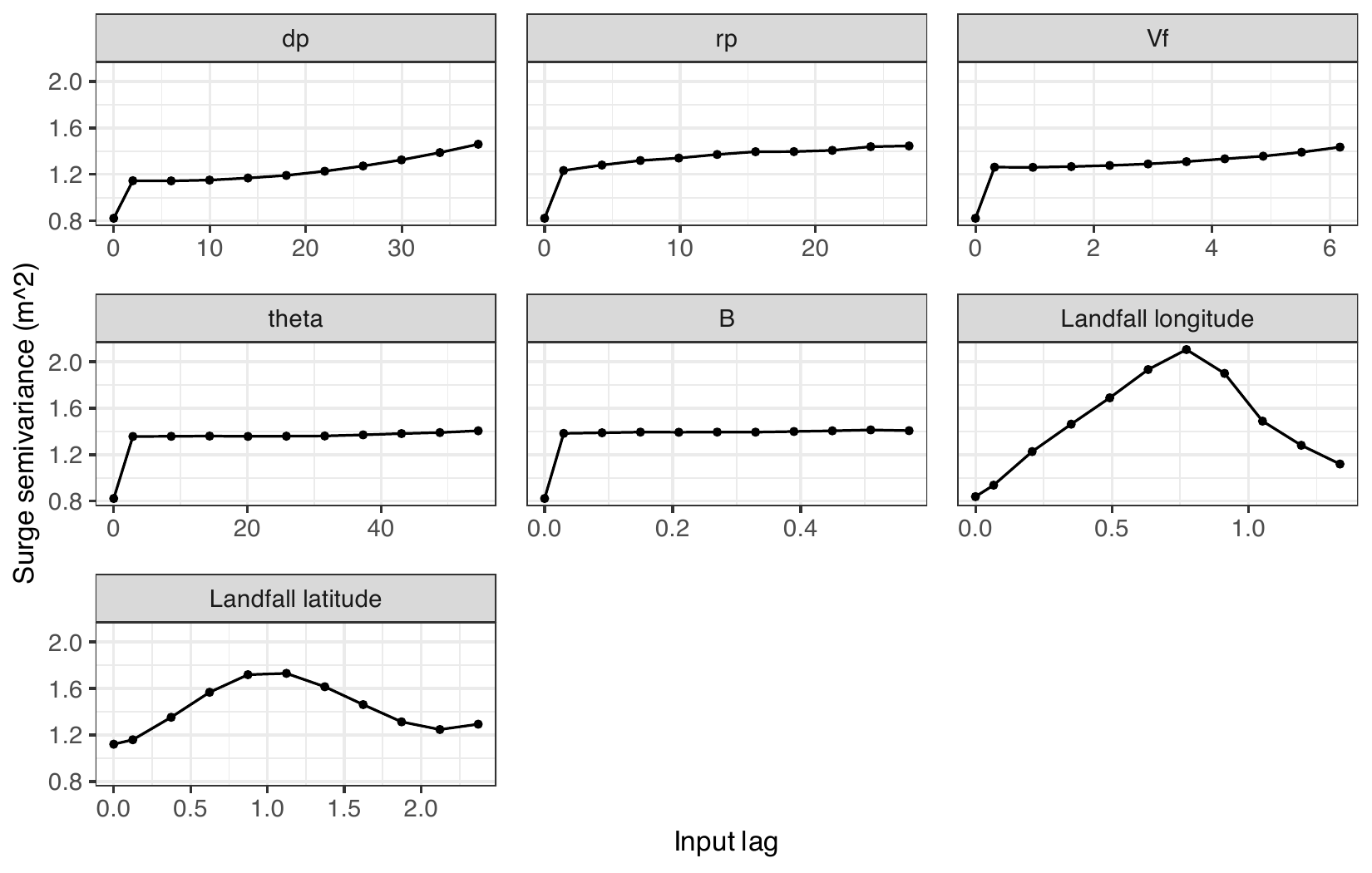}}
\caption{Directional empirical semivariograms along each input.}
\label{fig: semivariograms}
\end{center}
\end{figure}

\subsection{Additional results}

Figure~\ref{fig: app scatterplot} shows the predictive mean against held-out observations with 45 degree line. Figures~\ref{fig: app predictive SE sSGV},~\ref{fig: app predictive SE PP},~\ref{fig: app predictive SE WN} show the predictive standard deviations with respect to each input under sSGV, ResTGP-PP, and ResTGP-WN. Within each panel, the dots represent the predictive standard deviations at one variable when other variables take different values. 

\begin{figure}[!htpb]
\begin{center}
\makebox[\textwidth][c]{\includegraphics[width=1\textwidth, height=0.4\textheight]{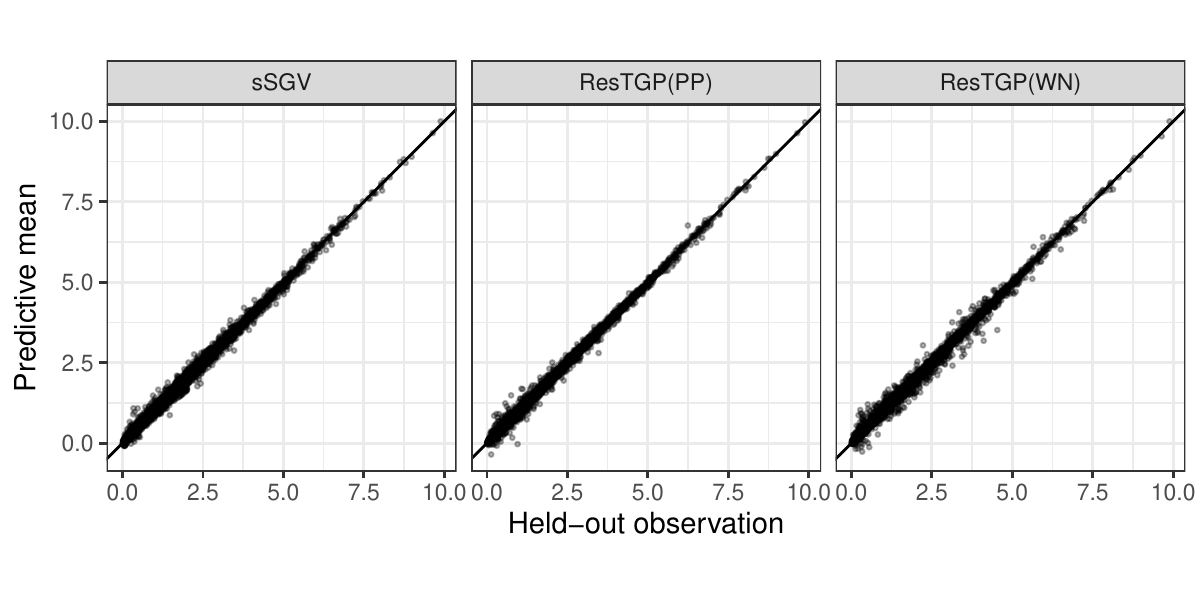}}
\caption{Scatterplot of predictive mean against held-out data under different models.}
\label{fig: app scatterplot}
\end{center}
\end{figure}

\begin{figure}[!htpb]
\begin{center}
\makebox[\textwidth][c]{\includegraphics[width=1\textwidth, height=0.4\textheight]{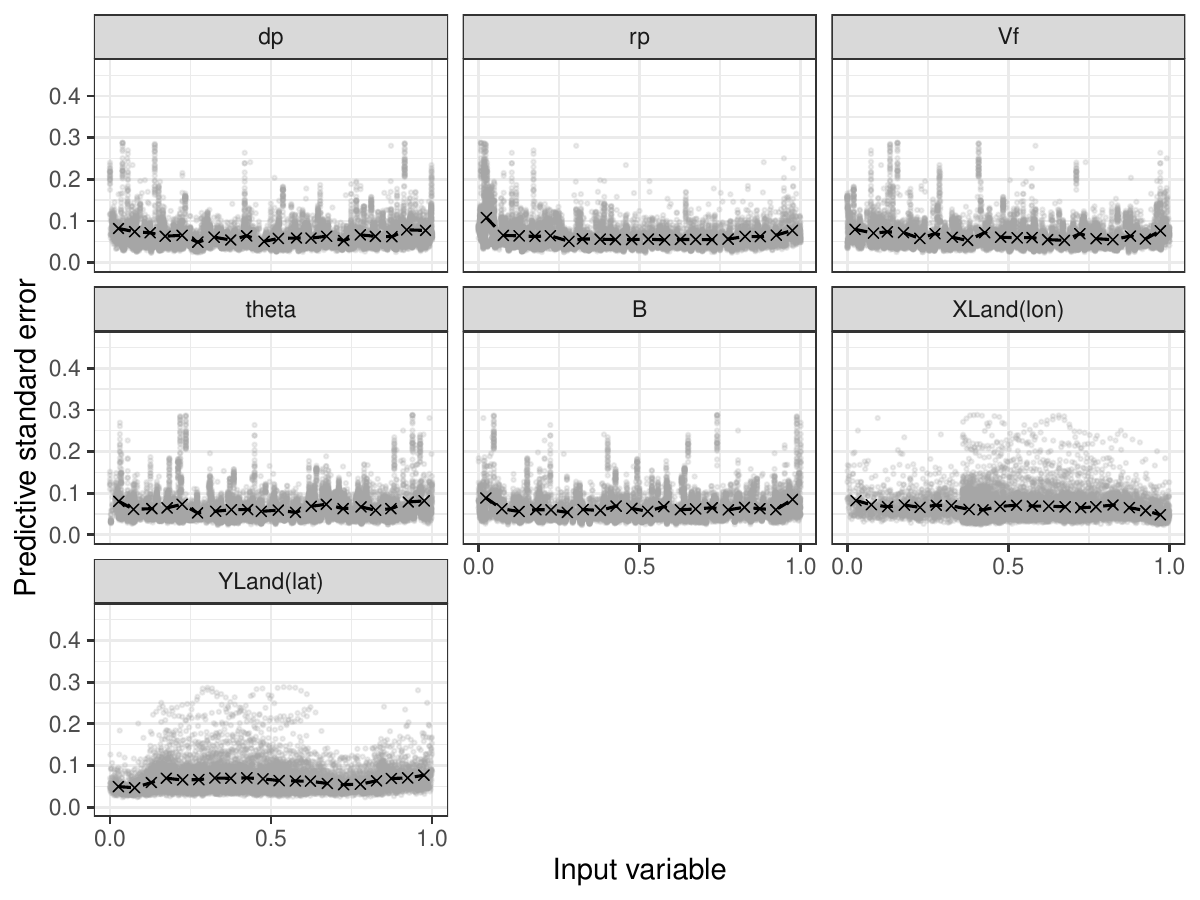}}
\caption{Predictive standard deviations against input variables under sSGV. Cross sign represents the averages of predictive standard deviations over other variables.}
\label{fig: app predictive SE sSGV}
\end{center}
\end{figure}

\begin{figure}[!htpb]
\begin{center}
\makebox[\textwidth][c]{\includegraphics[width=1\textwidth, height=0.4\textheight]{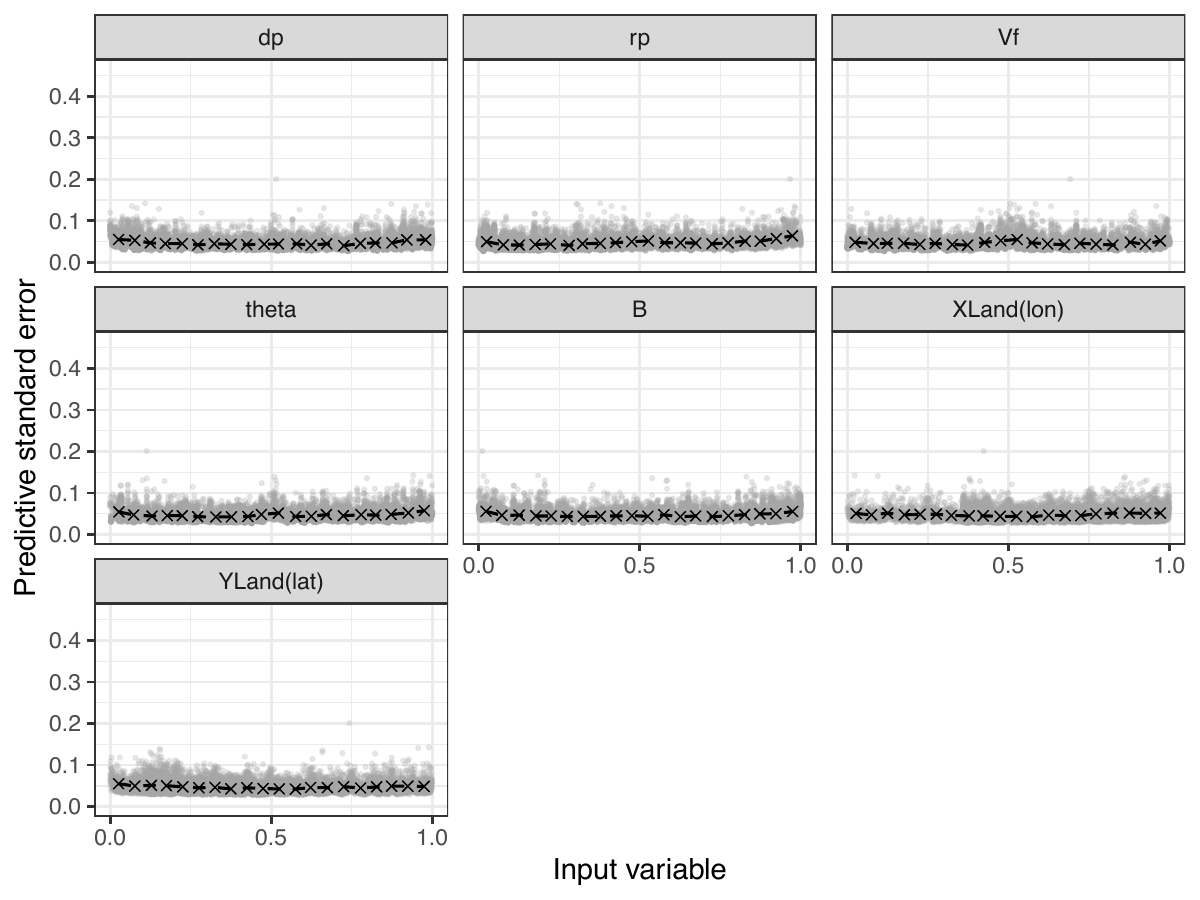}}
\caption{Predictive standard deviations against input variables under ResTGP-PP. Cross sign represents the averages of predictive standard deviations over other variables.}
\label{fig: app predictive SE PP}
\end{center}
\end{figure}

\begin{figure}[!htpb]
\begin{center}
\makebox[\textwidth][c]{\includegraphics[width=1\textwidth, height=0.4\textheight]{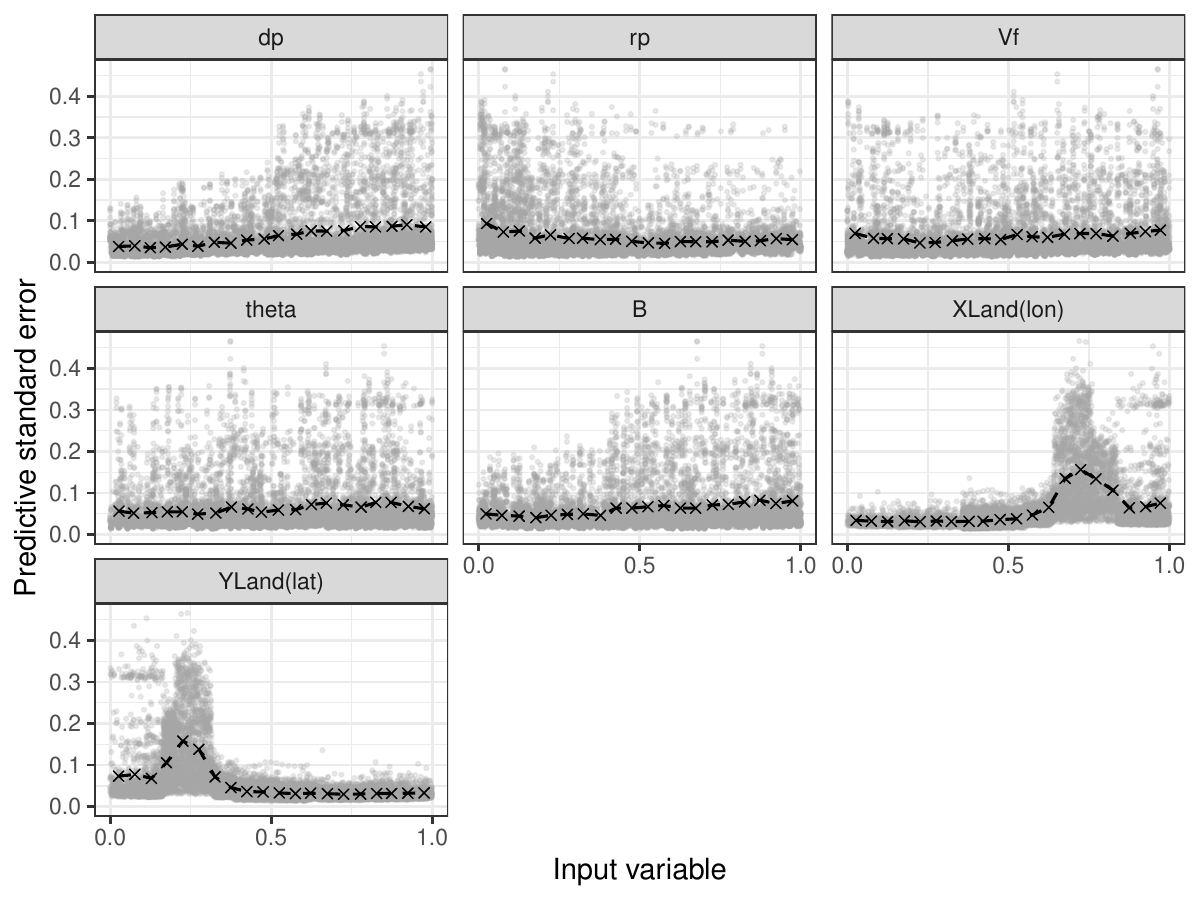}}
\caption{Predictive standard deviations against input variables under ResTGP-WN. Cross sign represents the averages of predictive standard deviations over other variables.}
\label{fig: app predictive SE WN}
\end{center}
\end{figure}

\newpage

\section{Multi-Scale Decomposition} \label{supp: ms decomposition} 
\subsection{Ancillary results} \label{supp: orthogonality for Gaussian rv}
Let $\bfY$ be an $n$-dimensional Gaussian random vector with mean vector $\bfmu$ and positive-definite covariance matrix $\bfSigma$. Let $\{\bfy_1, \bfy_2\}$ be a partition of $\bfY$ such that $\bfY=(\bfy_1^\top, \bfy_2^\top)^\top$, where $\bfy_1 \sim {N}_{r_1}(\bfmu_1, \bfSigma_{11})$, $\bfy_2 \sim {N}_{r_2}(\bfmu_2, \bfSigma_{22})$, $ \bfSigma_{12}:=\text{cov}(\bfy_1, \bfy_2)$, and $r_1 + r_2=n$. We define the quantities $\delta_{0,1}=\bfy_1$ and $\delta_{1,2}:=\bfy_2 - E[\bfy_2 \mid \bfy_1]$. One can easily check that the random vectors $\delta_{0,1}$ and $\delta_{1,2}$ are independent Gaussian random vectors. In fact, we have density factorization $p(\bfY) = p(\delta_{0,1}) \times p(\delta_{1,2})$. This allows one to shift the computation for the Gaussian density on an $n$-dimensional space to the computation on two lower-dimensional spaces. In fact, such orthogonal decomposition can be generalized to break the computation into computations on $k$ lower-dimensional spaces individually. 

\begin{lemma}\label{lem: OD for Gaussian rv}
Let $\bfY$ be an $n$-dimensional Gaussian random vector with mean vector $\bfmu$ and positive-definite covariance matrix $\bfSigma$. Suppose that $\{\bfy_1, \bfy_2, \ldots, \bfy_k\}$ is an arbitrary partition of $\bfY$ such that $\bfY = (\bfy_1^\top, \bfy_2^\top, \ldots, \bfy_k^\top)^\top$, where $\bfy_i\sim {N}_{r_i}(\bfmu_i, \bfSigma_{ii})$, $\bfSigma_{ij}:=\text{cov}(\bfy_i, \bfy_j)$, and $\sum_{i=1}^k r_i = n$. For $m=1, \ldots, k-1$, we define the following quantities:
\begin{align*}
\delta_{m, i} := \delta_{m-1, i} - E[\delta_{m-1, i} \mid \delta_{m-1, m}], \, i=m+1, \ldots, k,
\end{align*}
where $\delta_{0, \ell} = \bfy_{\ell}$ for $\ell=1, \ldots, k$. The following statements hold true:
\begin{itemize}[itemsep=-2pt, topsep=0pt, partopsep=0pt]
\item[(a)] For $m=1, \ldots, k-1$, the random vector $\delta_{m,i}$ is an $r_i$-dimensional Gaussian random vector with mean  zero 
and cross-covariance matrix given by 
\begin{align*}
\bfSigma_{i,j}^{[m]} :=\text{cov}(\delta_{m,i}, \delta_{m,j}) = \bfSigma_{i,j}^{[m-1]} - \bfSigma_{i, m}^{[m-1]} \biggl (\bfSigma_{m,m}^{[m-1]} \biggr)^{-1} \bfSigma_{m,j}^{[m-1]},\, i,j=m+1, \ldots, k,
\end{align*}
where $\bfSigma_{i, j}^{[0]}:=\text{cov}(\delta_{0,i}, \delta_{0, j})= \bfSigma_{i,j}$ for $j=1, \ldots, k$.  
\item[(b)] $\{\delta_{m,m+1}: m=0, \ldots, k-1\}$ are mutually independent.
\item[(c)] The density of $\bfY$ can be factorized as  
\begin{align*}
p(\bfY) = \prod_{m=1}^k p(\delta_{m-1, m}) =  p(\delta_{0, 1}) \times p(\delta_{1,2}) \times \cdots \times p(\delta_{k-1,k}).
\end{align*}
\end{itemize}
\end{lemma}

\begin{proof}
Note that $(c)$ implies $(b)$, so it suffices to prove $(a)$ and $(c)$. 
\begin{itemize}
\item[(a)] Since linear combinations of Gaussian random vectors are still Gaussian, it suffices to compute the mean and covariance matrix for $\delta_{m,i}$ with $m=1, \ldots, k-1$ and $i=m+1, \ldots, k$. Taking the expectation for $\delta_{m,i}$ immediately indicates that $E[\delta_{m,i}]=\bfzero$. The covariance expression for $\delta_{m,i}$ follows by direct calculation.
\item[(c)] Note that $\delta_{1, i} = \delta_{0,i} - E[\delta_{0,i} \mid \delta_{0,1}]$ for $i=2, \ldots, k$. Let 
$$ \delta_{1,2:k}:=\{\delta_{1,2}, \ldots, \delta_{1,k}\}.$$ Then it follows that 
\begin{align*}
p(\bfY) = p(\delta_{0,1}) \times p(\delta_{1, 2:k}). 
\end{align*}
Since $\delta_{2, i} = \delta_{1,i} - E[\delta_{1,i} \mid \delta_{1,2}]$ for $i=3, \ldots, k$. Let $\delta_{2, 3:k}=\{\delta_{2,3}, \ldots, \delta_{2, k}\}$. We have $p(\delta_{1,2:k}) = p(\delta_{1,2}) \times p(\delta_{2, 3:k})$. Thus, 
\begin{align*}
p(\bfY) = p(\delta_{0,1}) \times p(\delta_{1,2}) \times p(\delta_{2,3:k}).
\end{align*}
Repeatedly applying this procedure $k-2$ times yields the desired result. 
\end{itemize}
\end{proof} 

The key idea behind this novel decomposition is to successively regress one subset of data over the other subset of data, based on which, we can obtain residuals over the remaining data. These residuals can be treated as ``new" data.  We then again choose a subset of data from this new data to obtain a new set of residuals. Recursively applying this procedure provides a new way to compute any finite-dimensional high-dimensional Gaussian probability density function. Sometimes, simulating from multivariate Gaussian density is also important. The following lemma shows how to simulate from multivariate Gaussian density efficiently. 
\begin{corollary} \label{cor: simulation} 
Let $\delta_{m,m+1}$ be a independent realization from ${N}(\bfzero, \bfSigma_{m+1,m+1}^{[m]})$ for $m= 1, \ldots k-1$ and let $\delta_{0,1}$ be an realization from  ${N}(\bfmu_1, \bfSigma_{11})$ independent of $\{\delta_{m,m+1}: m=1, \ldots, k-1\}$. Define the bottom-up recursion formula: 
$$\delta_{m-1, i}=\delta_{m,i} + \bfSigma_{i,m}^{[m-1]} \biggl (\bfSigma_{m,m}^{[m-1]} \biggr)^{-1} \delta_{m-1,m}$$
for $m=k-1, \ldots, 2$ and $i=m+1, \ldots, k$. Then a realization from $\bfY=(\bfy_1^\top, \ldots, \bfy_k^\top)^\top$
can be obtained by setting $\bfy_1=\delta_{0,1}$ and $\bfy_i = \delta_{1, i} + \mu_i + \bfSigma_{i,1} \bfSigma^{-1}_{11} (\bfy_1-\mu_1)$ for $i=2, \ldots, k$.
\end{corollary}

This corollary follows immediately from Lemma~\ref{lem: OD for Gaussian rv}. While the above result is for finite-dimensional Gaussian random vector, as shown below, the orthogonal decomposition can also be generalized to Gaussian processes.

\subsection{Proof of Theorem~1} \label{supp: orthogonal decomposition}
\begin{proof}
We prove by induction.
First, when $k=1$, the result follows from Lemma~\ref{lem: OD for Gaussian rv}. Second, we assume that when $k=m$, the result holds true. That is, let $\{\calR^0, \ldots, \calR^{m-1}, \calA_{m}\}$ be a partition of $\calA^0$ in the sense of observation locations. We have 
\begin{align} \label{eqn: induction k=m}
p(y(\calA^0)) = p(f_0(\calR^0)) \times p(f_1(\calR^1)) \times \cdots \times p(f_{m-1}(\calR^{m-1})) \times p(\xi_{m}(\calA^m)).
\end{align}
Third, we need to show that the result holds true for $k=m+1$. 

By assumption, observations in $\{\calR^m, \calA^{m+1}\}$ is a partition of observations in $\calA^m$. By assumption, 
\begin{gather*}
f_m(\cdot) = E\{\xi_m(\cdot) \mid \xi_m(\calR^m) \}, \quad
\xi_{m+1}(\cdot) = \xi_m(\cdot) - f_m(\cdot).
\end{gather*}
By the property of spatial ``BLUP'' \citep[][Section 1.5]{Stein1999}, the process $f_m(\cdot)$ and the process $\xi_{m+1}(\cdot)$ are independent. As $\calA^m = \calR^m \cup \calA^{m+1}$, it follows immediately that 
\begin{align*}
p(\xi_m(\calA^m)) = p(f_m(\calR^m)) \times p(\xi_{m+1}(\calA^{m+1})). 
\end{align*}
According to Equation~\eqref{eqn: induction k=m}, we have 
\begin{align*}
p(y(\calA^0))& = p(f_0(\calR^0)) \times p(f_1(\calR^1)) \times \cdots \times p(f_{m-1}(\calR^{m-1})) \times p(\xi_{m}(\calA^m)) \\
			   & = p(f_0(\calR^0)) \times p(f_1(\calR^1)) \times \cdots \times p(f_{m-1}(\calR^{m-1})) \times p(f_m(\calR^m)) \times p(\xi_{m+1}(\calA^{m+1})).
\end{align*} 
\end{proof}

\subsection{Proof of Theorem~2} \label{supp: marginal likelihood}
\begin{proof}

 The proof of this uses the Bayes theorem for fixed tree. When the tree is given, if $\calA$ is a leaf node, the likelihood is given by $q_0(\calA)$. As the leaf model includes predictive process (PP) and white-noise models, $q_0(\calA)$ will be computed under these leaf models respectively. 
 
For the PP leaf model, let $\calR_\calA$ denote a set of knots and $C_{\calA}$ denote the covariance function at $\calA$. Then the covariance matrix of $\xi(\calA)$ over $\calA$ is 
\begin{align*}
\tilde{\bfC}_{\calA}:=C_{\calA}(\calA, \calR_\calA) C_{\calA}(\calR_{\calA}, \calR_{\calA})^{-1} C_{\calA}(\calA, \calR_\calA)^\top + V_{\calA}
\end{align*}
where $V_{\calA}$ is the diagonal matrix with diagonal entry $$V_{\calA}(\bfx, \bfx) := C_{\calA}(\bfx, \bfx) - C_\calA(\calR_\calA, \bfx)^\top C_{\calA}(\calR_\calA, \calR_\calA)^{-1} C_\calA(\calR_\calA, \bfx), \bfx \in \calA.$$ This is known as the modified predictive process construction \citep{Finley2009}. Then $q_0$ can be computed via
\begin{align*}
q_0(\calA) = (2\pi)^{-\frac{n_\calA}{2}} |\tilde{\bfC}_{\calA}|^{-1/2} \exp\left\{- \frac{\xi(\calA)^\top \tilde{\bfC}^{-1}_{\calA} \xi(\calA)}{2} \right\},
\end{align*}
where $n_{\calA}$ denotes the number of data points in $\calA$. 

For the WN leaf model, $\xi(\calA)$ is modeled as iid observations from $\xi_i(\calA) \overset{iid}{\sim} {N}(\mu_A, \kappa^2_\calA)$ independently at each node $\calA$. Integrating out parameters $(\mu_\calA, \kappa_\calA)$ with respect to the prior $\mu_{\calA} | \kappa^2_{\calA} \sim {N}(\mu_{0}, \kappa^2_\calA/\kappa_0), \kappa^2_{\calA} \sim IG(a_0, b_0)$ yields the marginal likelihood 
\begin{align*}
q_0(\calA) &= \int (2\pi)^{-n_{\calA}/2} \{\kappa_\calA^2\}^{-n_{\calA}/2} \exp\left\{- \frac{( \xi(\calA) - \mu_\calA)^\top (\xi(\calA) - \mu_\calA)}{2\kappa_\calA^2} \right\} \\ 
&\quad \times N(\mu_\calA| \mu_0, \kappa^2_\calA/\kappa_0) IG(\kappa^2_\calA| a_0, b_0) \mathrm{d} (\mu_\calA, \kappa_\calA^2) \\
&= (2\pi)^{-n_\calA/2} \left(\frac{\kappa_0}{\tilde{\kappa}}\right)^{1/2} \frac{b_0^{a_0}}{\tilde{b}^{\tilde{a}}} \frac{\Gamma(\tilde{a})}{\Gamma(a_0)}
\end{align*}
where 
\begin{gather*}
\tilde{a}:=a_0 + n_\calA/2, \tilde{\kappa}:=\kappa_0 + n_\calA \\
\tilde{\mu}:= \tilde{\kappa}^{-1} \{\kappa_0 \mu_0 + n_\calA \bar{\xi}(\calA)\} \\
\tilde{b}:= b_0 + \frac{1}{2} \sum_{i=1}^{n_\calA} (\xi_{i,\calA} - \bar{\xi}(\calA))^2 + \frac{\kappa_0 n_\calA}{2\tilde{\kappa}} (\bar{\xi}(\calA) - \mu_0)^2
\end{gather*}
with default values $(\mu_0, \kappa_0, a_0, b_0)=(0,1,2,0.01)$. 

When $\calA$ is not leaf, it is then cut into two child nodes $\calA_\lchild^{J, L}, \calA_\rchild^{J,L}$ according to $J$ and $L$. The likelihood under $S(\calA)=0$ becomes $q_\calA(\calR_\calA) q(\calA_\lchild^{J,L}) q(\calA_\rchild^{J,L})$, where $q_{\calA}(\calR_\calA)$ denotes the likelihood over knots $\calR_{\calA}$. 

Combining these two situations yields 
\begin{align*}
q(\calA) = S(\calA) q_0(\calA) + (1-S(\calA)) q_{\calA}(\calR_\calA) q(\calA_\lchild^{J,L}) q(\calA_\rchild^{J,L}). 
\end{align*}
Then integrating out $S, J, L$ with respect to their prior yields the integrated likelihood 
\begin{align*}
\Phi(\calA) &= \int q(\calA) \mathbb{P}(\mathrm{d} \calT) \\
                  &=\rho(\calA) q_0(\calA) +(1-\rho(\calA)) \biggr\{ \sum_{j=1}^d \lambda_j q_{\calA}(\calR_{\calA}) 
\times \int \Phi(\calA^{j,c}_{\lchild}) \cdot \Phi(\calA^{j,c}_{\rchild}) F_\calA(\mathrm{d} c | j)  \biggr\},
\end{align*}
as desired. 

Once this is computed, applying Bayes theorem can yield the posterior probabilities for tree decisions $(S, J, L)$ at each node: 
\begin{align*}
\begin{split}
    \tilde{\rho}(\calA)&:=P(S(\calA)=1) = \frac{\rho(\calA) q_0(\calA)}{\Phi(\calA)}, \\
    \tilde{\lambda}_{j}(\calA) &:=P(J(\calA)=j| S(\calA)=0) \\
    & \propto \lambda_j(\calA) q_{\calA}(\calR_{\calA})\times \int \Phi(\calA_{\lchild}^{j,c}) \Phi(\calA_{\rchild}^{j,c}) F_{\calA}(\mathrm{d} c | j),  j=1, \ldots, d, \\
    \tilde{F}_{\calA}(\mathrm{d} c \mid j) &:= P(L_\calA=\mathrm{d} c| S(\calA)=0, J(\calA)=j) = \frac{\Phi(\calA_{\lchild}^{j,c} ) \Phi(\calA_{\rchild}^{j,c} )}{\int \Phi(\calA_{\lchild}^{j,u} ) \Phi(\calA_{\rchild}^{j,u} ) F_{\calA}(\mathrm{d} u | j)} \cdot F_{\calA}(\mathrm{d} c | j)
\end{split}
\end{align*}

\end{proof}

\subsection{Proof of Theorem~3} \label{supp: residual predictive formula}
\begin{proof}

Let $\calB(\bfx_0)$ be the child of $\calA$ containing $\bfx_0$. The arrival residual satisfies 
\begin{align*}
\xi_\calA(\bfx_0) = y(\bfx_0) - \sum_{\calB \in \downarrow\!\calA} h_{\calB}(\bfx_0).
\end{align*}

Denote by 
\begin{align*}
q(z|\calA) = p_{\xi_{\calA}(\bfx_0)}(z \mid \bfy, \calT, \bftheta, v), \quad z\in \mathbb{R}
\end{align*}
the conditional density of the arrival residual at the input location $\bfx_0$ evaluated at $z$. 

If $\calA\in \calL$, $q(z|\calA) = q_0(z | \calA)$; otherwise $\xi_{\calA}(\bfx_0) = h_\calA(\bfx_0) + \xi_\calB(\bfx_0)$. Since $\bfy, \calT, \bftheta, v$ are given, $h_{\calA}(\bfx_0)$ is fixed and hence $q(z|\calA) = q(z - h_\calA(\bfx_0) \mid \calB)$. It follows that 
\begin{align} \label{eqn: predictive mean at node A} 
\phi(\bfx_0; \calA) &=  E\{h_\calA(\bfx_0) + \xi_\calB(\bfx_0)\} = h_\calA(\bfx_0) + \phi(\bfx_0; \calB), \\ \label{eqn: predictive variance at node A} 
V(\bfx_0; \calA) &= \text{Var}\{h_\calA(\bfx_0) + \xi_\calB(\bfx_0)\} = V(\bfx_0; \calB).
\end{align}
Since the above formula is a recursion from root to leaf nodes that contain $\bfx_0$, we have 
\begin{align*}
y(\bfx_0) =  h_{\calT}(\bfx_0) + \xi_{\calA_*}(\bfx_0)
\end{align*}

\noindent \textbf{\textit{PP leaf model}}. $\xi_{\calA}(\bfx_0) \sim N(\phi_0(\bfx_0; \calA), V_0(\bfx_0; \calA))$, where
\begin{gather} \label{eqn: pp mean}  
\phi_0(\bfx_0; \calA) =  h_\calA(\bfx_0) \\ \label{eqn: pp variance} 
V_0(\bfx_0; \calA) = C_{\calA}(\bfx_0, \bfx_0)  - C_{\calA}(\calR_\calA, \bfx_0)^\top {C}_{\calA}^{-1}(\calR_\calA, \calR_\calA) C_{\calA}(\calR_\calA, \bfx_0). 
\end{gather}
Thus, combining Equations~\eqref{eqn: predictive mean at node A}--\eqref{eqn: pp variance} yields the desired result under PP leaf model. 

\noindent \textbf{\textit{WN leaf model}}. $\xi_{\calA}(\bfx_0) \sim N(\mu_\calA, \kappa^2_\calA )$ with prior
\begin{align*}
\mu_\calA | \kappa^2_\calA \sim N(\mu_0, \kappa_\calA^2/\kappa_0) \\
\kappa_\calA^2 \sim IG(a_0, b_0). 
\end{align*}

Using the notations from Theorem 2, standard calculation yields that the posterior predictive distribution of $\xi_{\calA}(\bfx_0)$ is a Student-$t$ distribution with degrees of freedom $2\tilde{a}$ whose mean and variances 
\begin{gather*}
\phi_0(\bfx_0; \calA) = \tilde{\mu} = \tilde{\kappa}^{-1} \{\kappa_0 \mu_0 + n_\calA \bar{\xi}(\calA)\} \\
V_0(\bfx_0; \calA) = \frac{\tilde{b} ( 1 + \tilde{\kappa}^{-1})}{\tilde{a} - 1}. 
\end{gather*}
Adding $h_{\calT}(\bfx_0)$ yields the desired result. 

\end{proof}

%%%%%%%%%%%%%%%%%%%%%%%%%%%%%%%%%%%%%%%%%%%%%%%%%%%%%%%%
%%%%%%%%%%%%%%%%%%%%%%%%%%%%%%%%%%%%%%%%%%%%%%%%%%%%%%%%
%%%%%%%%%%%%%%%%%%%%%%%%%%%%%%%%%%%%%%%%%%%%%%%%%%%%%%%%

\section{Theoretical Results on Posterior Consistency} \label{supp: posterior consistency}

\subsection{Model properties}  \label{supp: ancillary results}
This section presents the model bounds and continuity used to verify C1-C3. 

\begin{prop} \label{prop: residual process bound} 

For a tree $\calT$ with  $K<\infty$  leaves in the residual tree GP prior $g(\cdot)$, write $r_{\calA} :=|\calR_\calA|, \tilde{r} = r+1$. Under Assumptions 3-4, the following statements hold.
\begin{itemize}
\item[(i)] Conditional on $(\calT, \bftheta)$, the random function $g$ has the equivalent representation 
\begin{align}\label{eqn: model form} 
g(\bfx) = \mu_0 + b_{\calT, \bftheta}^\top Z, \quad Z \sim N_{q}(\mathbf{0}, \mathbf{I}_{q}),
\end{align}
where $b_{\calT,\bftheta}: \calD \mapsto \mathbb{R}^{q}$ is some deterministic function and 
\begin{align*}
q = \sum_{\calA \in \calI} r_\calA + K \leq r(K-1) + K \leq \tilde{r} K.
\end{align*}
\item[(ii)] For fixed $(\calT, \bftheta)$ and fixed vectors $\bfz \in \mathbb{R}^{q}$, define $g_{\calT, \bftheta, z}(\bfx) = \mu_0 + b_{\calT, \bftheta}(\bfx)^\top \bfz$. Write $B:=(\sigma_+^2 + \kappa_+^2(1+\kappa_0^{-1}))^{1/2}$. Then 
\begin{align*}
 \sup_{\calT, \bftheta, \bfx} \| b_{\calT, \bftheta}(\bfx)\|_2 \leq B, \quad 
\| g_{T,\bftheta, \bfz} - g_{T, \bftheta, \bfz'}\|_{\infty} \leq B \| \bfz - \bfz'\|_2.
\end{align*}
\end{itemize}

\end{prop} 

\begin{proof}
\begin{itemize}
\item[(i)]
Write $L(\calR_\calA)$ as the Cholesky factor of $C_{\calA}(\calR_\calA, \calR_\calA)$ such that $L(\calR_\calA) L(\calR_\calA)^\top = C_{\calA}(\calR_\calA, \calR_\calA)$. 

For internal node $\calA$,  let $Z_{\calA} \sim {N}_{r_A}(\mathbf{0}, \mathbf{I})$. Then the predictive process $f_\calA$ can be written as $f_{\calA}(\bfx) = C_{\calA}(\calR_\calA, \bfx)^\top L_{\calR_\calA}^{-\top} Z_{\calA} = b_\calA(\bfx)^\top Z_\calA$, where $b_\calA(\bfx) := L(\calR_\calA)^{-1}C_{\calA}(\calR_\calA, \bfx) \mathds{1}_\calA(\bfx)$ for $\calA \in \calI$; For leaf node $\calA$,  let $Z_{\calA} \sim N(0,1)$. Then $\delta_{\calA} = \mu_0  + \kappa_\calA\sqrt{1+\kappa_0^{-1}} Z_{\calA} = \mu_0 + b_\calA(\bfx) Z_\calA$, where $b_\calA(\bfx) := \kappa_\calA\sqrt{1+\kappa_0^{-1}} \mathds{1}_\calA(\bfx)$ for $\calA \in \calL$.

Let $\calA_1, \ldots, \calA_J$ denote all the internal nodes of the tree $\calT$ and $\calB_1, \ldots, \calB_K$ denote all the leaf nodes of the tree.  Define 
\begin{gather*}
b_{\calT, \bftheta}(\bfx) := [b_{\calA_1}(\bfx), \ldots, b_{\calA_J}(\bfx), \kappa_{\calB_1}\sqrt{1+\kappa_0^{-1}} \mathds{1}_{\calB_1}(\bfx), \ldots, \kappa_{\calB_K}\sqrt{1+\kappa_0^{-1}} \mathds{1}_{\calB_K}(\bfx)]^\top, \\
Z := (Z_{\calA_1}, \ldots, Z_{\calA_J}, Z_{\calB_1}, \ldots, Z_{\calB_K})^\top.
\end{gather*}

Then the residual tree GP prior $g$ can be equivalently represented in the following form:
\begin{align*}
g(\bfx): = \mu_0 + b_{\calT,\bftheta}(\bfx)^\top Z, \quad Z \sim \mathcal{N}_{q}(\mathbf{0}, \mathbf{I}), 
\end{align*}
where $q = \sum_{\calA \in \calI} r_A + K \leq (r+1) K.$
\item[(ii)]
It is immediate that 
$$\|b_{\calT, \bftheta} (\bfx) \|_2^2 = \sum_{\substack{\calA \in \calI  \\ \bfx \in \calA} } C^*_{\calA}(\bfx, \bfx) + \kappa^2_{\calB} (1+\kappa_0^{-1}) \leq C_0(\bfx, \bfx) + \kappa_\calB^2(1+\kappa_0^{-1}) \leq B^2$$
where $\calB$ denotes the leaf node containing $\bfx$. 

For every deterministic $\bfz,\bfz' \in \mathbb{R}^q$, $g_{\calT,\bftheta}$ is a realized quantity for fixed $\calT, \bftheta$.  Then for every $\bfx\in \scrD$, it follows from Cauchy-Schwarz inequality that 
\begin{align*}
| b_{\calT,\bftheta}(\bfx)^\top  \bfz - b_{\calT, \bftheta}(\bfx)^\top \bfz'|^2 &= \left\| \sum_{j=1}^q b_j(\bfx) (z_j -z_j') \right\|^2 \\
&\leq\left( \sum_{j=1}^q b_j(\bfx)^2 \right)  \left( \sum_{j=1}^q (z_j - z_j')^2 \right) \\
&\leq B^2 \| \bfz - \bfz'\|_2^2
\end{align*}
where $b_j(\bfx), z_j, z_j'$ are components of $b_{\calT,\bftheta}, \bfz, \bfz'$. Taking square root and then the supremum over $\bfx$ yields 
\begin{align*}% \label{eqn: continuity of g}
\| g_{\calT,\bftheta, \bfz} - g_{\calT,\bftheta, \bfz'} \|_{\infty} \leq B \| \bfz - \bfz'\|_2.
\end{align*}
\end{itemize}

\end{proof} 

\begin{prop} \label{prop: tree size} 
Under the tree prior $\calT\sim \pi$, the tree is finite almost surely, and there exists constants $C_{\mathcal T},c_{\mathcal T}>0$ depending only on $\rho_0,\beta_0$ such that
\[
\mathbb P(K>k)\leq C_{\mathcal T}e^{-c_{\mathcal T}k},\qquad k\geq0.
\]
\end{prop} 

\begin{proof}
Under the Bayesian CART tree prior, write $s_m=\rho_0(1+m)^{-\beta_0}$ with $m$ denoting the depth of the tree. Choose $m_0$ so that $s_m\leq1/4$ for $m\geq m_0$. In breadth-first order let $B_i$ be the split indicator, set to zero after extinction, and $\mathcal F_{i-1}$ the preceding history. At most $2^{m_0}-1$ nodes have depth below $m_0$, so
\[
\mathbb P(B_i=1\mid\mathcal F_{i-1})\leq\bar p_i
:=\mathds{1}_{\{i<2^{m_0}\}}+\tfrac14\mathds{1}_{\{i\geq2^{m_0}\}},\qquad
\mathbb E(e^{2cB_i}\mid\mathcal F_{i-1})\leq e^{(e^{2c}-1)\bar p_i}.
\]
For the total node count $N$, the queue identity and the conditional Chernoff argument of \citet[Lemma 5.1]{Rockova2019} give, for integer $u\geq0$,
\[
\begin{aligned}
\{N>u\}&\subset\{2\textstyle\sum_{i=1}^u B_i\geq u\},\\
\mathbb P(N>u)&\leq e^{-cu}\mathbb E e^{2c\sum_{i=1}^uB_i}
\leq\exp\{-cu+(e^{2c}-1)(2^{m_0}-1+u/4)\}.
\end{aligned}
\]
Taking $c=(\log2)/2$ and $c_{\mathcal T}=(\log2)/2-1/4>0$ yields $\mathbb P(N>u)\leq e^{2^{m_0}-1-c_{\mathcal T}u}$. Hence $N<\infty$ almost surely, $N=2K-1$, and
\[
\mathbb P(K>k)\leq\mathbb P(N>\lfloor k\rfloor)
\leq e^{2^{m_0}-1+c_{\mathcal T}}e^{-c_{\mathcal T}k},\qquad k\geq0,
\]
as desired. 
\end{proof}

\begin{prop} \label{prop: covariance continuity} 
For any fixed $\calT, \bftheta'$ in Assumption~4, 
\begin{equation}\label{eqn: cov lipschitz}
\sup_\bfx\|b_{\mathcal T,\bftheta'}(\bfx)-b_{\mathcal T,\bftheta}(\bfx)\|_2
\leq C\sqrt{K(\mathcal T)}\,\|\bftheta'-\bftheta\|_2, 
\end{equation}
for any fixed $\bftheta'$ in the parameter space. Here $C$ depends only on $d,\nu,r,\sigma_+$
\end{prop} 
\begin{proof}
It is well-known \citep{Stein1999} that the spectral density of an isotropic Mat\'ern covariance of the form $\calM_{\nu}(h) = \sigma^2 2^{1-\nu}/\Gamma(\nu) h^{\nu} \calK_{\nu}(h)$ is  
\begin{align*}
s_0(\omega) = \frac{\sigma^2 \Gamma(\nu+d/2)}{\pi^{d/2}\Gamma(\nu)}
(1+\|\omega\|_2^2)^{-\nu-d/2}, \, \omega \in \mathbb{R}^d. 
\end{align*} 
For stationary input-rescaled Mat\'ern covariance $C_{\bftheta}$, write $D_{\ell}:=\diag\{\ell_1, \ldots, \ell_d\}$. Then $C_{\bftheta} = \calM_{\nu}(D_{\ell}^{-1} h)$. Applying the change-of-variables technique yields  the spectral density of $C_{\bftheta}$ 
\begin{align*}
s_{\bftheta}(\omega)&=
\frac{\sigma^2|D_\ell|\Gamma(\nu+d/2)}{\pi^{d/2}\Gamma(\nu)}
(1+\|D_\ell\omega\|_2^2)^{-\nu-d/2}, \, \omega \in \mathbb{R}^d. 
\end{align*}

Recall that $\bftheta:=(\theta_0, \theta_1, \ldots, \theta_d)=(\log \sigma^2, \log \ell_1, \ldots, \log \ell_d)$. Write $\partial_{\theta} f : = \partial f / \partial \theta$ for any function of $s$. Taking derivative of $s_{\bftheta}$ with respect to $\bftheta$ yields 
\begin{align*}
\partial_{\theta_0}\log s_{\bftheta} =1,\qquad
\partial_{\theta_j}\log s_{\bftheta} 
=1-(2\nu+d)\frac{\ell_j^2\omega_j^2}{1+\sum_{k=1}^d \ell_k^2\omega_k^2}, \, j=1, \ldots, d.
\end{align*} 
Thus $|\partial_{\theta_j}\log s_{\bftheta}(\omega)|\leq 1+2\nu+d\leq C$ for $0\leq j\leq d$, uniformly in $\bftheta,\omega$ for some fixed constant $C$. For $h\in\mathbb R^{d+1}$, applying the chain rule and fundamental theorem of calculus yields 
\begin{align*}
\left|\log\frac{s_{\bftheta+h}(\omega)}{s_{\bftheta}(\omega)}\right|
=\left|\int_0^1\sum_{j=0}^d h_j\partial_{\theta_j}\log s_{\bftheta+th}(\omega)\,dt\right|
\leq C\|h\|_1.
\end{align*}
Exponentiating and using $\partial_{\theta_j}s=s\,\partial_{\theta_j}\log s$ yields
\[
s_{\bftheta+h}\leq e^{C\|h\|_1}s_{\bftheta},\qquad
|\partial_{\theta_j}s_{\bftheta+h}|\leq Ce^{C\|h\|_1}s_{\bftheta}.
\]
For the $j$th coordinate vector $e_j$ and $0<|u|\leq1$, applying the fundamental theorem of calculus gives
\begin{align*}
\left|\frac{s_{\bftheta+ue_j}-s_{\bftheta}}{u}\right|
=\left|\int_0^1\partial_{\theta_j}s_{\bftheta+tue_j}\,dt\right|
\leq Ce^C s_{\bftheta}.
\end{align*}
By the spectral representation of the $C_{\bftheta}$ and dominated convergence theorem, $C_{\bftheta}$ is differentiable in $\bftheta$. By continuity of the derivatives of the spectral density and the above bound, $C_{\bftheta}$ has  continuous derivatives in $\bftheta$.

Next, we show that $\partial_{\theta_j} b_{\calT, \bftheta}(\cdot)$ is also bounded. 

Write $\mathbf K_{ij}=C_{\bftheta}(\bfx_i,\bfx_j)$ and $\bfK=[K_{ij}]_{i,j=1,\ldots, k}$ for a set of fixed $k$ locations $\bfx_1, \ldots, \bfx_k$. Denote the matrix derivative by $\dot{\mathbf K}=\partial_{\theta_j}\mathbf K$. For fixed vectors $\bfa_1,\bfa_2\in\mathbb R^k$, with components $a_{1,i},a_{2,i}, i=1,\ldots, k$, 
\[
\bfa_1^\top\dot{\mathbf K}\bfa_2
=\sum_{i,j=1}^k \bfa_{1,i}\bfa_{2,j}\,\partial_{\theta_j}C_{0,\bftheta}(\bfx_i,\bfx_j)
=\partial_{\theta_j}(\bfa_1^\top\mathbf K\bfa_2).
\]
Observe that  
\begin{align*}
\dot{\mathbf K}_{pq}
=\int\cos\{\omega^\top(\bfx_p-\bfx_q)\}\dot s_{\bftheta}(\omega)\,d\omega,
\qquad |\dot s_{\bftheta}|\leq Cs_{\bftheta}.
\end{align*}
Direct calculation yields that 
\begin{equation} \label{eqn: bilinear-derivative}
\begin{split}
|\bfa_1^\top\dot{\mathbf K}\bfa_2|
&=\left| \int
\left(\sum_i a_{1,i}e^{\mathrm i\omega^\top \bfx_i}\right)
\overline{\left(\sum_j a_{2,j}e^{\mathrm i\omega^\top \bfx_j}\right)}
\dot s_{\bftheta}(\omega)\,d\omega\right|\\
&\leq\int\left|\sum_i a_{1,i}e^{\mathrm i\omega^\top \bfx_i}\right|
\left|\sum_j a_{2,j}e^{\mathrm i\omega^\top \bfx_j}\right|\,|\dot s_{\bftheta}(\omega)|\,d\omega\\
&\leq C\int\left|\sum_i a_{1,i}e^{\mathrm i\omega^\top \bfx_i}\right|
\left|\sum_j a_{2,j}e^{\mathrm i\omega^\top \bfx_j}\right|s_{\bftheta}(\omega)\,d\omega\\
&\leq C\left(\int\left|\sum_i a_{1,i}e^{\mathrm i\omega^\top \bfx_i}\right|^2s_{\bftheta}\,d\omega\right)^{1/2}
\left(\int\left|\sum_j a_{2,j}e^{\mathrm i\omega^\top \bfx_j}\right|^2s_{\bftheta}\,d\omega\right)^{1/2}\\
&=C(\bfa_1^\top\mathbf K\bfa_1)^{1/2}(\bfa_2^\top\mathbf K\bfa_2)^{1/2},
\end{split}
\end{equation}
where the last inequality is due to Cauchy-Schwarz inequality. 
For the last equality, direct calculation gives 
\[
\int\left|\sum_i a_{1,i}e^{\mathrm i\omega^\top \bfx_i}\right|^2s_{\bftheta}(\omega)\,d\omega
=\sum_{i,j}a_{1,i}a_{1,j}\int\cos\{\omega^\top(\bfx_i-\bfx_j)\}s_{\bftheta}(\omega)\,d\omega
=\bfa_1^\top\mathbf K\bfa_1,
\]
and likewise for $\bfa_2$. The finite sums are bounded and $s_{\bftheta}$ is integrable, so all integrals above are absolutely convergent. 

Fix $\bfx$ and $j\in\{0,\ldots,d\}$. Let $U$ contain the $k\geq1$ knots on the active internal-node path. Set 
\begin{align*}
\mathbf K=C_{\bftheta}(U,U)=\mathbf L\mathbf L^\top,\qquad
\bfc_{\bfx}=C_{\bftheta}(U,\bfx),\qquad b_U(\bfx)=\mathbf L^{-1}\bfc_{\bfx}.
\end{align*}
The knots are distinct by Assumption~3. Substituting $\bfa_1=\mathbf L^{-\top}u$, $\bfa_2=\mathbf L^{-\top}w$ in \eqref{eqn: bilinear-derivative} gives
\begin{align*}
|u^\top\mathbf L^{-1}\dot{\mathbf K}\mathbf L^{-\top}w|
\leq C\|u\|_2\|w\|_2,\qquad
\|\mathbf L^{-1}\dot{\mathbf K}\mathbf L^{-\top}\|_{\mathrm{op}}\leq C.
\end{align*}
Here $\|\cdot\|_{\mathrm{op}}$ denotes  the operator norm. 

Write $\bfH=\bfL^{-1} \dot{\bfL}$. Then it can be verified that 
\begin{align*}
\mathbf H+\mathbf H^\top&=\mathbf L^{-1}\dot{\mathbf K}\mathbf L^{-\top},\\
\|\mathbf H\|_{\mathrm{op}}&\leq\|\mathbf H\|_{\mathrm F}
\leq\|\mathbf H+\mathbf H^\top\|_{\mathrm F}\\
&\leq\sqrt{k}\,\|\mathbf L^{-1}\dot{\mathbf K}\mathbf L^{-\top}\|_{\mathrm{op}}
\leq C\sqrt{k}.
\end{align*}
Here $\|\cdot\|_{\mathrm F}$ is the Frobenius norm with
  $\|\mathbf H+\mathbf H^\top\|_{\mathrm F}^2=2\|\mathbf H\|_{\mathrm F}^2+2\sum_a H_{aa}^2$. 

Applying \eqref{eqn: bilinear-derivative} to the covariance matrix at $(U,\bfx)$, with vectors $a_1=\binom{\mathbf L^{-\top}u}{0}$ and $a_2=\binom{0_k}{1}$ yields
\begin{align*}
|u^\top\mathbf L^{-1}\dot \bfc_{\bfx}|
\leq C\{u^\top\mathbf L^{-1}\mathbf K\mathbf L^{-\top}u\}^{1/2}\sigma
=C\sigma\|u\|_2,\qquad
\|\mathbf L^{-1}\dot \bfc_{\bfx}\|_2\leq C\sigma,
\end{align*}
where the last inequality follows by taking the supremum over unit $u$. 

Differentiating $\mathbf Lb_U=\bfc_{\bfx}$ and multiplying by $\mathbf L^{-1}$ yields
\begin{align*}
\dot b_U&=\mathbf L^{-1}\dot \bfc_{\bfx}-\mathbf H b_U,\\
\|\dot b_U(\bfx)\|_2
&\leq\|\mathbf L^{-1}\dot \bfc_{\bfx}\|_2+\|\mathbf H\|_{\mathrm{op}}\|b_U(\bfx)\|_2
\leq C\sigma+C\sqrt{k}\,\sigma
\leq C(1+\sqrt{k})\sigma_+.
\end{align*}

These calculations also hold when $\bfx\in U$: the augmented covariance need only be nonnegative definite. Indeed, for an active internal node $\mathcal A$, let $V$ list its strict ancestor knots. Then  
\begin{align*}
[b_U(\bfx)]_{\mathcal R_{\mathcal A}}
&=\mathbf L_{\mathcal A}^{-1}\bigl\{C_{\bftheta}(\mathcal R_{\mathcal A},\bfx)
-C_{0,\bftheta}(\mathcal R_{\mathcal A},V)C_{\bftheta}(V,V)^{-1}C_{\bftheta}(V,\bfx)\bigr\}\\
&=\mathbf L_{\mathcal A}^{-1}C_{\mathcal A,\bftheta}(\mathcal R_{\mathcal A},\bfx)
=u_{\mathcal A,\bftheta}(\bfx),
\end{align*}
where the subtracted term is zero for $V=\varnothing$. Off-path coordinates vanish, and leaf scales and partition indicators do not depend on $\bftheta$. Hence 
\begin{align*}
\|\partial_{\theta_j}b_{\mathcal T,\bftheta}(\bfx)\|_2^2
=\sum_{\substack{\mathcal A\in\mathcal I(\mathcal T)\\\bfx\in\mathcal A}}
\|\partial_{\theta_j}u_{\mathcal A,\bftheta}(\bfx)\|_2^2
=\|\dot b_U(\bfx)\|_2^2.
\end{align*}
Since $k\leq r(K-1)$,
\[
\|\partial_{\theta_j}b_{\mathcal T,\bftheta}(\bfx)\|_2
\leq C\sigma_+(1+\sqrt{r(K-1)})
\leq C\sigma_+(1+\sqrt r)\sqrt K
\leq C\sqrt K.
\]
For $k=0$ the derivative is zero. The bound is uniform in $\bfx,\bftheta$ and $j$. Along $\bftheta_t=(1-t)\bftheta+t\bftheta'$, the chain rule and fundamental theorem of calculus give
\[
\begin{aligned}
b_{\mathcal T,\bftheta'}(\bfx)-b_{\mathcal T,\bftheta}(\bfx)
&=\int_0^1\sum_{j=0}^d\partial_{\theta_j}b_{\mathcal T,\bftheta_t}(\bfx)(\theta_j'-\theta_j)\,dt,\\
\sup_{\bfx}\|b_{\mathcal T,\bftheta'}(\bfx)-b_{\mathcal T,\bftheta}(\bfx)\|_2
&\leq C\sqrt K\,\|\bftheta'-\bftheta\|_1
\leq C\sqrt{(d+1)K}\,\|\bftheta'-\bftheta\|_2.
\end{aligned}
\]
Absorbing $\sqrt{d+1}$ into $C$ yields the desired results. 

\end{proof} 

\begin{remark}
Propositions~\ref{prop: residual process bound} and~\ref{prop: covariance continuity} make $(\bftheta,z)\mapsto g_{\mathcal T,\bftheta,z}$ continuous in the uniform norm for each tree. Proposition~\ref{prop: tree size} gives the total prior mass one for finite trees that form a countable set by Assumption 2. Thus, the map $(\mathcal T,\bftheta,z)\mapsto g_{\mathcal T,\bftheta,z}$ is Borel for every $\bfx$. 
\end{remark} 

%%%%%%%%%%%%%%%%%%%%%%
\subsection{Verification of Condition C1}

\begin{prop}\label{prop: C1}
For every $\delta>0$,
\[
\mathcal B_\delta=\{(g,v):\|g-g_0\|_\infty<\delta,
\ |v-v_0|<\delta\},\qquad
\Pi_{g,v}(\mathcal B_\delta)\geq m_\delta>0,
\]
with $m_\delta$ independent of $n$.
\end{prop}
\begin{proof}
Fix $\delta>0$. Since $g_0\in C(\scrD)$  and $\scrD$ is compact under Assumption 1,
\begin{align*}
\omega_{g_0}(t):=\sup_{\substack{\bfx,\bfx'\in\scrD\\\|\bfx-\bfx'\|_2\leq t}}
|g_0(\bfx)-g_0(\bfx')|\longrightarrow0
\end{align*}
as $t \to 0^+$. 

Choose $m\geq1$ with $\omega_{g_0}(\sqrt d\,2^{-m})<\delta/3$. Let $\mathcal T_m$ split every cell at its midpoint in coordinates $1,\ldots,d$ for $m$ cycles. Its depth is $H=md$, and its $2^H$ leaves have side length $2^{-m}$. Choose $\bfx_{\mathcal A}\in\mathcal A$ in each leaf. Write $z_{\mathcal I}$ for the concatenated internal coefficient blocks and $q_{\mathcal I}=\sum_{\mathcal B\in\mathcal I}r_{\mathcal B}\geq1$ for their dimension.

For fixed  $\bftheta$, fixed $z$, and $\bfx$ in its assigned leaf $\mathcal A$, it follows from Proposition~\ref{prop: residual process bound} that 
\begin{align*}
g_{\mathcal T_m,\bftheta,z}(\bfx)
&=\sum_{\mathcal B\in\mathcal I} \mathds1_{\mathcal B}(\bfx)u_{\mathcal B,\bftheta}(\bfx)^\top z_{\mathcal B} + \mu_0 + \kappa_{\mathcal A}\sqrt{1+\kappa_0^{-1}} z_{\mathcal A},\\ 
\sum_{\mathcal B\in\mathcal I(\mathcal T_m)} \mathds1_{\mathcal B}(\bfx)\|u_{\mathcal B,\bftheta}(\bfx)\|_2^2
&=C_{0,\bftheta}(\bfx,\bfx)-C_{\mathcal A,\bftheta}(\bfx,\bfx) \leq\sigma_+^2.
\end{align*}

By Cauchy-Schwarz inequality, 
\begin{align*}
\left|\sum_{\mathcal B\in\mathcal I } \mathds1_{\mathcal B}(\bfx)u_{\mathcal B,\bftheta}(\bfx)^\top z_{\mathcal B}\right|^2
&\leq\left(\sum_{\mathcal B\in\mathcal I} \mathbf1_{\mathcal B}(\bfx)\|u_{\mathcal B,\bftheta}(\bfx)\|_2^2\right) \|z_{\mathcal I}\|_2^2 \leq\sigma_+^2\|z_{\mathcal I}\|_2^2.
\end{align*}

Set 
\begin{align*}
E_z:=\left\{z:\|z_{\mathcal I}\|_2<\frac{\delta}{3\sigma_+},\,
\max_{\mathcal A\in\mathcal L} 
|\mu_0 + \kappa_{\mathcal A}\sqrt{1+\kappa_0^{-1}} z_{\mathcal A}-g_0(\bfx_{\mathcal A})|<\frac{\delta}{3}\right\}. 
\end{align*}

Since $\|\bfx-\bfx_{\mathcal A}\|_2\leq\sqrt d\,2^{-m}$, direct calculation yields
\begin{align*}
\|g_{\mathcal T_m,\bftheta,z}-g_0\|_\infty
&\leq\sigma_+\|z_{\mathcal I}\|_2 
+\max_{\mathcal A\in\mathcal L} 
|\mu_0 + \kappa_{\mathcal A}\sqrt{1+\kappa_0^{-1}} z_{\mathcal A}-g_0(\bfx_{\mathcal A})|+\omega_{g_0}(\sqrt d\,2^{-m}) \\
& < \delta/3 + \delta/3 + \delta/3=\delta
\end{align*}
for $z \in E_z$. 

As $E_z$ is a non-empty open set, $\Pi_z(E_z | \calT_m)>0$  for any $\delta>0$. Next we show positive prior probability for tree partition. 

Write $s_h=\rho_0(1+h)^{-\beta_0}$, $j_h=1+(h\bmod d)$, $\lambda_* =\min_j\lambda_j$, and let $w_{\mathcal A}\geq w_*>0$ be the conditional midpoint probabilities in Assumption~2. There are $2^h$ internal nodes at depth $h<H$ and $2^H$ terminal nodes at depth $H$. Thus, 
\begin{align*}
\pi(\mathcal T_m)
&=\left[\prod_{h=0}^{H-1}\prod_{\substack{\mathcal A\in\mathcal T_m\\\operatorname{depth}(\mathcal A)=h}}
(s_h\lambda_{j_h}w_{\mathcal A})\right](1-s_H)^{2^H}\\
&\geq\{\rho_0(1+H)^{-\beta_0}\lambda_*w_*\}^{2^H-1}(1-\rho_0)^{2^H}>0.
\end{align*}

Finally, it follows from Assumption 4 that 
\begin{align*}
\Pi_{g,v}(\mathcal B_\delta)
&\geq\pi(\mathcal T_m)\int \Pi_z(E_z\mid\mathcal T_m)\Pi_v(|v-v_0|<\delta)\,\Pi_{\bftheta}(\mathrm{d}\bftheta)\\
&=\pi(\mathcal T_m)\Pi_z(E_z\mid\mathcal T_m)\Pi_v(|v-v_0|<\delta) =:m_\delta>0,
\end{align*} 
where $m_{\delta}$ is independent of $n$. Thus the conclusion holds. 
\end{proof}

For the likelihood-moment calculation, define the log ratios at possible observations by
\begin{align}\label{eqn: log-likelihood-ratios}
\begin{split}
L_{n,i}(g,v;o_i)&=\log\frac{\varphi_{v_0}(y_i-g_0(\bfx_i))}{\varphi_v(y_i-g(\bfx_i))},\\
L_n(g,v;o)&=\sum_{i=1}^nL_{n,i}(g,v;o_i)
=\log\frac{p_{g_0,v_0}^{(n)}(o)}{p_{g,v}^{(n)}(o)}.
\end{split}
\end{align}
Here $\bfx_i$ is the fixed design point or the input component of $o_i$, as in Section~4. An omitted observation argument means evaluation at $O_i$ or $\mathcal O_n$, respectively. Thus the integral in C1 is
\begin{align}\label{pub:evidence-definition}
\int e^{-L_n(g,v)}\,\Pi_{g,v}(dg,dv)
=\frac{\int p_{g,v}^{(n)}(\mathcal O_n)\,\Pi_{g,v}(dg,dv)}
{p_{g_0,v_0}^{(n)}(\mathcal O_n)}.
\end{align}
This rewrites the denominator of Eq.~(5) after division by the true density.

\begin{prop}\label{prop: evidence}
Condition C1 holds under either design. More precisely, there is $C_*>0$ such that, for $0<\delta\leq\min(1,v_0/2)$ and $\epsilon=C_*\delta$ with $n\epsilon^2\geq1$,
\begin{equation}\label{eqn: evidence-probability}
\mathbb P_0\left(\int e^{-L_n(g,v)}\,\Pi_{g,v}(dg,dv)
\leq m_\delta e^{-2n\epsilon^2}\right)
\leq\frac{1}{n^2\epsilon^4}.
\end{equation}
\end{prop}

\begin{proof}
Fix a deterministic $(g,v)\in\mathcal B_\delta$ and put $\Delta=g-g_0$; suppress this fixed pair in $L_{n,i}$ and $L_n$ within the proof for notational simplicity. By \citet[Lemma 2.7(ii)]{Ghosal2017}, under both designs,
\begin{equation}\label{eqn: gaussian-moments}
\mathbb E_0L_n
=\frac n2\left\{\log\frac v{v_0}+\frac{v_0}v-1+
\frac{\|\Delta\|_{L_2(\mathbb P_{\bfX,n})}^2}{v}\right\}
\leq C_{\mathrm{KL}}n\delta^2,
\end{equation}
where $C_{\mathrm{KL}}=3/v_0^2+1/(2v_-)$. Indeed, $\log t+t^{-1}-1$ and its first derivative vanish at $t=1$, and its second derivative $(2-t)/t^3$ lies in $(0,12]$ on $[1/2,3/2]$; Taylor's theorem applies with $t=v/v_0$.

Write $\Delta_i=\Delta(\bfx_i)$ in fixed design and $\Delta_i=\Delta(\bfX_i)$ in random design. Assumption~1 makes $\Delta_i$ independent of $\varepsilon_i$; in fixed design it is deterministic. 

Observe that 
\begin{align*}
W_i:=L_{n,i}-\mathbb E_0L_{n,i}
=\frac{v_0-v}{2v}(\varepsilon_i^2-1)
-\frac{\sqrt{v_0}\Delta_i}{v}\varepsilon_i
+\frac{\Delta_i^2-\mathbb E_0\Delta_i^2}{2v}.
\end{align*}
The last term is zero in fixed design. For the first two terms, Gaussian integration yields
\begin{align*}
&\int\left|\frac{v_0-v}{2v}(e^2-1)-\frac{\sqrt{v_0}\Delta(\bfx)}v e\right|^4\varphi_1(e)\,de\\
&\qquad=\frac{\frac{15}{4}(v_0-v)^4+15(v_0-v)^2v_0\Delta(\bfx)^2+3v_0^2\Delta(\bfx)^4}{v^4}.
\end{align*}

This uses $\mathbb E(\varepsilon^2-1)^4=60$, $\mathbb E[(\varepsilon^2-1)^2\varepsilon^2]=10$, $\mathbb E\varepsilon^4=3$, and zero odd moments. Since $|\Delta_i^2-\mathbb E_0\Delta_i^2|\leq\delta^2$, the bound $|a+b|^4\leq8(|a|^4+|b|^4)$ gives
\begin{equation}\label{eqn: component-fourth}
\mathbb E_0W_i^4\leq C_4\delta^4,\qquad
C_4=\frac8{v_-^4}\left(\frac{61}{16}+15v_0+3v_0^2\right).
\end{equation}
Independence within each experiment and $\mathbb E_0W_i^2\leq(\mathbb E_0W_i^4)^{1/2}$ imply
\begin{equation}\label{eqn: sample-moments}
\mathbb E_0|L_n-\mathbb E_0L_n|^4
=\sum_i\mathbb E_0W_i^4+6\sum_{i<j}\mathbb E_0W_i^2\mathbb E_0W_j^2
\leq3C_4n^2\delta^4.
\end{equation}
Choose $C_*\geq\max\{1,\sqrt{C_{\mathrm{KL}}},(3C_4)^{1/4}\}$. Equations~\eqref{eqn: gaussian-moments}--\eqref{eqn: sample-moments} and Proposition~\ref{prop: C1} verify the KL, centered fourth-moment and prior-mass requirements of \citet[Lemma 8.21]{Ghosal2017}, with $k=4$, $D=1$ and $\epsilon=\epsilon_{\mathrm{ev}}$, proving \eqref{eqn: evidence-probability}; see also \citet[Lemma 10]{Ghosal2007}. For fixed $a>0$, choose $\delta$ with $2C_*^2\delta^2<a$. Then
\[
\sum_n(n^2\epsilon_{\mathrm{ev}}^4)^{-1}<\infty,\qquad
m_\delta e^{-2n\epsilon_{\mathrm{ev}}^2}\geq e^{-an}
\quad\text{for all sufficiently large }n.
\]
The Borel--Cantelli lemma therefore establishes Condition C1.
\end{proof}

\subsection{Verification of Condition C2}

\begin{prop}\label{prop: C2}
Denote $\alpha_{\max}: =\min\{1,(4\tilde r\log2)^{-1}\}$. For $0<\alpha\leq\alpha_{\max}$, let
\begin{align*}
\mathcal S_{n,\alpha}
=\bigcup_{\mathcal T: K \leq\lfloor\alpha n\rfloor}
\left\{(g_{\mathcal T,\bftheta,z},v):
\|z\|_2\leq\sqrt n,\ \max_{1\leq j\leq d}|\theta_j|\leq n\right\},
\end{align*}
where $\bftheta$ and $v$ range over the domains in Assumption~4, with any fixed covariance coordinates held at their prescribed values. For sufficiently large $n$,
\begin{align*}
\Pi_{g,v}(\mathcal S_{n,\alpha}^c)\leq C_\alpha e^{-c_\alpha n},
\qquad c_\alpha=\min\{c_{\mathcal T}\alpha,1/8,c_\ell\}>0.
\end{align*}
\end{prop}
\begin{proof}
The finite cut set gives finitely many labeled trees with $K\leq\lfloor\alpha n\rfloor$. For each such tree, it follows from Propositions~\ref{prop: residual process bound} and~\ref{prop: covariance continuity} that
\begin{align*}
\|g_{\mathcal T,\bftheta,z}-g_{\mathcal T,\bftheta',z'}\|_\infty
\leq B\|z-z'\|_2+C\sqrt K\,\|z'\|_2\|\bftheta-\bftheta'\|_2.
\end{align*}
Consequently $(\bftheta,z,v)\mapsto(g_{\mathcal T,\bftheta,z},v)$ is continuous. 

For fixed $\mathcal T$ with $K\leq k_n:=\lfloor\alpha n\rfloor$, Proposition~\ref{prop: residual process bound} gives $q \leq\tilde r\alpha n$. The Gaussian exponential moment and Markov's inequality give
\begin{align*}
\int e^{\|z\|_2^2/4}\,\Pi_z(dz\mid\mathcal T)=2^{q/2},\qquad
\Pi_z(\{z:\|z\|_2>\sqrt n\}\mid\mathcal T)
\leq e^{-n/4+(q/2)\log2}\leq e^{-n/8}.
\end{align*}
Every coordinate draw satisfying the sieve restrictions maps into $\mathcal S_{n,\alpha}$. By Proposition~\ref{prop: tree size} and Assumption 4, direct calculation yields 
\begin{align*}
\Pi_{g,v}(\mathcal S_{n,\alpha}^c)
&\leq\sum_{\mathcal T:K>k_n}\pi(\mathcal T)
+\sum_{\mathcal T:K\leq k_n}\pi(\mathcal T)
\Pi_z(\{z:\|z\|_2>\sqrt n\}\mid\mathcal T)\\
&\quad+\Pi_{\bftheta}(\max_{1\leq j\leq d}|\theta_j|>n)\\
&\leq C_{\mathcal T}e^{c_{\mathcal T}}e^{-c_{\mathcal T}\alpha n}
+e^{-n/8}+C_\ell e^{-c_\ell n}\leq C_\alpha e^{-c_\alpha n},
\end{align*}
with $C_\alpha=C_{\mathcal T}e^{c_{\mathcal T}}+1+C_\ell$, as desired. 
\end{proof}

\subsection{Verification of Condition C3}

\begin{prop}\label{prop: C3}
For every fixed $0<\epsilon\leq1$, there are $0<\alpha_\epsilon\leq\alpha_{\max}$ and deterministic Borel tests $\psi_n$, as in C3, such that, for all sufficiently large $n$,
\begin{equation}\label{eqn: test-errors}
\mathbb E_0\psi_n(\mathcal O_n)\leq e^{-n\epsilon^2/16},\qquad
\sup_{\substack{(g,v)\in\mathcal S_{n,\alpha_\epsilon}\\H_{\mathbb P_{\bfX,n}}((g,v),(g_0,v_0))>\epsilon}}
\mathbb E_{g,v}^{(n)}(1-\psi_n)\leq e^{-n\epsilon^2/8}.
\end{equation}
An empty alternative set uses the zero test.
\end{prop}
\begin{proof}
Write $\mathcal N(u,E,\cdot)$ for the minimum number of open radius-$u$ balls covering $E$. For any input probability measure $\mathbb P$, \citet[Lemma 2.7(i)]{Ghosal2017} gives
\begin{equation}\label{prop: hellinger-upper}
H_{\mathbb P}((g,v),(g',v'))
\leq\frac{\|g-g'\|_{L_2(\mathbb P)}}{2\sqrt{v_-}}+
\frac{|v-v'|}{2v_-}
\leq C\{\|g-g'\|_{L_2(\mathbb P)}+|v-v'|\}.
\end{equation}
Choose $C\geq1$ also valid in Eq.~\eqref{eqn: cov lipschitz}, and put $u=\epsilon/(216C)$.

Next, we compute the metric entropy. For each tree with $K\leq\alpha n$, a covariance grid $\mathcal Q_{\mathcal T}$ of coordinate mesh $u/\{C\sqrt{(d+1)nK}\}$ covers the covariance box. Proposition~\ref{prop: covariance continuity} implies
\begin{align*}
|\mathcal Q_{\mathcal T}|\leq(1+C_\epsilon n^2)^{d+1},\qquad
\inf_{\bftheta'\in\mathcal Q_{\mathcal T}}\sup_{\|z\|_2\leq\sqrt n}
\|g_{\mathcal T,\bftheta,z}-g_{\mathcal T,\bftheta',z}\|_\infty\leq u.
\end{align*} 
Write $b=b_{\mathcal T,\bftheta'}$. The coefficient image $E=\{b^\top z:\|z\|_2\leq\sqrt n\}$ is an ellipsoid in $L_2(\mathbb P_{\bfX,n})$ with squared semiaxes equal to the eigenvalues of $n\int bb^\top\,\mathrm{d}\mathbb P_{\bfX,n}$, including zeros. Let $a_j$ denote the length of the $j$th semiaxis of the ellipsoid $E$, that is $a_j=\sqrt{n \lambda_j}$ where $\lambda_j$ is the $j$th eigenvalue of $\int b(\bfx) b(\bfx)^\top \mathrm{d} \mathbb{P}_{\bfX,n}(\bfx)$. 
Proposition~\ref{prop: residual process bound} gives $\sum_{j=1}^q a_j^2 = n \int \| b(\bfx)\|_2^2  \mathrm{d} \mathbb{P}_{\bfX,n}(\bfx) \leq nB^2$.

Since covering numbers are invariant under translation, we have 
\begin{align*}
    \mathcal N(u,\mu_0 + E,L_2(\mathbb P_{\bfX,n})) = \mathcal N(u,E,L_2(\mathbb P_{\bfX,n})).
\end{align*}
Use the packing-volume argument of \citet[Appendix C.1 and Proposition C.2]{Ghosal2017}: in Euclidean coordinates, $D=\operatorname{diag}(a_j+u/2)$ satisfies $E+(u/2)B_2^q\subset2DB_2^q$. Disjoint open radius-$u/2$ balls around a maximal subset with pairwise distances at least $u$ therefore give
\begin{equation}\label{eqn: coefficient-cover}
\begin{split}
\log\mathcal N(u,E,L_2(\mathbb P_{\bfX,n}))
&\leq\log\frac{\operatorname{vol}(2DB_2^q)}{\operatorname{vol}((u/2)B_2^q)}
=\sum_{j=1}^q\log\{2(1+2a_j/u)\}\\
&\leq q\log\left\{2\left(1+\frac{2B\sqrt n}{u\sqrt q}\right)\right\}.
\end{split}
\end{equation}
The inclusion follows from $\|D^{-1}e\|_2\leq1$ for $e\in E$ and $\|D^{-1}w\|_2\leq1$ for $\|w\|_2\leq u/2$; the last inequality uses concavity of $\log$ and $\sum_j a_j\leq\sqrt{q\sum_j a_j^2}$. Here $B_2^q$ is the Euclidean unit ball. A variance grid of mesh $u$ completes a Hellinger cover of radius $3Cu=\epsilon/72$.

Let $\mathcal C_{\mathrm{cut}}$ be the finite relative cut set in Assumption~2. There are at most $(4d|\mathcal C_{\mathrm{cut}}|)^{K-1}$ labeled trees with $K$ leaves. The function $q\log\{2(1+c/\sqrt q)\}$ increases in $q>0$: its derivative is $\log\{2(1+t)\}-t/[2(1+t)]>0$, where $t=c/\sqrt q$. For $n\geq2$ with $\lfloor\alpha n\rfloor\geq1$, sum over $K\leq\lfloor\alpha n\rfloor$ and use $q\leq\tilde r\alpha n$ to obtain
\begin{equation}\label{eqn: entropy}
\begin{split}
\log\mathcal N(\epsilon/72,\mathcal S_{n,\alpha},H_{\mathbb P_{\bfX,n}})
&\leq C\alpha n+\log n+(d+1)\log(1+C_\epsilon n^2)+\log C_\epsilon\\
&\quad+\tilde r\alpha n\log\{2(1+C_\epsilon/\sqrt\alpha)\}\\
&\leq C_\epsilon\{n\alpha\log(e/\alpha)+\log n\}.
\end{split}
\end{equation}
All bounds are uniform in the input probability measure. Fixed covariance coordinates require only one grid point.

Then we construct the tests. Since $\alpha\log(e/\alpha)\to0$, choose $\alpha_\epsilon\leq\alpha_{\max}$ so that the cover size $N_n$ satisfies $\log N_n\leq n\epsilon^2/16$ for sufficiently large $n$. In each ball intersecting the alternative set, choose one alternative as center; the resulting radius is $\epsilon/36<\epsilon/5$. 

In fixed design, the $i$th factors of the sampling densities in Section~4 have squared Hellinger distance $H_{\delta_{\bfx_i}}^2((g,v),(g',v'))$; averaging over $i$ gives $H_{\mathbb P_{\bfX,n}}^2((g,v),(g',v'))$. In random design, each factor has squared distance $H_{\mathbb P_{\bfX}}^2((g,v),(g',v'))$, giving the same identity. Thus by \citet[Proposition D.9]{Ghosal2017}, with $c=1$, there exists a local Borel test $\psi_{n,j}$ with both errors at most $e^{-n\epsilon^2/8}$ at each center and throughout its ball. Their maximum $\psi_n=\max_j\psi_{n,j}$ satisfies
\begin{align*}
\mathbb E_0\psi_n(\mathcal O_n)&\leq\sum_{j=1}^{N_n}\mathbb E_0\psi_{n,j}(\mathcal O_n)
\leq N_ne^{-n\epsilon^2/8}\leq e^{-n\epsilon^2/16},\\
\mathbb E_{g,v}^{(n)}(1-\psi_n)
&\leq\mathbb E_{g,v}^{(n)}(1-\psi_{n,j})\leq e^{-n\epsilon^2/8}
\end{align*}
whenever $(g,v)$ belongs to the $j$th ball. 
\end{proof}

\subsection{Proof of the criterion in Section 4.2} 

\begin{proof}[Proof of the criterion in Section~4.2] 

Fix $\epsilon>0$, write $\mathcal H_n=\{(g,v):H_{\mathbb P_{\bfX,n}}((g,v),(g_0,v_0))>\epsilon\}$ and $\Psi_n=\psi_n(\mathcal O_n)$, and put
\begin{align*}
N_{1,n}=(1-\Psi_n)\int_{\mathcal H_n\cap\mathcal S_n}e^{-L_n(g,v)}\,\Pi_{g,v}(dg,dv),\qquad
N_{2,n}=\int_{\mathcal S_n^c}e^{-L_n(g,v)}\,\Pi_{g,v}(dg,dv).
\end{align*}
The density-ratio identity gives, for every fixed $(g,v)$,
\begin{align*}
\mathbb E_0\{(1-\Psi_n)e^{-L_n(g,v)}\}
=\int(1-\psi_n(o))\frac{p_{g,v}^{(n)}(o)}{p_{g_0,v_0}^{(n)}(o)}\,\mathbb P_0^{(n)}(do)
=\mathbb E_{g,v}^{(n)}(1-\psi_n).
\end{align*}
Here $\mathbb P_0^{(n)}$ is the law of $\mathcal O_n$ under $\mathbb P_0$; the same identity without the test gives $\mathbb E_0e^{-L_n(g,v)}=1$. 

It follows from C2--C3 and Markov's inequality that with $c=c_\epsilon$,
\begin{align*}
\mathbb E_0N_{j,n}\leq C_\epsilon e^{-cn},\qquad
\mathbb P_0(N_{j,n}>e^{-cn/2})\leq C_\epsilon e^{-cn/2}\ (j=1,2),\qquad
\sum_n\mathbb P_0(\Psi_n=1)<\infty.
\end{align*}
For $0<a<c/2$, it follows from Borel--Cantelli and C1 that eventually almost surely,
\begin{align*}
\Pi_{g,v}(\mathcal H_n\mid\mathcal O_n)
\leq\Psi_n+\frac{N_{1,n}+N_{2,n}}{\int e^{-L_n(g,v)}\,\Pi_{g,v}(dg,dv)}
\leq2e^{-n(c/2-a)}\longrightarrow0.
\end{align*}
This proves the criterion for general experiments, without an iid assumption across observations or independence across experiments.
\end{proof}

\subsection{Proof of Theorem 4}
\begin{proof}
Apply the preceding criterion to $\Pi_{g,v}$ and the sampling densities in Section~4. Note that Propositions~\ref{prop: C1}--\ref{prop: evidence} verifies C1, and Propositions~\ref{prop: C2}--\ref{prop: C3} verify C2--C3  with 
\begin{align*}
\mathcal S_n=\mathcal S_{n,\alpha_\epsilon},\qquad
c_\epsilon=\min\{c_{\alpha_\epsilon},\epsilon^2/16\},\qquad
C_\epsilon=\max\{C_{\alpha_\epsilon},1\}.
\end{align*}
Thus, the conclusion holds. 
\end{proof}

\subsection{Proof of Corollary 2}
\begin{proof}[Proof of Corollary~2]
Fix an input probability measure $\mathbb P$ and write $\Delta=g-g_0$. The Gaussian affinity formula \citep[proof of Lemma 2.7]{Ghosal2017} is
\begin{equation}\label{eqn: normal-affinity}
H_{\mathbb P}^2((g,v),(g_0,v_0))
=2-2\left\{\frac{2\sqrt{vv_0}}{v+v_0}\right\}^{1/2}
\int e^{-\Delta(\bfx)^2/[4(v+v_0)]}\,\mathbb P(d\bfx).
\end{equation}
Since $2(1-\sqrt t)\geq1-t$ for $0\leq t\leq1$ and $v,v_0\leq v_+$,
\begin{equation}\label{eqn: metric-lower-bounds}
H_{\mathbb P}^2\geq\frac{(\sqrt v-\sqrt{v_0})^2}{v+v_0}
\geq\frac{(v-v_0)^2}{8v_+^2},\qquad
H_{\mathbb P}^2\geq2(1-e^{-s^2/(8v_+)})\mathbb P(|\Delta|>s),\quad s>0.
\end{equation}
Put $\mathcal W_{\epsilon,\mathbb P}=\{d_{\mathbb P}(g,g_0)<\epsilon,\ |v/v_0-1|<\epsilon\}$. The definition of $d_{\mathbb P}$ gives
$d_{\mathbb P}(g,g_0)\geq\epsilon\Rightarrow\mathbb P(|\Delta|>\epsilon/2)\geq\epsilon/2$.
Consequently, uniformly in $\mathbb P$,
\begin{align*}
\mathcal W_{\epsilon,\mathbb P}^c\subset\{H_{\mathbb P}>\delta\},\qquad
\delta^2=\tfrac12\min\left\{\frac{v_0^2\epsilon^2}{8v_+^2},
\epsilon(1-e^{-\epsilon^2/(32v_+)})\right\}>0.
\end{align*}
Applying Theorem~4 with $\mathbb P=\mathbb P_{\bfX,n}$ yields the desired results. 
\end{proof}

\subsection{Proof of Corollary 3}

\begin{proof}[Proof of Corollary~3]
Let $\mathcal E_M=\{(g,v):\|g\|_\infty\leq M\}$ and $c_M=\Pi_{g,v}(\mathcal E_M)$. For $0<\delta<\min\{M-\|g_0\|_\infty,1,v_0/2\}$, by Proposition~\ref{prop: C1}, 
\begin{align*}
\mathcal B_\delta\subset\mathcal E_M,\qquad c_M\geq m_\delta>0,\qquad
\Pi_{g,v}^M(\mathcal B_\delta)=\Pi_{g,v}(\mathcal B_\delta)/c_M\geq m_\delta.
\end{align*}
Thus the conditioned prior $\Pi_{g,v}^M$ in Corollary~3 is proper. The moment bounds in Proposition~\ref{prop: evidence} do not depend on the prior, so its proof gives C1 for $\Pi_{g,v}^M$. Also
\begin{align*}
\Pi_{g,v}^M(\mathcal S_{n,\alpha}^c)\leq c_M^{-1}\Pi_{g,v}(\mathcal S_{n,\alpha}^c)
\leq c_M^{-1}C_\alpha e^{-c_\alpha n}.
\end{align*}
The tests in Proposition~\ref{prop: C3} are unchanged. The criterion and the deterministic inclusion in the preceding proof therefore give
$\Pi_{g,v}^M(\mathcal W_{\eta,\mathbb P_{\bfX,n}}^c\mid\mathcal O_n)\to0$ almost surely for every $\eta>0$.

On $\mathcal E_M$, $|g-g_0|\leq2M$. For $s>0$ with $\mathbb P(|g-g_0|>s)<s$, splitting the integral at $s$ gives \citep[cf.\ Remark 8]{Choi2007}
\begin{align*}
\|g-g_0\|_{L_1(\mathbb P)}\leq s+2M\mathbb P(|g-g_0|>s)<(1+2M)s.
\end{align*}
Taking the infimum over such $s$ yields $\|g-g_0\|_{L_1(\mathbb P)}\leq(1+2M)d_{\mathbb P}(g,g_0)$. Hence
\begin{align*}
\{\|g-g_0\|_{L_1(\mathbb P)}>\epsilon\}\cup\{|v/v_0-1|\geq\epsilon\}
\subset\mathcal W_{\epsilon/(1+2M),\mathbb P}^c.
\end{align*}
Use $\mathbb P=\mathbb P_{\bfX,n}$. Since $p_{\bfX}$ is bounded away from zero, there exists some $p_{\bfX,*}>0$  such that $p_{\bfX} \geq p_{\bfX,*}>0$. Thus, 
$\|g-g_0\|_{L_1(d\bfx)}\leq p_{\bfX,*}^{-1}\|g-g_0\|_{L_1(\mathbb P_{\bfX})}$ by applying the preceding conclusion at radius $\min(p_{\bfX,*},1)\epsilon$.
\end{proof}

\end{document}